\documentclass[12pt,openany]{article}
\usepackage{ae,aecompl}

\usepackage[T1]{fontenc}
\usepackage[utf8]{inputenc}
\usepackage{geometry}
\usepackage{amsmath}
\usepackage{amsthm}
\usepackage{amssymb}
\usepackage{setspace}
\usepackage[authoryear]{natbib}
\usepackage{stackrel}
\usepackage{graphicx}
\usepackage{caption}
\usepackage{subcaption}
\usepackage{hyperref}

\newcounter{innersubfigure}[subfigure]
\renewcommand{\theinnersubfigure}{\alph{innersubfigure}}

\makeatletter
\usepackage{setspace}
\usepackage{pdflscape}
\usepackage{afterpage}
\usepackage{datetime}
\usepackage{xcolor}

\DeclareMathOperator{\Var}{Var}

\DeclareMathOperator{\Cov}{Cov}

\DeclareMathOperator{\KL}{KL}

\DeclareMathOperator*{\argmin}{arg\,min}

\DeclareMathOperator*{\tr}{Trace}

\newtheorem{prop}{Proposition}

\newtheorem{assu}{Assumption}
\newtheorem{lemma}{Lemma}
\newtheorem{thm}{Theorem}
\newtheorem{cor}{Corollary}
\newtheorem{axiom}{Axiom}
\newtheorem{defn}{Definition}
\newcommand{\norm}[1]{\left\lVert#1\right\rVert}

\date{}

\title{Misspecification-Averse Estimation}
\author{ Isaiah Andrews, Ricky Li, and Yucheng Shang\thanks{\scriptsize This version: April 24, 2026.  First posted version: March 30, 2026. Andrews: MIT Department of Economics and NBER; iandrews@mit.edu.  Li, MIT Department of Economics; rickyli@mit.edu.  Shang, MIT Department of Economics; ycshang@mit.edu.  We thank Ashesh Rambachan, Brit Sharoni, and Tomasz Strzalecki for collaboration on an earlier project which derived related results, Drew Fudenberg and seminar participants at Harvard/MIT and the University of Michigan for very helpful feedback, and Claude Opus 4.6 and GPT 5.4 Pro for outstanding research assistance.}}

\makeatother

\begin{document}

\maketitle

\begin{abstract} 
We study optimal estimation when the likelihood may be misspecified.  Building on
tools from the theory of decision-making under uncertainty, we analyze a class of axiomatically grounded optimality criteria which nests several existing misspecification-robust objectives. Within this class, we introduce the constrained multiplier criterion, which allows for flexible misspecification attitudes. We prove a local
asymptotic minimax theorem for this criterion, extending a classical efficiency
bound to a limit experiment which incorporates moment-constrained misspecification concerns. We characterize asymptotically optimal estimators as Bayes decision rules 
under a flat prior and an exponentially tilted likelihood that incorporates the moment
constraints, and show that feasible plug-in analogs are asymptotically optimal.
\end{abstract}

\section{Introduction}

    Researchers in economics are often concerned that their models are wrong,
and about the consequences  for estimation and inference.  A large literature
considers the problem of estimation and inference under misspecification
with different optimality criteria, including worst-case mean squared
error \citep{Huber1964, BonhommeWeidner2022} and confidence interval length
\citep{ArmstrongKolesar2020}.  Other work considers the related
problem of optimal decision-making in settings with misspecification concerns
\citep{HansenSargent2001, HansenSargent2008,DuchiNamkoong2021DRO,GaoKleywegt2023DRO,
CerriaViolioetal24}.  These contributions adopt different objective
functions, and the choice among criteria has substantive consequences
for the resulting procedures.  A natural question is thus what
objective a researcher \emph{should} use when their model (and in particular, the model-implied mapping from parameters to data distributions) may be
wrong, and what different choices of criteria imply for the resulting ranking of decision rules.

The question of how to evaluate decisions in the presence of model
uncertainty is a central concern of the microeconomic theory of
choice under ambiguity.  That literature provides axiomatic
foundations for classes of preferences, linking axioms on preferences to 
 functional representations of the decision-maker's objective
\citep{GilboaSchmeidler1989, maccheroni2006ambiguity,
CerriaViolioetal24}.  
This paper applies these tools to the problem of choosing estimation
criteria when the model may be misspecified, connecting axiomatic
foundations for ambiguity-averse preferences to the design of
statistical procedures.

Our first contribution is to introduce \emph{constrained multiplier
preferences} as an optimality criterion for estimation under
misspecification.  Under this criterion, a researcher evaluates an
estimator by its worst-case expected loss over distributions that
satisfy constraints encoding what is known about potential
misspecification, penalized by a Kullback-Leibler divergence term that
reflects the difference from the baseline model.  We start
from a broad class of misspecification-averse preferences axiomatized
by \citet{CerriaViolioetal24}, which nests many different
misspecification-robust criteria.  We then discuss additional axioms
that select specific subclasses within this family: constraint
preferences \citep[as in][]{BonhommeWeidner2022}, which are
characterized by the certainty independence axiom of
\citet{GilboaSchmeidler1989}; multiplier preferences
\citep[as in][]{HansenSargent2001}, characterized by the sure thing
principle following \citet{strzalecki2011axiomatic}; and the
constrained multiplier class, which combines constraint and multiplier
features and whose axiomatic characterization is new.  Building on results from the
generalized empirical likelihood literature, we further show that when the constraints
can be expressed as moment conditions, the constrained 
multiplier preference has a computationally tractable dual representation.  The moment
constraints allow the researcher to express partial trust in the
model: for instance, one may trust the model's implications for certain moments of the data, but
not the full data distribution.

Our second contribution is to derive optimal estimators for this
criterion.  We prove a local asymptotic minimax theorem that extends
the classical lower bound of \citet{Hajek1972}
to the constrained multiplier objective: for convex loss, no sequence
of estimators can achieve worst-case risk below our bound.  Obtaining this bound
requires deriving a novel limit experiment that
incorporates both misspecification concerns and moment information.
We then characterize estimators that attain this bound.  They
take the form of Bayes rules under an exponentially tilted likelihood 
that incorporates the moment constraints.  Under squared error loss, the
structure simplifies: when no moment constraints are imposed, or when only the mean of a vector of sample average moments is constrained, the maximum likelihood estimator
is optimal; when higher moments of the sample average moment vector are constrained, the
optimal rule linearly adjusts the MLE using the sample average moments.\footnote{Contemporaneous and independent work by \cite{adusumilli2026} derives related asymptotic results for worst-case Bayes risk in the case with no constraints on misspecification, likewise showing that the rules which are asymptotically optimal under the model remain optimal under misspecification concern.  His results cover parametric and semiparametric estimation with general symmetric and quasi-convex loss functions, as well as treatment assignment problems.}  We further show that plug-in finite-sample
analogs of the optimal rules in the limit experiment are asymptotically optimal under regularity
conditions.  Our results imply that the common practice in economics of specifying a full parametric model but conducting estimation based solely on a subset of model-implied moments due to misspecification concern \citep[see e.g.][for examples]{andrews2025purpose} is suboptimal absent extreme misspecification aversion.

Optimal estimation requires the researcher to specify their (scalar) degree of misspecification aversion since, absent such concern, conventional optimality results apply.  
Since it may sometimes be difficult for researchers to specify this preference parameter, following \cite{ArmstrongKlineSun2025} we close by exploring the scope for adapting to the degree of misspecification concern.  Focusing on settings with a scalar parameter, a scalar moment condition, and squared error loss, we find that (as in \citealt{ArmstrongKlineSun2025}) two simple classes of decision rules attain nearly the same degree of adaptation as do much more flexible rules.  The resulting simple adaptive rules may be attractive to researchers interested in applying our results but uncertain how to specify the misspecification-aversion parameter.

Section \ref{sec: loss axioms} introduces the decision-theoretic
framework and the axiomatization of constrained multiplier
preferences.  Section \ref{sec: LAM} states the asymptotic risk
bound.  Section \ref{sec: optimal rules} characterizes optimal
estimators in the limit experiment and their feasible analogs, 
and Section \ref{sec: adaptive rules} studies adaptive rules.
Appendix \ref{sec: appendix axioms} provides further detail on
the axiomatic foundations of our approach, while the remaining proofs appear in
Appendix \ref{sec: proofs}.

\section{Preferences Over Losses}\label{sec: loss axioms}

Consider a researcher who will observe data $X$ from a sample space $\mathcal{X}$.  For an unknown parameter $\theta$ in a parameter space $\Theta$, this researcher must choose an action $a$ from a set of feasible actions $\mathcal{A}$.  The researcher's objective is specified by a loss function $l:\mathcal{A}\times\Theta\to\mathbb{R}$, where $l(a,\theta)$ describes the loss from taking action $a$ when the true parameter is $\theta$.

Since the researcher does not observe $\theta$ directly, they must choose an action based on the data.  A (randomized) decision rule $\delta:\mathcal{X}\to\Delta(\mathcal{A})$ maps data realizations to distributions over actions.  A decision rule $\delta$ thus induces an expected loss that depends on $(\theta,X)$,
\[L_\delta(\theta,X)=\int l(a,\theta)d\delta(a;X).\]
We will take the induced loss functions $L:\Theta\times\mathcal{X}\to\mathbb{R}$ as our starting point and rank decision rules based on our preference over their induced losses $L_\delta$.  Our approach is thus similar in spirit to \citet{Stoye2012}, who takes risk (expected loss integrating over the distribution of the data) functions as a primitive to study the choice of decision rules. We work with loss rather than risk since risk depends on the distribution of the data, and we are interested in settings where this distribution may differ from that assumed by the model.

\paragraph{Example: Average Treatment Effects}

Suppose the researcher observes a sample of $n$  observations $X=(X_1,...,X_n)$ for $X_i=(Y_i,D_i)$, where $D_i\in\{0,1\}$ is a binary treatment and $Y_i\in\{0,1\}$ is a binary outcome. The unknown parameter $\theta=(\mu_0,\mu_1)$ collects the mean potential outcomes $\mu_d=E[Y_i(d)]$ in a population of policy interest (e.g. where the researcher is considering rolling out the treatment), and the target parameter is the average treatment effect (ATE) $\kappa(\theta)=\mu_1-\mu_0$. We consider loss $l(a,\theta)=(a-\kappa(\theta))^2$ and action space $\mathcal{A}=[-1,1]$.  

One decision rule the researcher could consider is the (non-randomized) difference-in-means estimator $\hat\kappa_{DM}=\bar{Y}_1-\bar{Y}_0=\frac{\sum_i D_i Y_i}{\sum_i D_i }-\frac{\sum_i (1-D_i) Y_i}{\sum_i (1-D_i) }$, which induces loss $L_{DM}(\theta,X)=(\hat\kappa_{DM}-\kappa(\theta))^2$.\footnote{If either the treatment or control group is empty, define $\hat\kappa_{DM}=0$.}  Another decision rule is the (randomized) half-sample difference-in-means $\hat\kappa_{DM,\frac{1}{2}}$, which drops half of the treatment and control observations, selected at random, and computes the difference-in-means over the remaining observations.  This induces loss $L_{DM,\frac{1}{2}}(\theta,X)=E\left[\left(\hat\kappa_{DM,\frac{1}{2}}-\kappa(\theta)\right)^2|X\right]$ where the expectation is over the randomness induced by dropping observations. $\triangle$

\medskip

Standard decision-theoretic analysis proceeds by first characterizing preferences over a large menu of options (e.g. loss functions) including ones which may be infeasible, and then applying these preferences to select from the feasible set for a particular problem.  This approach is useful for characterizing preferences because the feasible set in a given setting is often quite restrictive, and it is easier to characterize preferences when working with a larger menu of choices.  For instance, constant loss functions that assign the same loss regardless of the state are used in many decision theory results, but will often be infeasible in estimation problems: if we consider an estimation problem under squared error loss $l(a,\theta)=(a-\kappa(\theta))^2,$ $L_\delta(\theta,X)$ will generally vary with $\theta$ so long as $\kappa(\theta)$ does.\footnote{By contrast, \cite{breza2025generalizability} consider an estimation problem where the researcher may decline to report an estimate at a constant cost, in which case certain constant loss functions are feasible.}

We assume the researcher has a statistical model $\mathcal{Q}$, which they may not fully trust. Formally, let $\mathcal{Q}=\{Q_\theta:\theta\in\Theta\}$ denote the model, where each $Q_\theta\in\Delta(\mathcal{X})$ is a distribution for the data.\footnote{For simplicity, we further assume that $\mathcal{X}$ is Polish, that all distributions are defined on the Borel $\sigma$-algebra, and that functions discussed are Borel-measurable.  Going forward, we suppress discussion of measure-theoretic details where possible.}
Let $\mathcal{L}$ denote the set of all bounded loss functions $L:\Theta\times\mathcal{X}\to\mathbb{R}$.\footnote{As is standard in the decision theory literature, proofs of the sufficiency direction of our representation theorems (the direction where the axioms imply the representation) will proceed on the domain of \textit{simple} (finite-valued) loss functions. Under an appropriate continuity axiom, such a representation admits a unique extension to bounded loss functions. See Appendix \ref{sec: appendix axioms} for details.} We take as primitive a family of preferences $\{\succsim_\Theta\},\{\succsim_\theta:\theta\in\Theta\}$ over loss functions in $\mathcal{L}$.  Throughout, $L\succsim L'$ means that $\succsim$ weakly prefers $L$ to $L'$.
The preference $\succsim_\Theta$ is the researcher's overall ranking of loss functions, and is the object we ultimately seek to characterize.  The auxiliary preferences $\{\succsim_\theta\}_{\theta\in\Theta}$ represent the researcher's evaluations given parameter value $\theta$; we axiomatize representations for these first and then impose additional axioms to derive the desired representation for $\succsim_\Theta$.  Intuitively, $\succsim_\theta$ represents the researcher's preference when they are certain that $\theta$ is the true parameter value but remain concerned that the data distribution $Q_\theta$ implied by their model may be misspecified. We therefore refer to $\{\succsim_\theta\}_{\theta \in \Theta}$ as conditional preferences. While we take these as primitive in the main text for ease of exposition, in Appendix \ref{sec: appendix axioms} we show how they may be derived from the researcher's actual preference over loss functions $\succsim_\Theta$.

By misspecification, we will mean that the data $X$ are distributed according to $P$ which differs from the distribution $Q_\theta$ implied by the true $\theta$.  Importantly, we assume the true value of $\theta$ remains well-defined even when the statistical model is wrong.  This is natural for parameters like causal effects and counterfactuals which remain well-posed even when the model the researcher uses to estimate them is incorrect, but is less natural for e.g. parameters in a parametric utility, for which it is difficult to define a ``true value'' absent correct specification.\footnote{We could weaken this assumption to instead require only that some function $\kappa(\theta)$ that enters the loss, $l(a,\theta)=\tilde{l}(a,\kappa(\theta)),$ have a model-agnostic definition, where the model implies a set of data distributions $\mathcal{Q}(\kappa^*)=\{Q_\theta:\theta\in\Theta,\kappa(\theta)=\kappa^*\}$ compatible with a given value $\kappa^*$. In particular, observe that for any function $f(Q_\theta)$, we have $\sup_{\theta\in\Theta}f(Q_\theta)=\sup_{\kappa^*}\sup_{Q\in\mathcal{Q}(\kappa^*)}f(Q),$ which may be used to reformulate our results in this format.  While this is a weaker assumption on the interpretation of $\theta,$ it leads to substantially heavier exposition for some of our results, so we impose the stronger condition that the full vector $\theta$ has a model-free interpretation, or equivalently that $\kappa$ does and that $\kappa(\cdot)$ is invertible.}  It also rules out the case where $\theta$ is defined as a statistical functional $\theta:\mathcal{Q}\to\mathbb{R}^p,$ since in this case the model is either well-specified (when $P\in\mathcal{Q}$) or $\theta$ is undefined (when $P\not\in\mathcal{Q}$).  In the terminology of \cite{andrews2025purpose}, we thus consider the problem of misspecification in the context of an \emph{econometric model} described by the pair $(\theta,P),$ rather than a  \emph{statistical} model described solely by the distribution $P$ of the data.
 
\paragraph{Example: Average Treatment Effects, Continued}

Suppose the researcher's model posits that the data are generated by (i) drawing a random sample from a population of interest which (ii) satisfies the potential outcomes model with each unit's outcome depending only on their own treatment and (iii) running a randomized trial where treatment is independently assigned to each unit with probability $\frac{1}{2}$. Under these assumptions $Y_i=Y_i(D_i)$, where $D_i$ is independent of the potential outcomes $(Y_i(0),Y_i(1))$.  Thus, under the researcher's model the 
observations $X_i=(Y_i,D_i)$ are i.i.d. draws from a multinomial distribution supported on $\{0,1\}^2$ with $E_{Q_\theta}[D_i]=\frac{1}{2},$ $E_{Q_\theta}[Y_i|D_i=0]=\mu_0,$ and $E_{Q_\theta}[Y_i|D_i=1]=\mu_1.$   

There are many ways in which the researcher's model could be wrong.  For instance, the population from which experimental participants are sampled could differ from the population of policy interest (in which case the marginal distribution of potential outcomes $Y_i(d)$ in the trial could differ from that implied by $\theta$).  Alternatively, the experimental sample might not be drawn i.i.d., e.g. oversampling the friends of early experimental participants due to recruitment through social networks (in which case the outcomes would not be independent across units).  Finally, treatment might not be assigned as prescribed by the protocol, e.g. treating participants with lower baseline outcomes with higher probability (in which case treatment would not be independent of potential outcomes).  If we combine these possibilities, any value of $\theta=(\mu_0,\mu_1)$ could in principle be compatible with any distribution $P\in\Delta(\{0,1\}^{2n})$ for the observable data $X$.

For the forms of misspecification discussed above, $\theta$ remains well-defined as the average potential outcomes in the target population, though it will not be identified absent restrictions on the possible misspecification.  There are other forms of misspecification one could contemplate, e.g. spillovers across units, where the definition of $\theta$ becomes more delicate. If spillovers are present, for instance, does $\mu_0$ correspond to the average outcome when no unit is treated, or when a given unit is untreated while others are treated i.i.d. with probability $\frac{1}{2}$?  For our analysis, we presume the researcher has adopted a definition for the ``true'' $\theta$ which remains well-posed under the forms of misspecification they contemplate, and wishes to choose among decision rules in a way which is robust to this misspecification concern. $\triangle$

\medskip

Results by \cite{maccheroni2006ambiguity} and \cite{CerriaViolioetal24} imply axioms on the preference relations $\{\succsim_\Theta\},\{\succsim_\theta:\theta\in\Theta\}$ which hold if and only if these preferences are represented by $V_\Theta$ and $\left\{V_\theta:\theta\in\Theta\right\}$ respectively, for 
\begin{equation}\label{eq:CV-rep}
V_\Theta(L)=\sup_{\theta\in\Theta}\sup_{P\in\Delta(\mathcal{X})}\left\{\int L(\theta,x)dP(x)-c_\theta(P)\right\}
\end{equation}
\begin{equation}\label{eq:variational-rep}
V_\theta(L)=\sup_{P\in\Delta(\mathcal{X})}\left\{\int L(\theta,x)dP(x)-c_\theta(P)\right\}
\end{equation}
in the sense that 
\[ 
L\succsim_\Theta L' \iff V_\Theta(L)\le V_\Theta(L')\]
\[
L\succsim_\theta L'\iff V_\theta(L)\le V_\theta(L').
\]
In both representations, $c_\theta:\Delta(\mathcal{X})\to[0,\infty]$ is a convex, lower-semicontinuous function with $c_\theta(Q_\theta)=0$ for each $\theta\in\Theta$. Intuitively, the preference $\succsim_\theta$
evaluates loss functions $L$ by (i) focusing on their behavior at parameter value $\theta$ and (ii) considering the penalized worst-case risk, which averages over $P$ but then subtracts off $c_\theta(P)$, effectively penalizing expected loss under data distributions $P$ that the researcher finds less plausible.  We further assume that the researcher finds the model-implied distribution $Q_\theta$ at least as plausible as any other. The preference $\succsim_\Theta$ does the same for each $\theta,$ but further takes the worst case over all $\theta$. As an immediate consequence, $V_\Theta$ implies the (penalized) worst-case risk bound
\[
\int L(\theta,x)dP(x) \le V_\Theta(L)+c_\theta(P) ~~ \forall \theta\in\Theta,P\in\Delta(\mathcal{X})
\]
which controls, uniformly over $(\theta,P)$, how large the risk may be.  Choosing a loss function which minimizes $V_\Theta(L)$ is thus the same as minimizing an upper bound on the risk.

\paragraph{Example: Average Treatment Effects, Continued}

Consider a researcher choosing between the Horvitz-Thompson estimator $\hat\kappa_{HT}=\frac{1}{n}\sum_i 2(2D_i-1)Y_i$ and the difference-in-means estimator $\hat\kappa_{DM}=\bar{Y}_1-\bar{Y}_0$.  Provided the researcher's preferences are represented by \eqref{eq:CV-rep}, they weakly prefer the difference-in-means estimator if and only if 
\[
V_\Theta(L_{DM})=\sup_{\theta\in[0,1]^2}\sup_{P\in\Delta(\{0,1\}^{2n})}\int \left(\hat\kappa_{DM}-\kappa(\theta)\right)^2dP(x)-c_\theta(P)\le\]
\[
\sup_{\theta\in[0,1]^2}\sup_{P\in\Delta(\{0,1\}^{2n})}\int \left(\hat\kappa_{HT}-\kappa(\theta)\right)^2dP(x)-c_\theta(P)=V_\Theta(L_{HT}).
\]
Moreover, the optimized value $V_\Theta(L_{DM})$ of the left-hand side provides a bound on how quickly the performance of the difference in means estimator moves away from that prescribed by the researcher's model.
 $\triangle$

\medskip

The representations \eqref{eq:CV-rep} and \eqref{eq:variational-rep} follow from arguments in  \cite{maccheroni2006ambiguity} and \cite{CerriaViolioetal24}, but our framing of the problem (e.g. choice domain, preferences) differs from theirs.  For completeness, and to aid interpretation for readers more used to working with risk and loss functions than the conventional decision-theory setup, we thus provide an axiomatization for our choice domain, along with proofs, in Appendix \ref{sec: appendix axioms}.

The axiomatization discussed in Appendix \ref{sec: appendix axioms} leaves the form of the penalty function $c$ unspecified.  Different choices of $c$ correspond to different attitudes toward misspecification, and restricting these attitudes narrows the class of penalties.  We next discuss three important cases: constraint preferences, multiplier preferences, and a novel class of constrained multiplier preferences that combines the other two. In each case, we discuss the additional axioms that distinguish these preferences from the broader class characterized \eqref{eq:CV-rep} and \eqref{eq:variational-rep}.

\subsection{Constraint Preferences}

A particularly simple penalty arises when the researcher requires that the true distribution $P$ lie in some set but imposes no further penalties. Consider a closed, convex set $\mathcal{P}_\theta\subseteq\Delta(\mathcal{X})$ of distributions (the ``ambiguity set'') that the researcher treats as plausible under parameter value $\theta$. We require that $Q_\theta \in \mathcal{P}_\theta$ so that the researcher finds their model plausible. Preferences of the form

\[
V_\Theta(L)=\sup_{\theta\in\Theta}\sup_{P\in\mathcal{P}_\theta}\int L(\theta,x)dP(x)
\]

\begin{equation}\label{eq: cond GS-pref}
V_\theta(L)=\sup_{P\in\mathcal{P}_\theta}\int L(\theta,x)dP(x)
\end{equation} 
evaluate loss functions based on their worst-case performance over both parameters $\theta$ and distributions $P$ in the ($\theta$-specific) ambiguity set.  Criteria of this form have a long history, and were discussed by e.g. \cite{Huber1964} in the context of contamination neighborhoods, and more recently by \cite{BonhommeWeidner2022} for estimation of misspecified economic models, while similar misspecification-aware setups are considered for confidence set construction by \cite{ArmstrongKolesar2020} and \cite{christensen2023sensitivity}.\footnote{Note that while the class $\mathcal{P}_\theta$ considered in e.g. \cite{BonhommeWeidner2022} need not be convex, $V_\Theta(L)$ and $V_\theta(L)$ are unchanged if we replace all non-convex $\mathcal{P}_\theta$ by their convex hull.}  Much of the large and active literature on distributionally robust optimization \citep[e.g.][]{DuchiNamkoong2021DRO,GaoKleywegt2023DRO,MontielOleaetal2026} can also be cast in this form.\footnote{For instance, the problem of \cite{DuchiNamkoong2021DRO} can be cast into our setting by taking $X=(X_1,...,X_n)$ to represent $n$ i.i.d. draws $X_i\in \mathcal{X}_0$, $\Theta\subseteq\Delta(\mathcal{X}_0),$ $Q_\theta=\times_{i=1}^n\theta=\theta^n$, $l(a,\theta)=E_\theta[l^*(a,X_i)]$ for some loss function $l^*,$ and $\mathcal{P}_\theta$ the convex hull of $\left\{P_0^n:P_0\in\Delta(\mathcal{X}_0),D_f(\theta|P_0)\le\rho\right\}$ for $D_f$ an $f$-divergence.}

\paragraph{Example: Average Treatment Effects, Continued}

Now suppose the experiment is conducted by two teams, with unit $i$ randomly assigned to team $C_i\in\{1,2\}$, where $\pi_j=P\{C_i=j\}$.  The researcher's model maintains that outcomes do not depend on team assignment, so $E[Y_i(d)\mid C_i=j]=\mu_d$ for all $j$, and the likelihood pools data from both teams.  
However, the researcher has greater confidence in team 1's execution than in team 2's: for instance, team 2 may have deviated from the treatment assignment protocol, implemented the treatment less carefully, or measured outcomes differently.  We may express this using constraint preferences.  

Let $\mathcal{P}_\theta$ be the set of distributions which arise by (i) drawing $C_i$ i.i.d. with $P\{C_i=j\}=\pi_j$ (ii) drawing $(Y_i,D_i)|C_i=1$ as described by the researcher's model and (iii) drawing the remaining outcomes $\{(Y_i,D_i)|i\in\{1,...,n\}\text{ such that }C_i=2\}$ from some other distribution.  This case lies strictly between full trust in the model ($\mathcal{P}_\theta=\{Q_\theta\}$) and complete agnosticism ($\mathcal{P}_\theta=\Delta(\mathcal{X})$): it requires that team 1's observations follow the model, while placing no restrictions (including independence) on team 2's observations. $\triangle$

\medskip

The preference \eqref{eq: cond GS-pref} is a special case of \eqref{eq:variational-rep} with the penalty function
\begin{equation}\label{eq: GS indicator}
c_\theta(P)=\begin{cases}0 & \text{if } P\in\mathcal{P}_\theta\\ \infty & \text{if } P\notin\mathcal{P}_\theta\end{cases}.
\end{equation}
These preferences, which can be viewed as a version of \citet{GilboaSchmeidler1989} min-max preferences, are distinguished from the more general class of variational preferences by the Certainty Independence axiom.    Certainty independence is stated using a \emph{constant loss}, by which we mean a loss function $r\in\mathcal{L}$ that takes the same value $r\in\mathbb{R}$ for all $(\theta,x)$, and we slightly abuse notation by identifying constant functions with their value. 
\begin{axiom}[Certainty Independence]\label{ax:CI}
For all $\theta\in\Theta$, $L,L'\in\mathcal{L}$, $r\in\mathbb{R}$, and $\alpha\in(0,1)$,
\[L \succsim_\theta L' \iff \alpha L+(1-\alpha)r \succsim_\theta \alpha L'+(1-\alpha)r.\]
\end{axiom}

Intuitively, certainty independence requires that our preference over estimators not change if, with some probability independent of the data and parameter, we will be randomly switched to instead use an estimator with constant loss.

\paragraph{Example: Average Treatment Effects, Continued}

Again consider a researcher choosing between Horvitz-Thompson and difference-in-means.  Now suppose that for each data realization this researcher is, with probability $1-\alpha$, randomized into instead using a noisy oracle $\hat\kappa_O=\kappa(\theta)+\varepsilon,$ where $\varepsilon\sim N(0,\sigma^2)$.  This results in randomized procedures with induced losses
$\alpha L_{HT}+(1-\alpha)\sigma^2$ and
$\alpha L_{DM}+(1-\alpha)\sigma^2$.
Certainty Independence requires that, for all $\alpha\in(0,1)$ and $\sigma^2$, the preference between
these two randomized procedures be the same as the preference between the pure
estimators $\hat\kappa_{HT}$ and $\hat\kappa_{DM}$. $\triangle$

\medskip

\begin{prop}[Gilboa and Schmeidler, 1989]\label{prop:GS}
Suppose that the conditional preferences $\{\succsim_\theta:\theta\in\Theta\}$ are  represented by \eqref{eq:variational-rep}.  Certainty Independence holds if and only if $c_\theta(\cdot)$ is of the form \eqref{eq: GS indicator} for all $\theta\in\Theta$.
\end{prop}

While Proposition \ref{prop:GS} characterizes when the researcher's preferences have an (unpenalized) min-max form, the class of possible ambiguity sets $\mathcal{P}_\theta,$ and thus preferences, remains quite large.   \cite{GhirardatoMarinacci2002} show that one can go further and infer the ambiguity sets $\mathcal{P}_\theta$ directly from $\succsim_\theta$.

To state this result, we must introduce some additional notation.  For $(\theta,P)\in\Theta\times\Delta(\mathcal{X})$, define the \emph{subjective expected utility (SEU) preference} $\succsim_{\theta,P}^{SEU}$ by
\[L\succsim_{\theta,P}^{SEU} L' \iff \int L(\theta,x)dP(x)\le\int L'(\theta,x)dP(x).\]
This is the preference of a decision maker who ranks decision rules based on their expected loss under $(\theta,P)$.
Following \citet{GhirardatoMarinacci2002},  $\succsim_\theta$ is \emph{more ambiguity averse} than $\succsim_{\theta,P}^{SEU}$ if for all $L\in\mathcal{L}$ and $r\in\mathbb{R}$,
\[L\succsim_\theta r \Rightarrow L\succsim_{\theta,P}^{SEU} r.\]
In words, any time the preference $\succsim_\theta$ ranks a state-dependent loss $L$ as weakly better than a constant loss $r$, $\succsim_{\theta,P}^{SEU}$ must do so as well.  Equivalently, $\succsim_\theta$ is ``more cautious'' than expected loss under $(\theta,P)$: it has a (weakly) higher bar for preferring uncertain losses to certain ones.
\citet{GhirardatoMarinacci2002} establish that such comparisons identify $\mathcal{P}_\theta$.
\begin{prop}[Ghirardato and Marinacci, 2002]\label{prop:GS 2}
Suppose the conditions of Proposition \ref{prop:GS} hold, and let 
\[
\mathcal{P}_\theta=\{P\in\Delta(\mathcal{X}):\succsim_\theta\text{ is more ambiguity averse than }\succsim_{\theta,P}^{SEU}\}.\]  Then $c_\theta(\cdot)$ is equal to \eqref{eq: GS indicator} with this $\mathcal{P}_\theta$.
\end{prop}

In particular, since we have assumed that $Q_\theta \in \mathcal{P}_\theta$, Proposition \ref{prop:GS 2} implies that
$\succsim_\theta$ is more ambiguity averse than $\succsim_{\theta,Q_\theta}^{SEU}$.

\paragraph{Example: Average Treatment Effects, Continued}

Consider a researcher choosing between the difference-in-means estimator $\hat\kappa_{DM}=\bar{Y}_1-\bar{Y}_0$ and the noisy oracle
$\hat\kappa_O=\kappa(\theta)+\varepsilon,$ $\varepsilon\sim N(0,\sigma^2)$.  If the researcher's conditional preferences $\succsim_\theta$ are represented by \eqref{eq:variational-rep} and satisfy Certainty Independence, and thus have a constraint representation with $Q_\theta \in \mathcal{P}_\theta$ for each $\theta \in \Theta$, it follows that $L_O\succsim_\theta L_{DM}$ whenever $\sigma^2\le E_{Q_\theta}[(\hat\kappa_{DM}-\kappa(\theta))^2].$ $\triangle$

\subsection{Multiplier Preferences}

Constraint preferences compute the worst-case expected loss over $\mathcal{P}_\theta$, and thus do not privilege any particular distribution in this set.  Since we have assumed the researcher has a model $\mathcal{Q}$, however, it also seems natural that under parameter value $\theta$ they might find distributions ``close'' to $Q_\theta$ more plausible than those ``far'' from $Q_\theta$, and thus wish to penalize some notion of distance from the base model.  As has previously been observed many times \citep[e.g. by][]{HansenSargent2001, HansenSargent2008},  Kullback-Leibler (KL) divergence is an especially convenient penalty for many purposes, and leads to
\[
V_\Theta(L)=\sup_{\theta\in\Theta}\sup_{P\in\Delta(\mathcal{X})}\left\{\int L(\theta,x)dP(x)-\lambda\cdot \KL(P\|Q_\theta)\right\},
\]
\begin{equation}\label{eq: mult preference}
V_\theta(L)=\sup_{P\in\Delta(\mathcal{X})}\left\{\int L(\theta,x)dP(x)-\lambda\cdot \KL(P\|Q_\theta)\right\},
\end{equation}
where $\lambda>0$ and
\[\KL(P\|Q)=\int\log\left(\frac{dP}{dQ}(x)\right)dP(x)\]
when $P\ll Q$ (i.e. all events which are probability zero under $Q$ are also probability zero under $P$) and $\KL(P\|Q)=\infty$ otherwise.
The parameter $\lambda$ controls the degree of concern for misspecification: as $\lambda\to\infty$, the preference converges to expected loss under $Q_\theta$, while as $\lambda\to 0$, the preference places more weight on worst-case scenarios.  The multiplier risk index $V_\Theta$ is equivalent to the worst-case Bayes risk used by \cite{adusumilli2026} for asymptotic analysis of point estimation and treatment assignment problems.

Results in the literature again characterize this penalty relative to the broad class $c_\theta(\cdot)$.  \citet{strzalecki2011axiomatic} shows a tight connection between the KL divergence and Savage's Sure Thing Principle.  To state the Sure Thing Principle, for all $L,L'\in \mathcal{L}$ define a spliced loss equal to $L$ on an event $\mathcal{E}$ and $L'$ otherwise,
\[
L_\mathcal{E}L'(\theta,X)=\begin{cases}L(\theta,X) & \text{if } X\in\mathcal{E}\\ L'(\theta,X) & \text{if } X\not\in\mathcal{E}\end{cases}.
\]
\begin{axiom}[Sure Thing Principle]\label{ax:STP}
For all $\theta\in\Theta$, $\mathcal{E}\subseteq\mathcal{X}$, and $L,L',M,M'\in\mathcal{L}$, 
\[L_\mathcal{E}M\succsim_\theta L'_\mathcal{E}M \iff L_\mathcal{E}M'\succsim_\theta L'_\mathcal{E}M'.\]
\end{axiom}

The Sure Thing Principle requires that preferences between loss functions that agree on some set ($\mathcal{E}^c$, in this case) depend only on their values elsewhere.  In our setting, this axiom has a natural connection to pretesting.

\paragraph{Example: Average Treatment Effects, Continued}

The experimental design implies $E_{Q_\theta}[D_i]=\frac{1}{2}$, which the researcher could test in order to check implementation fidelity.  Suppose that, conditional on not rejecting the null $H_0:E_P[D_i-\frac{1}{2}]=0$, the researcher will choose between the Horvitz-Thompson estimator $\hat\kappa_{HT}$ and the difference-in-means estimator $\hat\kappa_{DM}$, while if the null is rejected they will use some alternative procedure.  The Sure Thing Principle requires that the ranking between $\hat\kappa_{HT}$ and $\hat\kappa_{DM}$ under $\succsim_\theta$, conditional on the test not rejecting, not vary depending on what procedure is used when the test rejects.
$\triangle$

\medskip
 
The results of \cite{strzalecki2011axiomatic} imply that the Sure Thing Principle holds if and only if $c_\theta(\cdot)$ is KL divergence relative to some centering distribution. Our requirement that $c_\theta(Q_\theta)=0$ further ensures that this centering distribution is $Q_\theta$ (see Appendix \ref{sec: appendix axioms} for an axiomatic justification).  We further impose an axiom, 
 adapted from \cite{lanzani2025supplement}, which enforces that the misspecification-aversion parameter $\lambda$ is constant across $\theta$. 

\begin{axiom}[Uniform Misspecification Concern]
For each $\theta,\theta' \in \Theta$ and each $L,L' \in \mathcal{L}$ with
\begin{equation}\label{eq: pushforward equality}
Q_\theta \circ L_\theta^{-1}=Q_{\theta'} \circ (L_{\theta'}')^{-1}
\end{equation}
where $Q_\theta \circ L_\theta^{-1}$ denotes the distribution of $L(\theta,X)$ when $X\sim Q_\theta,$ it holds that
\[
L \succsim_\theta r \iff L'\succsim_{\theta'} r \quad \forall r \in \mathbb{R}
\]
\end{axiom}

\paragraph{Example: Average Treatment Effects, Continued} Suppose the researcher is considering an estimator $\hat{\kappa}$ whose induced distribution over squared estimation errors $ (\hat{\kappa}-\kappa(\theta))^2$
does not depend on $\theta$. For example, consider the (infeasible) estimator $\hat{\kappa}$ which yields an estimation error of $1/2$ if the sample fraction of treated individuals exceeds $1/2$ and $0$ otherwise.\footnote{Note that in the limit experiment discussed in the next section, there do exist estimators whose model-implied distribution does not depend on the parameter.} Since the law of the sample fraction of treated individuals is $n^{-1}\text{Bin}(n,1/2)$ under $Q_\theta$ for all $\theta$, $\hat{\kappa}$ satisfies the property \eqref{eq: pushforward equality} for all $\theta,\theta'\in\Theta$. $\triangle$

\medskip

\begin{axiom}[Monotone Continuity]\label{ax:6}
For all $\theta \in \Theta$, $L,L' \in \mathcal{L}$, $r \in \mathbb{R}$, and sequences of events\footnote{In this axiom, events are Borel subsets of $\mathcal{X}$.} $\{\mathcal{E}_n\}_{n\geq 1}$ with $\mathcal{X}\supseteq \mathcal{E}_1 \supseteq \mathcal{E}_2 \supseteq \cdots$ and $\cap_{n\geq 1} \mathcal{E}_n=\emptyset$: if $L \succ_\theta L'$, there exists $n^*\geq 1$ such that $r_{\mathcal{E}_{n^*}}L\succ_\theta L'$.
\end{axiom}

Loosely speaking, Monotone Continuity ensures that the representation only cares about countably additive probabilities. When $\mathcal{X}$ is finite, it has no bite.

Finally, say that an event $\mathcal{E} \subseteq \mathcal{X}$ is \textit{nonnull} under $\succsim_\theta$ if there exist $L,L',M \in \mathcal{L}$ such that $L_{\mathcal{E}} M \succ_\theta L'_{\mathcal{E}} M$. To state our representation theorem, we make the mild assumption that $\mathcal{X}$ has at least three disjoint nonnull events for each $\succsim_\theta$, and that there exists $q\in(0,1)$ such that for each $\theta,$ there exists an event $\mathcal{E}_\theta\subseteq \mathcal{X}$ with $Q_\theta(\mathcal{E}_\theta)=q$. 

\begin{prop}[Strzalecki 2011, Cerreia-Vioglio et al. 2025]\label{prop:multiplier}
Suppose that the conditional preferences $\{\succsim_\theta:\theta\in\Theta\}$ are  represented by \eqref{eq:variational-rep}. The Sure Thing Principle, Uniform Misspecification Concern, and Monotone Continuity hold if and only if $c_\theta(P)=\lambda\cdot \KL(P\|Q_\theta)$ for some $\lambda>0$.
\end{prop}

\subsection{Constrained Multiplier Preferences}

Multiplier preferences treat all forms of misspecification symmetrically: deviations from the model are penalized solely based on their KL divergence.  
This is natural when the researcher has no view about which aspects of the model are more likely to fail.  In many economic applications, however, researchers appear to have more confidence in some model predictions than in others.  For such researchers, it is natural to combine both elements: a hard constraint ruling out certain forms of misspecification a priori, and a KL penalty governing concern about remaining misspecification within the ambiguity set.

To capture the resulting preferences, we introduce a novel class of \emph{constrained multiplier preferences}, which rule out some DGPs a priori and then continuously penalize deviations within the ambiguity set using Kullback-Leibler divergence:
\begin{equation}\label{eq:CCM-rep} 
V_\Theta(L)=\sup_\theta\sup_{P\in\mathcal{P}_\theta}\left\{\int L(\theta,x)dP(x)-\lambda\cdot \KL(P\|Q_\theta)\right\}
\end{equation}
\begin{equation}\label{eq:CM rep}
V_\theta(L)=\sup_{P\in\mathcal{P}_\theta}\left\{\int L(\theta,x)dP(x)-\lambda\cdot \KL(P\|Q_\theta)\right\}.
\end{equation}
Mathematically, \eqref{eq:CCM-rep} and \eqref{eq:CM rep} correspond to special cases of \eqref{eq:CV-rep} and \eqref{eq:variational-rep}, respectively, which take $c_\theta(\cdot)$ equal to the sum of the convex indicator for the set $\mathcal{P}_\theta$, as in constraint preferences, and the KL divergence from the model distribution $Q_\theta$, as in multiplier preferences.

\paragraph{Example: Average Treatment Effects, Continued}

In our discussion of constraint preferences we considered the case where $\mathcal{P}_\theta$ exactly pins down the 
distribution of the data collected by team 1 while imposing no constraints on the data from team 2.  This represents extreme distrust of the second team, for instance treating it as equally plausible that this team adhered faithfully to the experimental protocol and that they fabricated the data wholesale.  Constrained multiplier preferences accommodate the
intermediate case where the researcher thinks the data from the second team may have a distribution different than that predicted by the model (in principle allowing any $P\in\mathcal{P}_\theta$ with $\KL(P\|Q_\theta)<\infty$) but continuously discounts distributions $P$ which are further from $Q_\theta$ as measured by KL divergence. $\triangle$

\medskip

Our axiomatization of constrained multiplier preferences will build on the constraint and multiplier preference axiomatizations.  Specifically, let $\succsim_{\theta}^C$ denote the preference the researcher would have if they (i) took as given that the true parameter is $\theta$ and (ii) were certain the true distribution lay in $\mathcal{P}_\theta$ but did not privilege any distribution in this set.  Formally, we assume the preferences $\left\{\succsim_{\theta}^C:\theta\in\Theta\right\}$ satisfy the conditions of Propositions \ref{prop:GS} and \ref{prop:GS 2}, and thus are constraint preferences with ambiguity set $\mathcal{P}_\theta$.  

Similarly, let $\succsim_\theta^M$ denote the preference the researcher would have if they (i) took as given that the true parameter is $\theta$ and (ii) were more concerned with performance ``close'' to the model-implied distribution $Q_\theta$ but did not have hard constraints on the class of possible data distributions and (iii) satisfied the Sure Thing Principle.  Formally, we assume the preferences $\left\{\succsim_{\theta}^M:\theta\in\Theta\right\}$ satisfy conditions of Proposition \ref{prop:multiplier}, and thus are multiplier preferences with centering distribution $Q_\theta$.

We consider a researcher who both believes the constraints on misspecification imposed by $\succsim_{\theta}^C$ and is more concerned with DGPs close to $Q_\theta$ as in $\succsim_\theta^M$.  The next axiom captures how their conditional preference given $\theta$, $\succsim_\theta$, combines these two elements.
\begin{axiom}[Indirect Pareto]\label{ax:indirect-pareto}
For all $\theta\in\Theta$, $L\in\mathcal{L}$, and $r\in\mathbb{R}$,
\[
L\succ_\theta r \iff 
\]
\[~ \exists~ L_C,L_M\in\mathcal{L},~r_C,r_M\in\mathbb{R}~ \text{s.t.}~L=L_C+L_M,~ r=r_C+r_M,~ L_C\succ_\theta^C r_C,~ L_M\succ_\theta^M r_M.\]

\end{axiom}

The Indirect Pareto axiom requires that the conditional preference $\succsim_\theta$ strictly prefers the loss function $L$ to a constant loss $r$ if and only if $L$ can be decomposed into two parts, corresponding to the constraint and multiplier preferences respectively, each of which is strictly preferred to its share of the constant under the respective component preference.  We show that this property characterizes the constrained multiplier preference.

\begin{thm}\label{thm:constrained-multiplier}
Suppose each $\succsim_\theta^C$ is a constraint preference with ambiguity set $\mathcal{P}_\theta$ and each $\succsim_\theta^M$ is a multiplier preference with centering measure $Q_\theta$ and parameter $\lambda$.  Under Axiom \ref{ax:indirect-pareto}, the preference $\succsim_\theta$ has the constrained multiplier representation (\ref{eq:CM rep}).
\end{thm}

\subsection{Dual Representation of Constrained Multiplier Preferences}

One practically appealing feature of multiplier preferences is that they imply highly tractable dual representations.
\begin{prop}[Dupis and Ellis, 1997]\label{thm:multiplier-dual}
For $V_\theta$ as in \ref{eq: mult preference},
\begin{equation}\label{eq:multiplier-dual}
V_\theta(L)=\lambda\cdot\log\left(E_{Q_\theta}\left[\exp\left(\frac{1}{\lambda}L(\theta,X)\right)\right]\right).
\end{equation}
\end{prop}
This result shows that the ranking over losses implied by multiplier preferences is precisely the same as the ranking one would obtain by assuming the researcher's model $Q_\theta$  is correct but using the exponentiated loss $\exp\left(\frac{1}{\lambda}L\right)$.  Consequently, one can compute optimal decision rules under multiplier preferences by applying standard  arguments for the correctly-specified case to the transformed loss function.

Parallel convex duality arguments also imply a tractable dual for constrained multiplier preferences in some contexts.  In particular, we focus on the case where the ambiguity set $\mathcal{P}_\theta$ can be written as the set of distributions satisfying a collection of moment equalities. For ambiguity sets of this form, convex duality arguments similar to those in the generalized empirical likelihood literature \citep{NeweySmith2004,Kitamura2009} imply a finite-dimensional dual for the constrained multiplier problem, even when $L$ may be unbounded.

\begin{prop}\label{prop:constrained-dual init}
Suppose $\mathcal{P}_\theta=\{P\in\Delta(\mathcal{X}):E_P[\varphi(\theta,X)]=0\}$ for a vector of moment functions $\varphi:\Theta\times \mathcal{X} \to\mathbb{R}^b$ satisfying $E_{Q_\theta}[\varphi(\theta,X)]=0$.  Then for $V_\theta$ as in (\ref{eq:CM rep}), where $L$ may be unbounded but $V_\theta(L)<\infty$,
\begin{equation}\label{eq:constrained-dual}
V_\theta(L)=\inf_{\beta\in\mathbb{R}^b}\lambda\cdot\log\left(E_{Q_\theta}\left[\exp\left(\frac{1}{\lambda}L(\theta,X)-\beta'\varphi(\theta,X)\right)\right]\right).
\end{equation}
\end{prop}
The representation (\ref{eq:constrained-dual}) introduces Lagrange multipliers $\beta$ that enforce the moment constraints.  The infimum over $\beta$ is the dual to the original problem of maximizing over $\mathcal{P}_\theta$.  While this optimization problem does not in general have a closed form solution, it is convex and sufficiently tractable to enable both computation and theoretical analysis.

\paragraph{Example: Average Treatment Effects, Continued}

Consider the two-team ATE setting introduced above.  The ambiguity set $\mathcal{P}_\theta$ restricts the joint distribution of the team-1 observations.  Fully expressing these constraints using moment equalities would require per-observation restrictions (e.g. $E_P[(Y_i-\mu_0)(1-D_i)1\{C_i=1\}]=0$ for each $i$) and cross-observation restrictions (e.g. zero covariance between $(D_i-\frac{1}{2})1\{C_i=1\}$ and $(D_j-\frac{1}{2})1\{C_j=1\}$ for $i\ne j$), with the total number of moment equalities growing with the sample size.

To obtain a more parsimonious (but also more permissive) set of moment conditions, we may instead constrain a sample average moment function. Let
\begin{equation}\label{eq: psi def}
\psi(\theta,X_i)=
\begin{pmatrix}(Y_i-\mu_0)(1-D_i)1\{C_i=1\}\\
(Y_i-\mu_1)D_i1\{C_i=1\}\\
(D_i-\tfrac{1}{2})1\{C_i=1\}\\
1\{C_i=1\}-\pi_1 \end{pmatrix}
\end{equation}
and define $\varphi(\theta,X)=\frac{1}{n}\sum_{i=1}^n\psi(\theta,X_i)$. The first two components capture the conditional mean outcomes among treated and control units in team 1, while the third and fourth elements restrict treatment and team assignment, respectively.  
The constraint $E_P[\varphi(\theta,X)]=0$ requires that the moment conditions hold on average across units, but does not constrain the marginal for a given unit and thus leads to a weakly higher worst-case risk.  As we discuss in the next section we can strengthen the constraint by including higher moments of the sample average.

With the moment-equality constraint set $\mathcal{P}_\theta=\{P:E_P[\varphi(\theta,X)]=0\}$, Proposition \ref{prop:constrained-dual init} applies, and the researcher's preference $\succsim_\Theta$ ranks estimators $\hat\kappa_\delta$ based on
\[
\sup_{\theta\in\Theta}\inf_{\beta\in\mathbb{R}^4}\lambda\cdot\log\left(E_{Q_{n,\theta}}\left[\exp\left(\frac{1}{\lambda}(\hat\kappa_\delta-\kappa(\theta))^2-\beta'\varphi(\theta,X)\right)\right]\right),\]
where we have used that the ranking is unchanged by monotone transformations of $V_\Theta$.
$\triangle$

Before moving on, we briefly note that the analytic tractability of constrained multiplier preferences extends to settings with moment inequality, rather than equality, constraints.  Specifically, if $\mathcal{P}_\theta=\{P:E_P[\varphi(\theta,X)]\le 0\}$, one can extend Proposition \ref{prop:constrained-dual init} to show that
\[
V_\theta(L)=\inf_{\beta\in\mathbb{R}^b}\sup_{\eta\in\mathbb{R}^b_+}\lambda\cdot\log\left(E_{Q_\theta}\left[\exp\left(\frac{1}{\lambda}L(\theta,X)-\beta'(\varphi(\theta,X)+\eta)\right)\right]\right)=\]
\[ \inf_{\beta\in\mathbb{R}^b_+}\lambda\cdot\log\left(E_{Q_\theta}\left[\exp\left(\frac{1}{\lambda}L(\theta,X)-\beta'\varphi(\theta,X)\right)\right]\right).\]

\section{Asymptotic Risk Bounds}\label{sec: LAM}

While we derived constrained multiplier preferences in a finite-sample setting, in most interesting economic models finite-sample performance is analytically intractable.  Following a foundational analytic approach for models without misspecification concern, we thus study local asymptotic performance instead.  This section develops a local asymptotic minimax (LAM) theorem for constrained multiplier preferences, paralleling a classical result for the correctly specified case.  We begin by introducing the asymptotic framework and reviewing the classical LAM theorem as a benchmark, then state our LAM result and discuss its implications.

\subsection{The Classical LAM Theorem}

Consider a sequence of estimation problems indexed by the sample size $n$.  For sample size $n$ the researcher observes data $X^n=(X_1,\ldots,X_n)$, which under their model is drawn from a distribution in $\mathcal{Q}_n=\{Q_{n,\theta}:\theta\in\Theta\}$, where $\Theta\subseteq\mathbb{R}^p$.\footnote{While we focus on parametric models to tractably accommodate time-series applications, for i.i.d. data we expect that our analysis below, like that of \cite{adusumilli2026}, will extend to semiparametric models.}  For instance, if the researcher's model implies the data are i.i.d. then $Q_{n,\theta}=\times_{i=1}^n Q_{1,\theta}$ for $Q_{1,\theta}$ the distribution of a single observation, though the framework allows for more general dependence structures.  We study performance when the true parameter is local to a base value $\theta_0$, taking the form $\theta_{n,h}=\theta_0+h/\sqrt{n}$ for a local parameter $h\in H=\mathbb{R}^p$, and we shorthand $Q_{n,\theta_{n,h}}=Q_{n,h}$.  The loss function in the sample of size $n$ is
\[
l_n(a,\theta)=\ell(\sqrt{n}(a-\kappa(\theta))),
\]
where $\kappa:\Theta\to\mathbb{R}^d$ is the target parameter and $\ell:\mathbb{R}^d\to\mathbb{R}$ is a fixed loss function.  We are interested in the worst-case asymptotic performance of decision rule sequences $\delta_n$ over the local parameter space $H$.

The notion of local asymptotic normality, a foundational tool for characterizing asymptotic performance, formalizes a sense in which regular statistical models are asymptotically equivalent to Gaussian location experiments.

\begin{defn}\label{defn: LAN}
The sequence of models $\{Q_{n,\theta}:\theta\in\Theta\}$ is \emph{locally asymptotically normal} (LAN) at $\theta_0$ with scaling coefficient $\sqrt{n}$ if there exists a sequence of random vectors $S_{n,\theta_0}$ and a nonsingular matrix $I_0$ such that for every sequence $h_n\to h$,
\[
\log\left(\frac{dQ_{n,\theta_0+h_n/\sqrt{n}}}{dQ_{n,\theta_0}}\right)=h'S_{n,\theta_0}-\frac{1}{2}h'I_0h+o_{Q_{n,\theta_0}}(1),
\]
where $S_{n,\theta_0}\stackrel[d]{}\to N(0,I_0)$ under $Q_{n,\theta_0}$.
\end{defn}

The LAN condition says that, in a local neighborhood of $\theta_0$, the log-likelihood ratio is asymptotically quadratic in the local parameter $h$ with Hessian $-I_0$, and thus resembles the log-likelihood of a normal model with Fisher information (and inverse variance) $I_0$.  For i.i.d. data, LAN follows from standard differentiability conditions on the single-observation density \citep[see e.g.][Chapter 7]{van_der_vaart_asymptotic_1998}.

For LAN models, the classical local asymptotic minimax theorem gives a lower bound on the worst-case local asymptotic risk of any sequence of estimators.

\begin{prop}[Local Asymptotic Minimax; \citealt{van_der_vaart_asymptotic_1998}, Theorems 8.11 and 9.4]\label{prop: LAM}
Suppose $\{Q_{n,\theta}:\theta\in\Theta\}$ is LAN at $\theta_0\in\mathrm{int}(\Theta)$ with nonsingular Fisher information $I_0$.  Then for any symmetric, quasi-convex loss $\ell:\mathbb{R}^d\to\mathbb{R}_+$ minimized at zero, any $\kappa:\Theta\to\mathbb{R}^d$ with derivative $K=\frac{\partial}{\partial\theta'}\kappa(\theta_0)\in\mathbb{R}^{d\times p}$, and any sequence of decision rules $\delta_n$,
\[
\sup_I\liminf_{n\to\infty}\sup_{h\in I}E_{Q_{n,h}}\left[\ell\left(\sqrt{n}(\delta_n(X^n)-\kappa(\theta_{n,h}))\right)\right]\ge\int\ell\,dN(0,KI_0^{-1}K'),
\]
where the supremum ranges over finite subsets $I$ of $H=\mathbb{R}^p$.
\end{prop}

The right-hand side is the minimax risk in the Gaussian limit experiment where one observes $X\sim N(h,I_0^{-1})$ and wishes to estimate $Kh$ under loss $\ell$.  The bound says no sequence of estimators can achieve lower worst-case local asymptotic risk than the finite-sample risk in this Gaussian problem.  Under squared error loss, the bound reduces to $KI_0^{-1}K'$, achieved by any efficient estimator, including the maximum likelihood estimator.  The classical LAM theorem thus provides a theoretical foundation for familiar efficiency claims.

\paragraph{Example: Average Treatment Effects, Continued}

Under the model described in Section \ref{sec: loss axioms}, the family $\{Q_{n,\theta}\}$ is LAN at any interior $\theta_0=(\mu_0,\mu_1)$ with Fisher information $I_0=\frac{1}{2}\operatorname{diag}(1/\sigma_0^2,1/\sigma_1^2)$ for $\sigma_d^2=\mu_d(1-\mu_d)$.  Since $K=(-1,1)$, the classical LAM bound under squared error loss is $KI_0^{-1}K'=2(\sigma_0^2+\sigma_1^2)$, achieved by the difference-in-means estimator. $\triangle$

\subsection{The Constrained Multiplier LAM Theorem}

We now develop an analogous result for constrained multiplier preferences.  The key additional ingredient is a moment function $\psi:\Theta\times\mathcal{X}_0\to\mathbb{R}^k$ satisfying $E_{Q_{n,\theta}}[\psi(\theta,X_i)]=0$ for all $n$ and $\theta$, where $\mathcal{X}_0$ denotes the sample space for a single observation $X_i$.  This moment function encodes the researcher's beliefs about which aspects of the model are correctly specified.  We assume that the researcher believes misspecification does not affect the first $M$ moments of the scaled sample average of $\psi$ evaluated at the true value of $\theta$.  Including higher moments is potentially important, since the forms of misspecification allowed by constrained multiplier preferences include ones which change the dependence structure of the data.  Thus, even if misspecification does not affect the marginal distribution of each observation, the distribution of sample averages could still change.  Constraining the moments of the sample average up to order $M$ restricts such possibilities.  We emphasize, however, that the moment conditions $\psi$ need not point-identify $\theta$ (e.g. the dimension of $\psi$ might be lower than that of $\theta$) so our results apply even in settings where there are too few moment functions to allow moment-based estimation.

Formally, let
\[
Y_{n,h}=\frac{1}{\sqrt{n}}\sum_{i=1}^n\psi(\theta_{n,h},X_i),
\]
and for $m=(m_1,\ldots,m_k)\in\mathbb{N}_0^k$ with $1\le \sum_{s=1}^k m_s\le M$, define
$\tilde W^m_{n,h}=\prod_{s=1}^k Y_{n,h,s}^{m_s}.$
Let $\tilde W_{M,n,h}$ collect $\tilde W^m_{n,h}$ over all such $m$, and define
\[
W_{M,n,h}=\tilde W_{M,n,h}-E_{Q_{n,h}}[\tilde W_{M,n,h}].
\]
The constraint set
\[
\mathcal{P}^M_{n,h}=\left\{P\in\Delta(\mathcal{X}_0^n):E_P[W_{M,n,h}]=0\right\}
\]
consists of all data distributions that preserve the first $M$ moments of $Y_{n,h}$.  This takes the form assumed in Proposition \ref{prop:constrained-dual init} with $\varphi(\theta_{n,h},X^n)=W_{M,n,h}$, so the duality result applies.

\paragraph{Example: Average Treatment Effects, Continued}
Returning to the two-team ATE example, recall the per-observation moment function \eqref{eq: psi def}.
The constraint set $\mathcal{P}^M_{n,h}$ enforces that misspecification not affect the first $M$ moments of $Y_{n,h}=\frac{1}{\sqrt{n}}\sum_i\psi(\theta_{n,h},X_i)\in\mathbb{R}^4$.  For $M=1$, as discussed above this constrains the mean of $Y_{n,h}$, requiring that $\sum_{i=1}^nE_P[\psi(\theta_{n,h},X_i)]=0$. For $M\ge 2$, the constraints additionally involve moments of $Y_{n,h}$ that depend on pairwise covariances of $\psi(\theta,X_i)$ across observations, restricting the impact of misspecification of the cross-observation dependence within team 1's data. $\triangle$

To derive our LAM theorem we impose two assumptions.  The first collects regularity conditions on the model and the moment function.

\begin{assu}\label{assu: moments}
The family $\{Q_{n,\theta}:\theta\in\Theta\}$ is stationary for each $n$ with
\[
E_{Q_{n,\theta}}[\psi(\theta,X_i)]=0 \quad \text{for all }\theta\in\Theta,
\qquad
E_{Q_{n,\theta_0}}\left[\frac{\partial}{\partial\theta'}\psi(\theta_0,X_i)\right]=\Psi.
\]
Moreover,  $\{Q_{n,\theta}:\theta\in\Theta\}$ has densities $\{q_{n,\theta}:\theta\in\Theta\}$, where for
\[
S_n=\frac{1}{\sqrt{n}}\frac{\partial}{\partial\theta}\log q_{n,\theta_0}(X^n),
\qquad
Y_{n,0}=\frac{1}{\sqrt{n}}\sum_{i=1}^n\psi(\theta_0,X_i),
\]
\[
\left(
\begin{array}{c}
S_n\\
Y_{n,0}
\end{array}
\right)
\stackrel[d]{Q_{n,\theta_0}}{\to}
N\left(
\left(
\begin{array}{c}
0\\
0
\end{array}
\right),
\Sigma\right),~~\Sigma=\left(
\begin{array}{cc}
I_0 & -\Psi'\\
-\Psi & \Omega
\end{array}
\right)
\]
where $\Omega$ has full rank.  Finally, $\psi(\theta,X_i)$ and $q_{n,\theta}(X^n)$ are differentiable at $\theta_0$ for all $X_i$ and $X^n$, and there exist an open neighborhood $\mathcal{N}$ of $\theta_0$ and a constant $C<\infty$ such that
\[
E_{Q_{n,\theta_0}}\left[\sup_{\theta\in\mathcal{N}}\left(\left\|\frac{\partial}{\partial\theta}\psi(\theta,X_i)\right\|\right)\right]\le C
\]
for all $n$.
\end{assu}

These conditions are standard: stationarity, a central limit theorem for the scaled moments, and smoothness of the moment function.
We also restrict the loss.
\begin{assu}\label{assu: loss}
The loss function in the sample of size $n$ is equal to $l_n(a,\theta)=\ell(\sqrt{n}(a-\kappa(\theta)))$, where $\kappa:\Theta\to\mathbb{R}^d$ is continuously differentiable at $\theta_0$ and $\ell:\mathbb{R}^d\to[0,\infty)$ is convex, finite-valued, and satisfies $\ell(u)\to\infty$ as $\|u\|\to\infty$.
\end{assu}

Relative to the classical LAM theorem, Assumption \ref{assu: loss} strengthens quasi-convexity to full convexity, but allows asymmetric loss.  Convexity ensures that randomized estimators cannot improve on deterministic ones, which simplifies our asymptotic results.  Note that while the results of Section \ref{sec: loss axioms} consider the case of bounded loss to simplify the axiomatic derivations, here we apply the resulting preferences to unbounded loss functions.

By Proposition \ref{prop:constrained-dual init}, in the sample of size ${n}$ the worst-case constrained multiplier risk of an estimator $\delta_n$ over a set of local parameters $I \subset H$ is
\[
\sup_{h\in I}\sup_{P\in\mathcal{P}^M_{n,h}}\left\{\mathbb{E}_P\left[l_n(\delta_n(X^n),\theta_{n,h})\right]-\lambda \KL(P\|Q_{n,h})\right\}=
\]
\[
\sup_{h\in I}\inf_\beta\lambda\cdot\log\left(E_{Q_{n,h}}\left[\ell^*\left(\sqrt{n}\left(\delta_n(X^n)-\kappa(\theta_{n,h})\right)\right)\exp\left(\beta'W_{M,n,h}\right)\right]\right)
\]
for $\ell^*(u)=\exp\left(\frac{1}{\lambda}\ell(u)\right)$.  
Under the conditions above, if we consider the liminf as $n\to\infty$ and take the worst case over $I$, 
we obtain the following local asymptotic minimax bound.

\begin{thm}\label{thm: Constrained LAM}
Assume the model $\{Q_{n,\theta}:\theta\in\Theta\}$ is locally asymptotically normal at $\theta_0$ with scaling coefficient $\sqrt{n}$ and nonsingular $I_0$, that $\theta_0\in\mathrm{int}(\Theta)$, and that for all $h\in\mathbb{R}^p$, the moments of $Y_{n,h}=\frac{1}{\sqrt{n}}\sum_i\psi(\theta_{n,h},X_i)$ up to order $M$ converge to the corresponding moments of $\xi\sim N(0,\Omega)$,
\[
E_{Q_{n,h}}\left[(v'Y_{n,h})^m\right]\to E\left[(v'\xi)^m\right]\quad\text{for all }v\in\mathbb{R}^k\text{ and all }m\in\{0,\ldots,M\}.
\]
Then under Assumptions \ref{assu: moments} and \ref{assu: loss}, for any sequence of decision rules $\delta_n$,
\[
\sup_I\liminf_{n\to\infty}\sup_{h\in I}\sup_{P\in\mathcal{P}^M_{n,h}}\left\{\mathbb{E}_P\left[l_n(\delta_n(X^n),\theta_{n,h})\right]-\lambda \KL(P\|Q_{n,h})\right\}\ge
\]
\[
\inf_\delta\sup_{h\in\mathbb{R}^p}\inf_{\beta\in\mathbb{R}^{b}}\lambda\cdot\log\left(E_{Q_h}\left[\ell^*\left(\delta(X,Y)-Kh\right)\exp\left(\beta'W_{M,h}\right)\right]\right),
\]
where the supremum on the left-hand side ranges over finite subsets $I\subset\mathbb{R}^p$, $K=\frac{\partial}{\partial\theta'}\kappa(\theta_0)\in\mathbb{R}^{d\times p}$, $Q_h$ is given by
\begin{equation}\label{eq: Limit Experiment}
\left(\begin{array}{c}
X\\
Y
\end{array}\right)\sim N\left(\left(\begin{array}{c}
h\\
-\Psi h
\end{array}\right),\left(\begin{array}{cc}
I_0^{-1} & -I_0^{-1}\Psi'\\
-\Psi I_0^{-1} & \Omega
\end{array}\right)\right),
\end{equation}
and $W_{M,h}$ collects the centered moments of $Y_h=Y+\Psi h$ up to order $M$,
\begin{equation}\label{eq: W_{M,h}}
W_{M,h}=\left(\prod_{s=1}^k Y_{h,s}^{m_s}-E\left[\prod_{s=1}^k\xi_s^{m_s}\right]:m\in\mathbb{N}_0^k,\ 1\le \sum_{s=1}^k m_s\le M\right).
\end{equation}
\end{thm}

Like the classical LAM theorem, Theorem \ref{thm: Constrained LAM} shows that the local asymptotic risk, now considering the constrained multiplier risk, is lower bounded by the risk in a Gaussian limit experiment.  The limit experiment now involves two statistics: $X$, which plays the same role as in the standard LAM theorem and corresponds to the asymptotic analog of the maximum likelihood estimator, and $Y$, which is the limit of the scaled sample average $Y_{n,0}$ and captures the information in the moment conditions.  The finite-sample constraint set $\mathcal{P}^M_{n,h}$, which requires that the first $M$ moments of $Y_{n,h}\approx Y_{n,0}+\Psi h$ be preserved, maps to the constraint $E_P[W_{M,h}]=0$ in the limit experiment.

Our use of a fixed, sample-size independent $\lambda$ in Theorem \ref{thm: Constrained LAM} is important for the result.\footnote{Formally, our use of a fixed $\lambda$ corresponds to using a single, sample-size independent constrained multiplier preference with state space $H\times\mathcal{X}$ where $\mathcal{X}=\mathcal{X}_0^\infty$.  Our asymptotic results then concern choice from the sequence of loss menus $\mathcal{L}_n=\{l_n(\delta_n(X),\theta_{n,h}):h\in H,X\in\mathcal{X},\delta_n\in\mathcal{D}_n\}$ where the rules in $\mathcal{D}_n$ are constrained to depend on $X=(X_1,X_2,...)$ only through the first $n$ observations.  Work in progress by Ricky Li explores the axiomatic implications of the other aspects of the LAM criterion, in particular the focus on the liminf and the use of finite index sets $I$.}  Note, in particular, that if $Q_{n,\theta}$ and $P$ are both product measures (corresponding to i.i.d. sampling), then $\KL(P\|Q_{n,\theta})$ is equal to $n$ times the KL divergence for a single draw from each distribution.  Thus, to attain a fixed KL divergence as $n\to\infty,$ the per-observation KL divergence must shrink at a $\frac{1}{n}$ rate.  For many decision rules $\delta_n$ the worst-case distributions $P$ which attain or approximate the constrained multiplier risk will not be i.i.d., but as the above calculation suggests they will nevertheless often correspond to ``local'' misspecification.  Thus, while the preferences we consider do not directly impose an assumption of local misspecification, a focus on the locally misspecified cases emerges naturally.

\paragraph{Example: Average Treatment Effects, Continued}

In the limit experiment \eqref{eq: Limit Experiment}, $X\in\mathbb{R}^2$ corresponds to the maximum likelihood estimator for $(\mu_0,\mu_1)$ pooling data from both teams, while $Y\in\mathbb{R}^4$ captures the team-1 specific moment information.  The matrices governing the limit experiment are
\[I_0^{-1}=2\operatorname{diag}(\sigma_0^2,\sigma_1^2),\quad \Psi=\begin{pmatrix}-\pi_1/2 & 0\\ 0 & -\pi_1/2 \\ 0 & 0\\ 0 & 0\end{pmatrix},\quad \Omega=\operatorname{diag}\left(\frac{\pi_1\sigma_0^2}{2},\,\frac{\pi_1\sigma_1^2}{2},\,\frac{\pi_1}{4},\pi_1(1-\pi_1)\right).\]
Note that the third and fourth rows of $\Psi$ are zero, so the corresponding elements of  $Y$ are uninformative on their own, but 
narrow the (asymptotic analog of the) ambiguity set. $\triangle$

When $M=0$ there are no moment constraints, so the statistic $Y$ plays no role: the adversary is free to distort its distribution.  If $\ell$ is additionally symmetric around zero, the result reduces to the classical LAM theorem applied to the loss $\ell^*$.  Specifically, in this case $\ell^*(u)=\exp(\ell(u)/\lambda)$ is symmetric and (quasi-)convex, so the classical LAM theorem (Proposition \ref{prop: LAM}) applies for each value of $\lambda$.
More generally, the choice of $M$ reflects what restrictions the researcher places on the forms of misspecification they consider.  When the researcher is uncertain what value of $M$ to impose, the $M\to\infty$ limit provides a natural benchmark.

\begin{cor}\label{cor: infinite case}
If the assumptions of Theorem \ref{thm: Constrained LAM} hold for all $M\in\mathbb{N}$, then
\[
\inf_M\sup_I\liminf_{n\to\infty}\sup_{h\in I}\sup_{P\in\mathcal{P}^M_{n,h}}\left\{\mathbb{E}_P\left[l_n(\delta_n(X^n),\theta_{n,h})\right]-\lambda \KL(P\|Q_{n,h})\right\}\ge
\]
\[
\inf_\delta\sup_{h\in\mathbb{R}^p}\lambda\cdot E_{Q_{Y,h}}\left[\log\left(E_{Q_{X|Y,h}}\left[\ell^*(\delta(X,Y)-Kh)\mid Y\right]\right)\right]
\]
where $Q_h$ is as in \eqref{eq: Limit Experiment}.
\end{cor}

Taking $M\to\infty$ requires that misspecification not distort any moment of $Y_{n,h}$.  Since the normal distribution is determined by its moments, this forces $Y$ to remain normally distributed in the limit.  The risk bound then simplifies because the KL divergence decomposes: the adversary can distort the conditional distribution of $X$ given $Y$, but not the marginal of $Y$.  The resulting bound has an intuitive form, with an inner conditional expectation of the exponentiated loss over $X\mid Y$, inside a logarithm, integrated over the marginal of $Y$.

Theorem \ref{thm: Constrained LAM} and Corollary \ref{cor: infinite case} provide lower bounds on the risk achievable by any sequence of estimators under constrained multiplier preferences.  In the next section, we characterize estimators which attain these bounds.

\section{Optimal Decision Rules}\label{sec: optimal rules}

To derive optimal estimators, we begin by exploiting the invariance structure of the limit experiment to show that we can limit attention to (asymptotically) equivariant decision rules.  We then characterize optimal rules for important special cases and show that their finite-sample analogs, based on plugging the MLE and moments into the limit-experiment optimal rule, are asymptotically optimal.  Practically, these results show how researchers who believe certain implications of their model (captured by $\psi$ and $M$) and have a given degree of misspecification concern (captured by $\lambda$) can optimally combine likelihood and moment information.  One implication of our results is that the widespread current practice (documented and discussed in e.g. \citealt{andrews2025purpose}) of focusing solely on moment-based estimation in settings where misspecification is a concern is unlikely to be optimal outside cases of extreme concern (i.e. $\lambda\to 0$).

\subsection{The Hunt-Stein Theorem}\label{sec: hunt-stein}

The limit experiment \eqref{eq: Limit Experiment} exhibits an important invariance structure.  Consider the group $G=\mathbb{R}^p$ acting on the sample space by $g\circ(X,Y)=(X+g,Y-\Psi g)$, on the action space $\mathbb{R}^d$ by $g\circ a=a+Kg$, and on the parameter space by $g\circ h=h+g$. These transformations leave the loss unchanged: $\ell((g\circ a)-K(g\circ h))=\ell(a-Kh)$ for all $a,h,g$.  Following \citet{LehmannCasella1998}, we say that a decision rule $\delta$ is \emph{equivariant} if $\delta(X+g,Y-\Psi g)=\delta(X,Y)+Kg$ for all $g\in\mathbb{R}^p$.  Let $\mathcal{D}^E$ denote the class of equivariant decision rules in the limit experiment.

While the constrained multiplier objective is nonstandard, we extend the classical Hunt-Stein theorem to show that for minimax purposes, it is without loss to limit attention to equivariant decision rules.

\begin{thm}\label{thm: equivariant}
Under Assumption \ref{assu: loss}, for any estimator $\delta$ in the limit problem there exists an equivariant estimator $\delta^E\in\mathcal{D}^E$ such that
\begin{align*}
\sup_{h\in\mathbb{R}^p,\, P\in\mathcal{P}_{h}^M}\left\{E_{P}\left[\ell\left(\delta\left(X,Y\right)-Kh\right)\right]-\lambda \KL(P\|Q_h)\right\}\ge\\
\sup_{h\in\mathbb{R}^p,\, P\in\mathcal{P}_{h}^M}\left\{E_{P}\left[\ell\left(\delta^E\left(X,Y\right)-Kh\right)\right]-\lambda \KL(P\|Q_h)\right\}.
\end{align*}
Consequently, to derive a minimax decision rule it is without loss to restrict attention to the class of equivariant rules.
\end{thm}

A useful consequence of equivariance is that the constrained multiplier risk does not depend on $h$.  Specifically, for any $\delta^E\in\mathcal{D}^E$, transitivity of the group action implies
\[
\sup_{P\in\mathcal{P}_{h}^M}\left\{E_{P}\left[\ell\left(\delta^E\left(X,Y\right)-Kh\right)\right]-\lambda \KL(P\|Q_h)\right\}=\sup_{P\in\mathcal{P}_{0}^M}\left\{E_{P}\left[\ell\left(\delta^E\left(X,Y\right)\right)\right]-\lambda \KL(P\|Q_0)\right\}.
\]
The minimax problem over equivariant rules thus reduces to
\[
\inf_{\delta^E\in\mathcal{D}^E}\sup_{P\in\mathcal{P}_{0}^M}\left\{E_{P}\left[\ell\left(\delta^E\left(X,Y\right)\right)\right]-\lambda \KL(P\|Q_0)\right\}.
\]

\paragraph{Example: Average Treatment Effects, Continued}

In the two-team ATE example, the group $G=\mathbb{R}^2$ shifts $(X,Y)$ by $(g,-\Psi g)$ and shifts the action by $Kg=(-1,1)g$.  Equivariance requires that the estimator respond to a location shift in the data $(X,Y)$ by a parallel shift in the estimated treatment effect.  The MLE $KX=(-1,1)X$ is equivariant, as is the (limit experiment analog of the) efficient GMM estimator $-K(\Psi'\Omega^{-1}\Psi)^{-1}\Psi'\Omega^{-1}Y$. $\triangle$

\medskip

\subsection{Optimal Equivariant Rules}\label{sec: optimal rules body}

We next characterize optimal equivariant rules, treating the $M<\infty$ and $M=\infty$ cases in turn.
By the duality of Proposition \ref{prop:constrained-dual init} and invariance, the minimax problem for equivariant rules with $M<\infty$ becomes
\begin{equation}\label{eq: joint opt}
\inf_{\delta^E\in\mathcal{D}^E}\inf_{\beta\in\mathbb{R}^b}E_{Q_0}\left[\ell^*\left(\delta^E\left(X,Y\right)\right)\exp\left(\beta'W_{M,0}\right)\right].
\end{equation}
Since this is a joint infimum over $(\delta^E,\beta)$, the order of minimization does not matter.  For each fixed $\beta$, however,  the minimization over $\delta^E$ corresponds to finding the best equivariant decision rule under an exponentially tilted likelihood, since
\[
\inf_{\delta^E\in\mathcal{D}^E}E_{Q_0}\left[\ell^*\left(\delta^E\left(X,Y\right)\right)\exp\left(\beta'W_{M,0}\right)\right]=\]
\[
\inf_{\delta^E\in\mathcal{D}^E}E_{Q_0}\left[\ell^*\left(\delta^E\left(X,Y\right)\right)\frac{\exp\left(\beta'W_{M,0}\right)}{E_{Q_0}\left[\exp\left(\beta'W_{M,0}\right)\right]}\right]E_{Q_0}\left[\exp\left(\beta'W_{M,0}\right)\right].
\]
Theorem 6.5 of \citet{Eaton1989} proves that the best equivariant estimator in a location problem is the Bayes decision rule under the flat prior, from which we immediately obtain the form of the optimal rule for our problem.

\begin{prop}\label{prop: best equivariant finite M}
Under Assumption \ref{assu: loss}, for each $\beta\in\mathbb{R}^b$ define
\[
\delta^*_\beta(X,Y)\in\argmin_{a}\int\ell^*(a-Kh)\,\pi_\beta(h\mid X,Y)\,dh,
\]
where $\pi_\beta(h\mid X,Y)$ is the posterior density under a flat prior on $h$,
\[
\pi_\beta(h\mid X,Y)=\frac{q_0(X-h,\,Y+\Psi h)\cdot\exp\left(\beta'W_{M,h}\right)}{\int q_0(X-h',\,Y+\Psi h')\cdot\exp\left(\beta'W_{M,h'}\right)dh'},
\]
and $q_0$ denotes the density of $(X,Y)$ under $Q_0$.  Then:
\begin{enumerate}
\item[(a)] $\delta^*_\beta$ minimizes $E_{Q_0}[\ell^*(\delta^E(X,Y))\exp(\beta'W_{M,0})]$ over $\delta^E\in\mathcal{D}^E$.
\item[(b)] There exists $\beta^*$ minimizing
\[
E_{Q_0}\left[\ell^*\left(\delta^*_\beta\left(X,Y\right)\right)\exp\left(\beta'W_{M,0}\right)\right],
\]
and moreover 
\item[(c)] $\delta^*_{\beta^*}$ is optimal in the limit experiment.
\end{enumerate}
\end{prop}

The previous result provides the form of the optimal estimator for the case with $M<\infty$.  In the case of $M=\infty,$ the risk of an equivariant rule takes the form
\begin{equation}\label{eq: M infinity risk}
R_\infty(\delta^E)=\lambda\cdot E_{Q_{Y,0}}\left[\log\left(E_{Q_{X|Y,0}}\left[\ell^*(\delta^E(X,Y))\mid Y\right]\right)\right].
\end{equation}
We have not found a closed-form characterization of the optimal rule in the $M=\infty$ case under general loss.  However, one can show that the risk \eqref{eq: M infinity risk} is convex in $\delta^E$, so if we restrict to a linear class of rules $\delta_\Gamma(X,Y)=\Gamma\phi(X,Y)$ for a finite-dimensional vector of basis functions $\phi(X,Y)$, the problem of finding the optimal rule in the class is likewise convex.

\begin{prop}\label{prop: convexity}
Under Assumption \ref{assu: loss}, let $\phi(X,Y)\in\mathbb{R}^J$ and let $\Gamma\in\mathbb{R}^{d\times J}$, so that $\delta_\Gamma(X,Y)=\Gamma\phi(X,Y)\in\mathbb{R}^d$.
\begin{enumerate}
\item[(a)] For $M\in\mathbb{N}$, the objective $E_{Q_0}[\ell^*(\Gamma\phi(X,Y))\exp(\beta'W_{M,0})]$ is jointly convex in $(\Gamma,\beta)$.
\item[(b)] For $M=\infty$, the objective $\lambda\cdot E_{Q_{Y,0}}[\log(E_{Q_{X|Y,0}}[\ell^*(\Gamma\phi(X,Y))\mid Y])]$ is convex in $\Gamma$.
\end{enumerate}
\end{prop}
Using this convexity, it is straightforward to solve numerically for optimal rules in a given linear class provided one can tractably compute expectations under $Q_0$.
For the important special case of squared error loss, however, we are able to go further and exactly characterize the optimal decision rule for all $M$.
\begin{prop}\label{prop: optim estimator square loss}
Under Assumption \ref{assu: loss} with $\ell(u)=\norm{u}^2$:
\begin{enumerate}
\item[(a)] If $M=0$ or $M=1$, the optimal equivariant estimator is $\delta^*(X,Y)=KX$.
\item[(b)] If $M\ge 2$ (including $M=\infty$), the optimal equivariant estimator takes the form $\delta^*(X,Y)=KX+C^*Z^I$ for $Z^I=Y+\Psi X$, where
\[
C^*\in\argmin_{C}
E_{Q_{Y,0}}\left[\log E_{Q_{X|Y,0}}\left[\exp\left(\frac{1}{\lambda}\left\|K(X+CZ^I)\right\|^2\right)\mid Y\right]\right].
\]
Moreover, the optimal risk is the same for all $M\ge 2$.
\end{enumerate}
\end{prop}

Part (b) implies that increasing the number of moment constraints $M$ beyond two does not improve the optimal risk under squared error loss.  Intuitively, we show that the best equivariant decision rule for $M=\infty$ is linear in $(X,Y)$, but for such rules the worst-case distribution $P^*$ is Gaussian, and a Gaussian distribution that matches the first two moments of $Y$ automatically matches all higher moments as well.

\paragraph{Example: Average Treatment Effects, Continued}

In the two-team ATE example, part (a) says the optimal rule when $M\le 1$ is $KX$, the MLE pooling data from both teams.  Under $M\ge 2$, the researcher can exploit the team-1 moment conditions: the optimal rule linearly adjusts the MLE using $Z^I=Y+\Psi X$, which is the asymptotic analog to the moment conditions evaluated at the MLE.  The optimal adjustment matrix $C^*$ can be computed numerically. $\triangle$

\medskip

\subsection{Feasible, Asymptotically Optimal Rules}\label{sec: feasible rules}

The results above characterize optimal decision rules in the limit experiment.  In practice, however, the researcher observes only finite-sample data.  We now show that plug-in finite-sample analogs of limit-experiment rules converge in both distribution and, under an integrability condition, risk.

The results in this subsection are local to the fixed base value $\theta_0$ used to define the limit experiment.  Accordingly, the limit-experiment rule may depend on the derivative
$K=\frac{\partial}{\partial\theta'}\kappa(\theta_0)\in\mathbb{R}^{d\times p}$
as well as on $\Sigma$.
For a continuous decision rule $\delta^c(X,Y;K,\Sigma)$ in the limit experiment, equivariant in $(X,Y)$ for each fixed $(K,\Sigma)$, define a plug-in finite-sample analog as
\[
\delta_n^c=\kappa(\hat\theta_n^{MLE})+\frac{1}{\sqrt{n}}\delta^c\left(0,\;\frac{1}{\sqrt{n}}\sum_{i=1}^n\psi(\hat\theta_n^{MLE},X_i);\;\hat K_n,\hat\Sigma_n\right),
\]
where $\hat\theta_n^{MLE}$ is the maximum likelihood estimator, $\hat K_n$ is a consistent estimator of $K$ (e.g. the derivative of $\kappa$ at $\hat\theta_n^{MLE}$), and $\hat\Sigma_n$ is a consistent estimator of $\Sigma$.

\begin{assu}\label{assu: smoothness}
For $\psi_j$ the $j$-th component of $\psi$, there exist $\eta>0$ and $\tau_{j,k,l}(X)$ such that
\[
\sup_{\theta\in\mathcal{N}}\left|\frac{\partial^2}{\partial\theta_k\partial\theta_l}\psi_j(\theta,X)\right|\le \tau_{j,k,l}(X)
\]
for all $j,k,l$, where $\mathcal{N}$ is a neighborhood of $\theta_0$ and
$E_{Q_{n,\theta_0}}[\tau_{j,k,l}(X_i)^{1+\eta}]<C<\infty$.
Moreover, for every $h\in\mathbb{R}^p$,
\[
\left(
\sqrt{n}(\hat\theta_n^{MLE}-\theta_{n,h}),\;
\frac{1}{\sqrt{n}}\sum_{i=1}^n\psi(\hat\theta_n^{MLE},X_i),\;
\hat K_n,\;
\hat\Sigma_n
\right)\stackrel[d]{Q_{n,h}}{\to}
\left(
X-h,\;
Y+\Psi X,\;
K,\;
\Sigma
\right).
\]
\end{assu}

\begin{prop}\label{prop: Convergence - Estimation Error}
Under Assumptions \ref{assu: moments}-\ref{assu: smoothness} and the conditions of Theorem \ref{thm: Constrained LAM}, let $\delta^c(X,Y;K,\Sigma)$ be continuous in $(X,Y,K,\Sigma)$ and equivariant in $(X,Y)$ for each fixed $(K,\Sigma)$.
Then for any $h\in\mathbb{R}^p$,
\[
\sqrt{n}\left(\delta_n^c-\kappa(\theta_{n,h})\right)\stackrel[d]{Q_{n,h}}{\to}\delta^c(X,Y;K,\Sigma)-Kh.
\]
\end{prop}

To obtain convergence of the (dual) constrained multiplier objective, we strengthen convergence in distribution to convergence of moments.

\begin{assu}\label{assu: UI}
For a continuous equivariant decision rule $\delta^c(X,Y;K,\Sigma)$ in the limit experiment, let
\[
\beta^{*,c}\in\argmin_{\beta\in\mathbb{R}^b}
E_{Q_0}\left[\ell^*\left(\delta^c(X,Y;K,\Sigma)\right)\exp\left(\beta'W_{M,0}\right)\right].
\]
For each finite $h\in\mathbb{R}^p$,
\[
E_{Q_{n,h}}\left[\ell^*\left(\sqrt{n}\left(\delta_n^c-\kappa(\theta_{n,h})\right)\right)\exp\left(\beta^{*,c\prime}W_{M,n,h}\right)\right]
\to\]
\[
E_{Q_h}\left[\ell^*\left(\delta^c(X,Y;K,\Sigma)-Kh\right)\exp\left(\beta^{*,c\prime}W_{M,h}\right)\right].
\]
\end{assu}

\begin{cor}\label{corr: Attainability}
Under Assumptions \ref{assu: moments}-\ref{assu: UI} and the conditions of Theorem \ref{thm: Constrained LAM}, for a continuous, equivariant decision rule $\delta^c(X,Y;K,\Sigma)$ in the limit experiment and $\delta_n^c$ its plug-in analog,
\[
\lim_{n\to\infty}\sup_{h\in I}\inf_\beta\lambda\cdot\log\left(E_{Q_{n,h}}\left[\ell^*\left(\sqrt{n}(\delta_n^c-\kappa(\theta_{n,h}))\right)\exp(\beta'W_{M,n,h})\right]\right)=
\]
\[
\sup_{h\in I}\inf_\beta\lambda\cdot\log\left(E_{Q_h}\left[\ell^*\left(\delta^c(X,Y;K,\Sigma)-Kh\right)\exp(\beta'W_{M,h})\right]\right)
\]
for any finite $I\subset\mathbb{R}^p$.  In particular, if $\delta^c(\cdot,\cdot;K,\Sigma)$ is an optimal rule for the limit experiment corresponding to the fixed base point $\theta_0$, then the plug-in estimator $\delta_n^c$ is locally asymptotically minimax at $\theta_0$:
\[
\sup_I\liminf_{n\to\infty}\sup_{h\in I}\sup_{P\in\mathcal{P}^M_{n,h}}\left\{\mathbb{E}_P\left[l_n(\delta_n^c(X^n),\theta_{n,h})\right]-\lambda \KL(P\|Q_{n,h})\right\}=
\]
\[
\inf_\delta\sup_{h\in\mathbb{R}^p}\inf_\beta\lambda\cdot\log\left(E_{Q_h}\left[\ell^*\left(\delta(X,Y)-Kh\right)\exp(\beta'W_{M,h})\right]\right).
\]
\end{cor}

\section{Adaptive Decision Rules}\label{sec: adaptive rules}

While the results in Section \ref{sec: optimal rules} characterize optimal decision rules under constrained multiplier risk, the resulting rules depend on the misspecification-aversion parameter $\lambda$.  While dependence of optimal rules on $\lambda$ is natural from a theoretical perspective (after all, a researcher entirely unconcerned with misspecification already knows the MLE is optimal), from a practical perspective it introduces a free parameter that a researcher interested in applying our methods must choose.  In this section, we follow \cite{ArmstrongKlineSun2025} and examine, for a special case of our setting, whether there exist simple estimators that perform reasonably for many different values of $\lambda$.

In particular, we consider the special case of squared error loss $\ell=\|u\|^2$ where both the parameter $\theta$ and the moment condition $\psi$ are scalar, where we normalize $I_0=1$ and $\Psi=-1$, so the limit problem is fully described by $\Omega=\Var(Y)$.\footnote{As in \cite{ArmstrongKlineSun2025}, one could more broadly interpret this setting as the limit problem when we are interested in a scalar target parameter and would like to combine an efficient but potentially misspecified estimator (represented by $X$) and a less efficient but more robust estimator (represented by $Y$).  The distinction between our analysis in this section and that of \cite{ArmstrongKlineSun2025} then stems from their focus on a version of constraint preferences, vs. ours on constrained multiplier preferences.} We take $M=\infty,$ corresponding to the case where the researcher thinks misspecification does not affect the moments at all.  Thus, the sole remaining decision for the researcher is the (unavoidable) choice of which moment function $\psi$ they think remains valid under misspecification.

To search for procedures which perform well across different values of $\lambda$, let $R_{\infty}^{\lambda}(\delta)$ denote the constrained multiplier risk of $\delta$ under misspecification-aversion parameter $\lambda,$
\[
R_{\infty}^{\lambda}(\delta)=\sup_{h\in\mathbb{R},\, P\in\mathcal{P}_{h}^\infty}\left\{E_{P}\left[\ell\left(\delta\left(X,Y\right)-Kh\right)\right]-\lambda \KL(P\|Q_h)\right\}.
\]
Following \cite{ArmstrongKlineSun2025} we consider the \emph{adaptive regret}
criterion \[
A_{\infty}(\delta) = \sup_{\lambda \in \Lambda} \frac{R_{\infty}^{\lambda}(\delta)}{\min_{\tilde{\delta}}R_{\infty}^{\lambda}(\tilde{\delta})},
\] 
which compares the performance of the rule $\delta,$ across values of $\lambda,$ to the performance of the family of $\lambda$-by-$\lambda$ optimal rules. By construction $A_{\infty}(\delta)\ge1,$ where if the adaptive regret is close to this lower bound it tells us that the rule $\delta$ is ``nearly'' optimal in a proportional sense uniformly across $\lambda$ values, and thus that if we opt to use $\delta$ rather than taking a stand on the ``correct'' $\lambda$ and using the resulting rule, the price of doing so (measured as the proportional increase in constrained multiplier risk) is not very high.

By Proposition \ref{prop: optim estimator square loss}, we know that to solve for the $\lambda$-specific minimized risk $\min_{\tilde{\delta}}R_{\infty}^{\lambda}(\tilde{\delta}),$ it suffices to consider linear equivariant rules, greatly facilitating computation.  Analogously, in the broader search for adaptive rules, we limit attention to equivariant rules $\delta^E\in\mathcal{D}^E$.  Let us normalize $I_0=1$ and $\Psi=-1,$ and again define $Z^I=Y+\Psi X$ as the limit experiment analog of the sample average moments evaluated at the MLE, or (equivalently) as the difference between the GMM and ML estimators.  One can show that $Z^I$ is a maximal invariant in the limit experiment, and thus that any equivariant decision rule can be written as 
\[
\delta^E(X,Y)=X+\gamma(Z^I).
\]
Thus, the problem of finding an equivariant estimator with a small adaptive risk is equivalent to picking a ``good'' $\gamma$. 

We consider two simple parameterizations of $\gamma(\cdot)$ with a single tuning parameter, both of which \cite{ArmstrongKlineSun2025} find perform well according to their criterion. The first is the soft-thresholding estimator with \[
\gamma_{ST,\tau}(Z^I) = \max\left\{|Z^I| - \tau, 0\right\}\mathrm{sgn}(Z^I),
\]
which is closely related to the LASSO estimator \citep{Tibshirani1996}.
The second is the adaptive empirical risk minimization (ERM) estimator \citep{Magnus2002,deChaisemartin2020} with 
\[
\gamma_{ERM, \tau}(Z^I) = \frac{(Z^I)^2}{(Z^I)^2+\tau} Z^I.
\]
In each class, we choose the tuning parameter $\tau$ to minimize the adaptive risk $A_\infty$.
To benchmark the performance of these simple parametric classes, we compare their adaptive risk to that of a flexible specification of $\gamma.$  In particular, motivated by Proposition \ref{prop: convexity}, we parameterize $\gamma$ as a linear spline (flattened for the extreme values of $Z^I$) and solve numerically for the optimal parameters via convex optimization.\footnote{Specifically, the maximization which defines $A_\infty$ preserves the convexity established by Proposition \ref{prop: convexity}, and we approximate the maximum over $\lambda$ by maximization over a finite, evenly spaced grid of $\log(\lambda)$ values from [-3,6].}  

Figure \ref{fig:adaptive-Omega2}(a) shows the resulting constrained multiplier risk functions for the $\lambda$-by-$\lambda$ optimal estimator, the adaptive linear spline estimator, and the two simple estimators, all for the case of $\Omega=2$.  As expected, for small $\lambda$ the optimal risk is close to that of GMM (i.e. $\Omega=2$) while for large $\lambda$ it is close to that of MLE (i.e. $I_0=1$).  In between, we see that the linear spline estimator has better performance than the simple estimators over an intermediate range of $\lambda$ values, but performs quite similarly for large and small $\lambda$.  Figure \ref{fig:adaptive-Omega2}(b) illustrates adaptive performance more directly, plotting the risk ratio $\frac{R_{\infty}^{\lambda}(\delta)}{\min_{\tilde{\delta}}R_{\infty}^{\lambda}(\tilde{\delta})}$ as a function of $\lambda$.  Here we see that, consistent with the findings of \cite{ArmstrongKlineSun2025} for their optimality criterion, the simple soft-thresholding and ERM estimators perform nearly as well as the more complicated linear spline procedure, with an adaptive risk close to 1.5 for all procedures considered. Whether a 50\% increase in constrained multiplier risk is an acceptable tradeoff for eliminating dependence on $\lambda$ seems to depend on one's priorities, but we find it encouraging that near-optimal adaptation is possible using simple combinations of the ML and GMM estimators.

Figure \ref{fig:adaptive-Omega6} presents similar results when $\Omega = 6$. The simple estimators perform comparably to the linear spline. However, in the intermediate range of $\lambda$, the relative performance is less clear because we optimize only the maximum of the regret criterion.

\begin{figure}[htbp]
    \centering

    %==================== Panel A ====================
    \begin{subfigure}[t]{\textwidth}
        \centering
        \setcounter{innersubfigure}{0}

        \begin{minipage}[t]{0.48\linewidth}
            \centering
            \includegraphics[width=\linewidth]{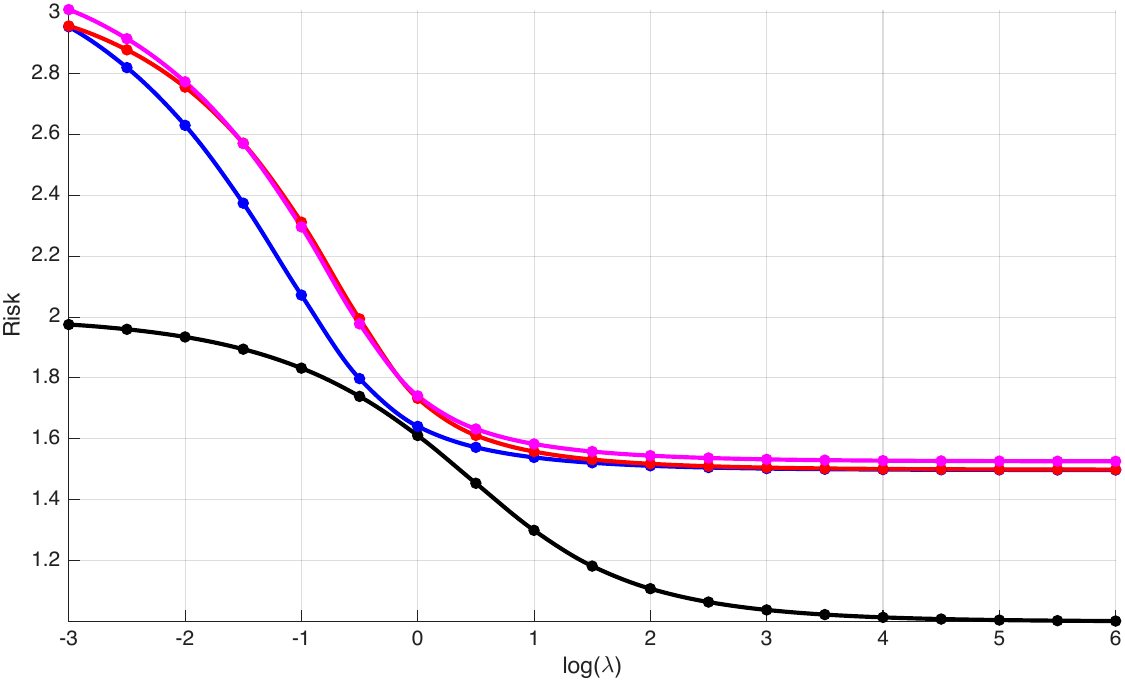}
            \refstepcounter{innersubfigure}
            \caption*{(\theinnersubfigure)\ Risk of Adaptive Estimators}
            \label{fig:risk-Omega2}
        \end{minipage}
        \hfill
        \begin{minipage}[t]{0.48\linewidth}
            \centering
            \includegraphics[width=\linewidth]{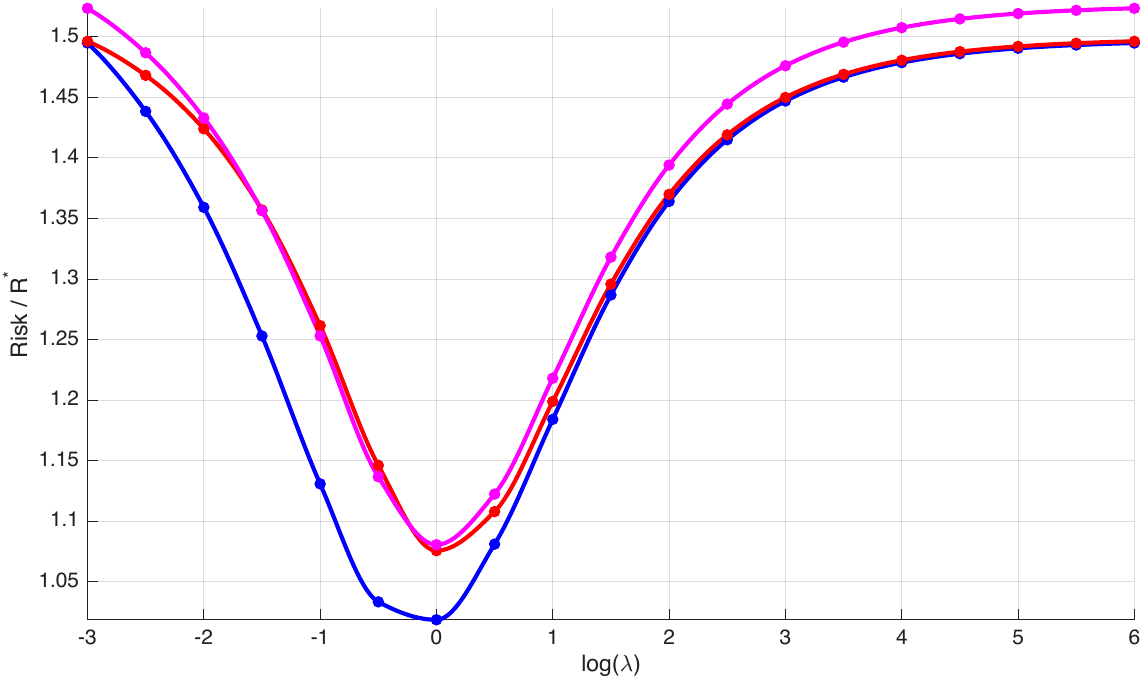}
            \refstepcounter{innersubfigure}
            \caption*{(\theinnersubfigure)\ Risk Ratio against Pointwise Optimum}
            \label{fig:risk-ratio-Omega2}
        \end{minipage}

        \caption{$\Omega = 2$}
        \label{fig:adaptive-Omega2}
    \end{subfigure}

    \vspace{1em}

    %==================== Panel B ====================
    \begin{subfigure}[t]{\textwidth}
        \centering
        \setcounter{innersubfigure}{0}

        \begin{minipage}[t]{0.48\linewidth}
            \centering
            \includegraphics[width=\linewidth]{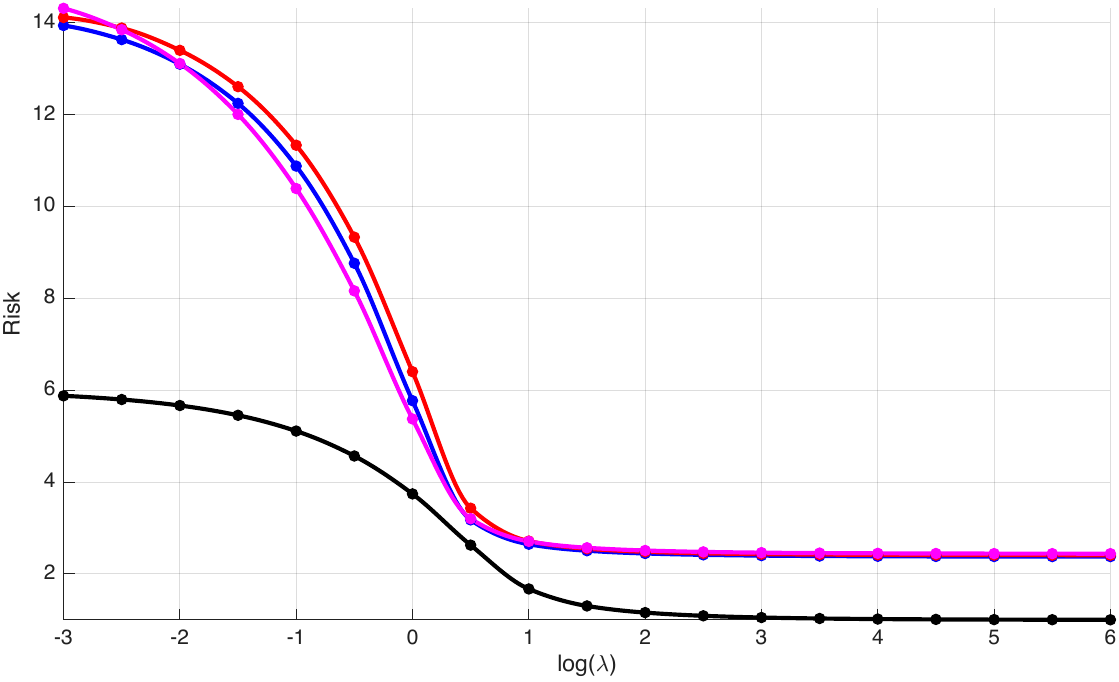}
            \refstepcounter{innersubfigure}
            \caption*{(\theinnersubfigure)\ Risk of Adaptive Estimators}
            \label{fig:risk-Omega6}
        \end{minipage}
        \hfill
        \begin{minipage}[t]{0.48\linewidth}
            \centering
            \includegraphics[width=\linewidth]{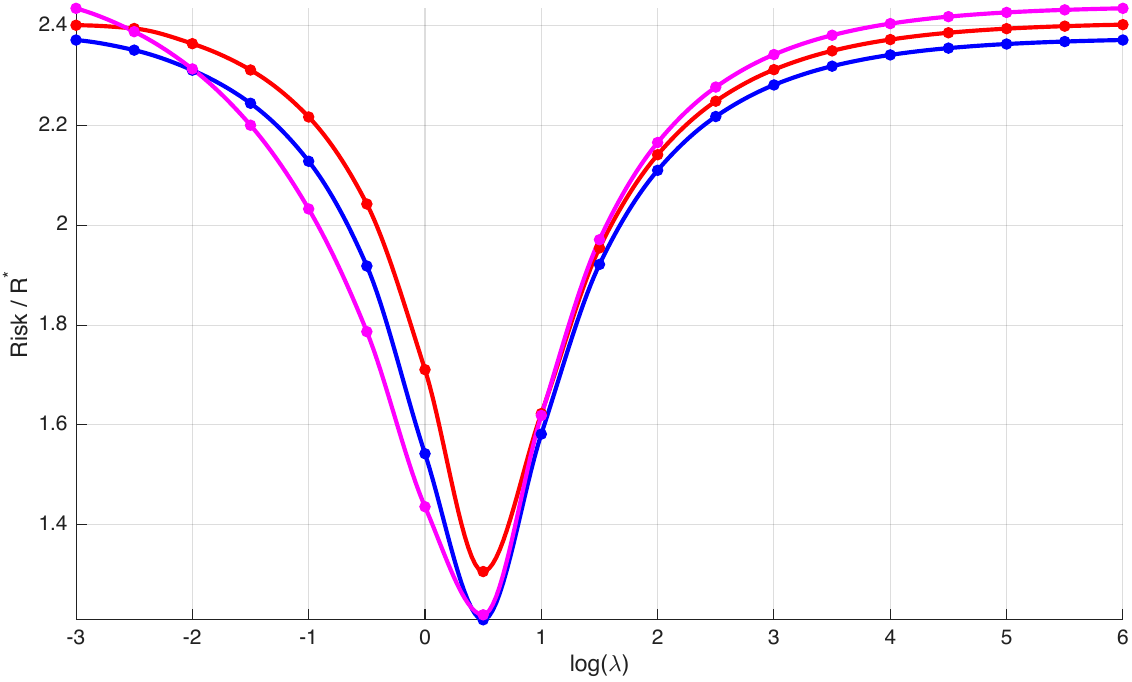}
            \refstepcounter{innersubfigure}
            \caption*{(\theinnersubfigure)\ Risk Ratio against Pointwise Optimum}
            \label{fig:risk-ratio-Omega6}
        \end{minipage}

        \caption{$\Omega = 6$}
        \label{fig:adaptive-Omega6}
    \end{subfigure}
    \begin{subfigure}[t]{\textwidth}
        \centering
        \includegraphics{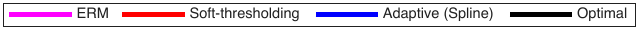}
    \end{subfigure}
    \caption{Adaptively Optimal, Soft-thresholding and ERM Estimators}
    \label{fig:adaptive-main}
\end{figure}

\bibliographystyle{econ-econometrica}
\bibliography{Misspec}

\appendix

\section{Axioms for Misspecification-Averse Preferences}\label{sec: appendix axioms}

We begin by presenting a series of axioms and corresponding representation theorems adapted to our loss function setting, which characterize the preferences \eqref{eq:CV-rep} and \eqref{eq:variational-rep}.  While the results in this appendix largely follow from arguments in \cite{maccheroni2006ambiguity} and \cite{CerriaViolioetal24}, we hope that our framing of the choice problem in terms of loss functions may be helpful for readers with a background in econometrics and statistics.

\subsection{Basic Axioms}
Throughout Appendix \ref{sec: appendix axioms}, let $L_\theta=L(\theta,\cdot): \mathcal{X} \rightarrow \mathbb{R}$ be the loss function $L$ conditional on $\theta$. Although the domain of $L_\theta$ is $\mathcal{X}$, we will abuse notation and identify it with the function $L_\theta \in \mathcal{L}$ with domain $\Theta \times \mathcal{X}$ satisfying: $L_\theta(\theta',x)=L(\theta,x)$ for all $\theta' \in \Theta$ and $x \in \mathcal{X}$. Note that we may perform this identification for any bounded, measurable function $\mathcal{X} \rightarrow \mathbb{R}$: letting $\mathcal{L}(\mathcal{X})$ be the space of such functions, we observe that any binary relation $\succsim$ on all of $\mathcal{L}$ induces a binary relation on $\mathcal{L}(\mathcal{X})$ simply by restricting $\succsim$ to $\mathcal{L}(\mathcal{X})$.

First, we establish how the conditional preferences $\{\succsim_\theta\}_{\theta \in \Theta}$ may be derived from $\succsim_\Theta$. Suppose that we only take $\succsim_\Theta$ as primitive. For each $\theta \in \Theta$, define the binary relation $\succsim_\theta$ on $\mathcal{L}$ as follows: for each $L,L' \in \mathcal{L}$, say that
\[
L \succsim_\theta L' \iff L_\theta k \succsim_\Theta L_\theta' k \quad \forall k \in \mathbb{R}
\]
where $L_\theta k \in \mathcal{L}$ is the loss function defined as: $(L_\theta k)(\theta',x)=L(\theta,x)$ if $\theta'=\theta$, and $(L_\theta k)(\theta',x)=k$ else.

\begin{lemma}\label{lem:cond_prefs}
Let the conditional preferences $\{\succsim_\theta\}_{\theta \in \Theta}$ be defined as above. Assume $\succsim_\Theta$ is represented by $V_\Theta(L)=\sup_{\theta \in \Theta} V_\theta(L)$ for some set of functions $\{V_\theta\}_{\theta \in \Theta}$ satisfying: for each $\theta \in \Theta$, $V_\theta(L)=V_\theta(L_\theta)$ for all $L \in \mathcal{L}$, and $V_\theta(k)=k$ for all $k \in \mathbb{R}$. Then, for each $\theta \in \Theta$, $V_\theta$ represents $\succsim_\theta$.
\end{lemma}

\paragraph{Proof of Lemma \ref{lem:cond_prefs}.} By definition of $\succsim_\theta$ and by assumption,
\begin{align*}
    L\succsim_\theta L' \iff V_\Theta(L_\theta k)\geq V_\Theta(L_\theta' k) \quad \forall k \in \mathbb{R} \\
    \iff \max\{V_\theta(L_\theta),k\}\geq \max\{V_\theta(L_\theta'),k\} \quad \forall k\in \mathbb{R} \\
    \iff V_\theta(L_\theta) \geq V_\theta(L_\theta')
\end{align*}
where the backwards direction of the last iff follows because the function $x \mapsto \max\{x,k\}$ is increasing for any fixed $k \in \mathbb{R}$, and the forwards direction of the last iff follows by assuming $V_\theta(L_\theta)<V_\theta(L_\theta')$ and choosing $k<V_\theta(L_\theta)$. \qed

\medskip

Note that the assumption in Lemma \ref{lem:cond_prefs} is satisfied by all the representations we consider. In this sense, it is without loss of generality to take as primitive the family of preferences $\{\succsim_\Theta\} \cup \{\succsim_\theta: \theta \in \Theta\}$. We do this henceforth. We proceed by imposing a set of basic regularity conditions on $\succsim_\Theta$ and $\{\succsim_\theta\}_{\theta \in \Theta}$. Recall that these are binary relations on $\mathcal{L}$, which is the set of bounded, Borel measurable functions $L:\Theta \times \mathcal{X} \rightarrow \mathbb{R}$.

\begin{axiom}[$\theta$-Relevance]\label{ax:7}
For all $\theta \in \Theta$ and all $L,L' \in \mathcal{L}$,
\[
L \succsim_\theta L' \iff L_\theta \succsim_\theta L_\theta'
\]
\end{axiom}

$\theta$-Relevance ensures that $\succsim_\theta$ only cares about loss functions conditional on $\theta$.

\begin{axiom}[Nontrivial Weak Order]\label{ax:1}
The relation $\succsim_\Theta$ is complete, transitive, and nontrivial.  For each $\theta\in\Theta$, the relation $\succsim_\theta$ is complete, transitive, and nontrivial.
\end{axiom}

Completeness requires that the researcher be able to rank any two loss functions.  Transitivity is a standard coherence requirement. Nontriviality rules out the uninteresting case where the researcher is indifferent between all loss functions.

\begin{axiom}[Monotonicity]\label{ax:2}
If $L(\theta',x)\le L'(\theta',x)$ for all $(\theta',x)\in\Theta\times\mathcal{X}$, then $L\succsim_\Theta L'$. For each $\theta\in\Theta$, if $L(\theta,x)\le L'(\theta,x)$ for all $x \in \mathcal{X}$, then $L\succsim_\theta L'$.
\end{axiom}

Monotonicity requires that the researcher prefer lower losses: if $L$ yields weakly lower loss than $L'$ in every state, then the researcher must weakly prefer $L$.\footnote{Note that Monotonicity implies: for each $\succsim \in \{\succsim_\Theta,\{\succsim_\theta\}_{\theta\in\Theta}\}$, $\succsim$ restricted to $\mathbb{R}$ satisfies: $r \leq r'$ if and only if $r \succsim r'$. In particular, this implies \cite{maccheroni2006ambiguity}'s Axiom A.7 (Unboundedness).}  This is a weak dominance condition satisfied by any sensible loss-based criterion.

Next, we state several continuity axioms. Such axioms are standard to derive utility representations, but we view them as primarily technical conditions. Our first continuity axiom is imposed on $\succsim_\Theta$.

\begin{axiom}[$\Theta$-Mixture Continuity]
For $L,L',L''\in\mathcal{L}$, the sets
$\{\alpha\in[0,1]:\alpha L + (1-\alpha)L'\succsim_\Theta L''\}$ and
$\{\alpha\in[0,1]:L''\succsim_\Theta \alpha L + (1-\alpha)L'\}$
are closed.
\end{axiom}

Mixture Continuity ensures that small perturbations to mixtures of fixed loss functions do not lead to discontinuous jumps in the preference ranking.  Here and throughout, the mixture $\alpha L + (1-\alpha)L'$ is defined pointwise: $[\alpha L + (1-\alpha)L'](\theta,x) = \alpha L(\theta,x) + (1-\alpha)L'(\theta,x)$.  Such mixtures arise naturally when the researcher randomizes between decision rules: if they use $\delta$ with probability $\alpha$ and $\delta'$ otherwise, independent of the data, then the induced loss is $\alpha L_\delta + (1-\alpha)L_{\delta'}$.

Our second notion of continuity, imposed on $\{\succsim_\theta: \theta \in \Theta\}$, is stronger and requires a topology on the space of loss functions. We endow $\mathcal{L}$ with the $\sup$-norm topology.

\begin{axiom}[$\theta$-Continuity]\label{ax:ten}
For each $L \in \mathcal{L}$, the sets
\[
\{L' \in \mathcal{L}: L'\succsim_\theta L\} \quad \text{and} \quad \{L' \in \mathcal{L}: L\succsim_\theta L'\}
\]
are closed.
\end{axiom}

$\theta$-Continuity ensures that small perturbations (in the sense of $\sup$-norm) of loss functions do not imply discontinuous jumps in the preference ranking. Note that, since $\alpha_n \to \alpha$ implies $\alpha_n L+(1-\alpha_n)L' \to \alpha L+(1-\alpha)L'$, $\theta$-Continuity implies $\theta$-Mixture Continuity.

\begingroup
\addtocounter{axiom}{-1}
\renewcommand{\theaxiom}{\arabic{axiom}'}
\begin{axiom}[$\theta$-Mixture Continuity]
For each $\theta \in \Theta$ and $L,L',L''\in\mathcal{L}$, the sets
$\{\alpha\in[0,1]:\alpha L + (1-\alpha)L'\succsim_\theta L''\}$ and
$\{\alpha\in[0,1]:L''\succsim_\theta \alpha L + (1-\alpha)L'\}$
are closed.
\end{axiom}
\endgroup

\subsection{Axiomatization of Variational Preferences}

We next present axioms that imply a particular functional form for the conditional preferences $\succsim_\theta$. As we discuss in the main text, this class of \emph{variational preferences} nests a number of important cases previously discussed in the literature on estimation and decision-making with misspecification concerns. 

The first axiom restricts how $\succsim_\theta$ responds to mixing with constant losses.

\begin{axiom}[Weak Certainty Independence]\label{ax:4}
For all $\theta\in\Theta$, $L,L'\in\mathcal{L}$, $r,r'\in\mathbb{R}$, and $\alpha\in(0,1)$,
\[\alpha L+(1-\alpha)r \succsim_\theta \alpha L'+(1-\alpha)r \Rightarrow \alpha L+(1-\alpha)r' \succsim_\theta \alpha L'+(1-\alpha)r'\]
\end{axiom}

Weak Certainty Independence requires that preferences between loss functions mixed with constant acts not depend on the value of the constant act.

\begin{axiom}[Uncertainty Aversion]\label{ax:5}
For all $\theta\in\Theta$, $L,L'\in\mathcal{L}$, and $\alpha\in(0,1)$,
\[L\sim_\theta L' \Rightarrow \alpha L+(1-\alpha)L'\succsim_\theta L\]
\end{axiom}

Uncertainty Aversion requires that the researcher have a weak preference for hedging: if they are indifferent between two loss functions, they must weakly prefer a mixture of the two.  This captures the intuition that diversification (e.g. via randomization) is valuable when facing uncertainty about the data generating process.

The remaining axioms imply certain properties of the cost function. First, we require that $\succsim_\theta$ always finds $Q_\theta$ at least as plausible as any other DGP.

\begin{axiom}\label{ax:grounded at Q}
For all $\theta \in \Theta$, $\succsim_\theta$ is more ambiguity averse than $\succsim_{\theta,Q_\theta}^{SEU}$.
\end{axiom}

Let $\Delta^F(\mathcal{X})$ be the set of finitely additive Borel probability measures on $\mathcal{X}$, endowed with the weak$^*$-topology. Recall that $\Delta(\mathcal{X})\subseteq \Delta^F(\mathcal{X})$ is the set of countably additive Borel probability measures on $\mathcal{X}$. For a set $C\subseteq \Delta^F(\mathcal{X})$, let $\Bar{C}$ denote its weak$^*$-closure. Note that, under the axioms we have imposed so far, for each $\theta \in \Theta$ and $L \in \mathcal{L}$, there exists a unique $CE_\theta(L) \in \mathbb{R}$ such that $CE_\theta(L) \sim_\theta L$.\footnote{The argument is exactly analogous to the proof of Step 1 of Theorem \ref{thm:CD-rep}(i).} Let $\mathcal{L}_0$ denote the set of simple, measurable loss functions. We will show that the following axiom ensures that the variational representation may be written only in the language of countably additive probabilities.

\begin{axiom}\label{ax:14}
\normalfont For all $\theta \in \Theta$ and $t\geq 0$,
\[
\left\{P \in \Delta^F(\mathcal{X}): \sup_{L \in \mathcal{L}_0} \left(\int L_\theta \ dP-CE_\theta(L) \right)\leq t\right\}=\overline{\left\{P \in \Delta(\mathcal{X}): \sup_{L \in \mathcal{L}_0} \left(\int L_\theta \ dP-CE_\theta(L) \right)\leq t\right\}}
\]
\end{axiom}

Finally, the following axiom ensures that the cost function is weakly lower semicontinuous.

\begin{axiom}\label{ax:15}
For all $\theta \in \Theta$ and $t\geq 0$, the set
\[
\left\{P \in \Delta(\mathcal{X}): \sup_{L \in \mathcal{L}_0} \left(\int L_\theta \ dP-CE_\theta(L) \right)\leq t\right\}
\]
is weakly closed.
\end{axiom}

\paragraph{Remark.} We may also endow $\Delta(\mathcal{X})$ with the topology of weak convergence. Under Axioms \ref{ax:7}--\ref{ax:5}, note that Axioms \ref{ax:14} and \ref{ax:15} are implications of Monotone Continuity, since by Theorem 13 of \cite{maccheroni2006ambiguity}, Monotone Continuity implies that the LHS set of Axiom \ref{ax:14} is a weakly compact (and hence weakly closed) subset of $\Delta(\mathcal{X})$.\footnote{More precisely, letting $c^*(P)=\sup_{L \in \mathcal{L}_0} \left(\int L_\theta \ dP-CE_\theta(L)\right)$, we see that Monotone Continuity implies that $\{P \in \Delta^F(\mathcal{X}): c^*(P)\leq t\}=\{P \in \Delta(\mathcal{X}): c^*(P)\leq t\}$, so Axiom \ref{ax:14} follows by taking the weak$^*$-closure of both sides.}

\medskip

\begin{defn}
The preference $\succsim_\theta$ has a \emph{variational representation} on $\mathcal{L}$ if there exists a convex, weak$^*$ lower-semicontinuous function $c_\theta:\Delta^F(\mathcal{X})\to[0,\infty]$ with $\inf_{P\in\Delta^F(\mathcal{X})}c_\theta(P)=0$ such that
\begin{equation}\label{eq:variational-rep 2}
V_\theta(L)=\max_{P\in\Delta^F(\mathcal{X})}\left\{\int L_\theta(x)dP(x)-c_\theta(P)\right\}
\end{equation}
represents $\succsim_\theta$ on $\mathcal{L}$, in the sense that $L\succsim_\theta L'\iff V_\theta(L)\le V_\theta(L')$ for all $L,L' \in \mathcal{L}$.
\end{defn}

Note that the $\max$ in Equation (\ref{eq:variational-rep 2}) is attained because $\Delta^F(\mathcal{X})$ is weak$^*$-compact and the map
\[
P \mapsto \int L_\theta \ dP-c_\theta(P)
\]
is the sum of a real-valued, weak$^*$ continuous function $P \mapsto \int L_\theta \ dP$ (since $L_\theta$ is bounded and measurable) and a weak$^*$ upper-semicontinuous function $P \mapsto -c_\theta(P)$.

Our main result in this section establishes that, under the axioms we have imposed thus far, the conditional preferences $\succsim_\theta$ have a variational representation whose cost function $c_\theta$ satisfies certain appealing properties. Many of the representation results in the main text refine this one by characterizing particular functional forms of the cost function $c_\theta$.

\begin{thm}\label{thm:D-rep}
\begin{itemize}
    \item[(i)] The preference $\succsim_\theta$ satisfies Axioms \ref{ax:7}--\ref{ax:5} if and only if it has a variational representation on $\mathcal{L}$.
    \item[(ii)] If $\succsim_\theta$ has a variational representation on $\mathcal{L}$, it is unique and given by:
    \[
    c_\theta(P)=\sup_{L \in \mathcal{L}_0} \left(\int L_\theta \ dP-CE_\theta(L) \right)
    \]
    \item[(iii)] Suppose $\succsim_\theta$ satisfies Axioms \ref{ax:7}--\ref{ax:5}. It additionally satisfies Axiom \ref{ax:grounded at Q} if and only if $c_\theta(Q_\theta)=0$.
    \item[(iv)] Suppose $\succsim_\theta$ satisfies Axioms \ref{ax:7}--\ref{ax:5}. It additionally satisfies Axiom \ref{ax:14} if and only if 
    \[
    \{P \in \Delta^F(\mathcal{X}): c_\theta(P)\leq t\}=\overline{\{P \in \Delta(\mathcal{X}): c_\theta(P)\leq t\}} \quad \forall t\geq 0
    \]
    In this case, for $V_\theta$ as defined in Equation (\ref{eq:variational-rep 2}),
    \[
    V_\theta(L)=\sup_{P \in \Delta(\mathcal{X})} \left\{\int L_\theta(x) \ dP(x)-c_\theta(P)\right\} \quad \forall L \in \mathcal{L}
    \]
    In particular for $L=0$,
    \[
    \inf_{P \in \Delta(\mathcal{X})} c_\theta(P)=0
    \]
    \item[(v)] Suppose $\succsim_\theta$ satisfies Axioms \ref{ax:7}--\ref{ax:5}. It additionally satisfies Axiom \ref{ax:15} if and only if the restriction of $c_\theta$ to $\Delta(\mathcal{X})$ is weakly lower-semicontinuous.
\end{itemize}
\end{thm}

Before proving Theorem \ref{thm:D-rep}, we prove two useful lemmas which we use throughout Appendices A and B. These two lemmas will allow us to 1) equivalently work with utility representations rather than risk representations; and 2) extend representation theorems for preferences restricted to $\mathcal{L}_0$ to preferences over all of $\mathcal{L}$. Hence, these lemmas address two features of our choice environment that are nonstandard in the microeconomic decision theory literature: our primitives are preferences over bounded measurable loss functions rather than simple measurable Anscombe--Aumann acts.

First, we observe that working with preferences over loss functions is equivalent to working with preferences over negative loss functions, or \emph{utility acts}. Recall that a function $V_\theta$ \emph{represents $\succsim_\theta$ on $\mathcal{L}$} if and only if
\[
L \succsim_\theta L' \iff V_\theta(L)\leq V_\theta(L') \quad \forall L,L' \in \mathcal{L}
\]
For each loss function $L \in \mathcal{L}$, define its induced \emph{utility act} to be $f_L=-L$. Let $\mathcal{F}=\{f_L: L \in \mathcal{L}\}$ be the set of utility acts induced by some (bounded, Borel measurable) loss function, and note that $\mathcal{F}$ is precisely the set of bounded, Borel measurable functions $f: \Theta \times \mathcal{X} \rightarrow \mathbb{R}$. Finally, for each $\succsim_\theta$, define the binary relation $\succsim_\theta^{\mathcal{F}}$ on $\mathcal{F}$ as:
\[
f_L \succsim_\theta^{\mathcal{F}} f_{L'} \iff L \succsim_\theta L'
\]
For a binary relation $\succsim_\theta'$ on $\mathcal{F}$ and a function $U_\theta: \mathcal{F}\rightarrow \mathbb{R}$, say that \emph{$U_\theta$ represents $\succsim_\theta'$ on $\mathcal{F}$} if and only if
\[
f \succsim_\theta' g \iff U_\theta(f)\geq U_\theta(g) \quad \forall f,g \in \mathcal{F}
\]

\begin{lemma}\label{lem:utility acts}
$V_\theta$ represents $\succsim_\theta$ on $\mathcal{L}$ if and only if $U_\theta=-V_\theta(-\cdot)$ represents $\succsim_\theta^{\mathcal{F}}$ on $\mathcal{F}$.
\end{lemma}

\paragraph{Proof of Lemma \ref{lem:utility acts}.}

The following statements (i)-(v) are equivalent:
\begin{itemize}
    \item[(i)] $V_\theta$ represents $\succsim_\theta$ on $\mathcal{L}$.
    \item[(ii)] $L \succsim_\theta L' \iff V_\theta(L)\leq V_\theta(L') \quad \forall L,L' \in \mathcal{L}$.
    \item[(iii)] $-L \succsim_\theta^{\mathcal{F}} -L' \iff -U_\theta(-L)\leq -U_\theta(-L') \quad \forall L,L' \in \mathcal{L}$.
    \item[(iv)] $f \succsim_\theta^{\mathcal{F}} g \iff U_\theta(f)\geq U_\theta(g) \quad \forall f,g \in \mathcal{F}$.
    \item[(v)] $U_\theta$ represents $\succsim_\theta^{\mathcal{F}}$ on $\mathcal{F}$.
\end{itemize}

\qed

\medskip

In particular, Lemma \ref{lem:utility acts} implies: $V_\theta$ as defined in Equation (\ref{eq:variational-rep 2}) represents $\succsim_\theta$ on $\mathcal{L}$ if and only if $U_\theta: \mathcal{F}\rightarrow \mathbb{R}$ defined as:
\begin{equation}\label{eq: var-rep-utility}
    U_\theta(f)=\min_{P \in \Delta^F(\mathcal{X})} \left\{\int f_\theta(x) \ dP(x)+c_\theta(P) \right\}
\end{equation}
represents $\succsim_\theta^{\mathcal{F}}$ on $\mathcal{F}$. Note that for each axiom we have defined on $\succsim_\theta$, there exists a corresponding, equivalent axiom on $\succsim_\theta^{\mathcal{F}}$. Moving forward, we abuse notation and use $\succsim_\theta$ to also refer to its induced preference on utility acts $\succsim_\theta^{\mathcal{F}}$.

Second, we establish that to obtain a variational representation $U_\theta$ (as defined in Equation \ref{eq: var-rep-utility}) of $\succsim_\theta$ on $\mathcal{F}$, it suffices to obtain a variational representation $U_\theta$ of $\succsim_\theta$ on the domain $\mathcal{F}_0(\mathcal{X}) \subseteq \mathcal{F}$ of \emph{simple} (finitely-valued) Borel measurable functions $f$ such that $f(\cdot,x)$ is constant in $\theta$. We abuse notation and identify this set with the set of simple, Borel measurable functions $f: \mathcal{X} \rightarrow \mathbb{R}$. 

\begin{lemma}\label{lem: simple-to-bounded}
Assume that $\succsim_\theta$ satisfies the basic axioms ($\theta$-Relevance, Nontrivial Weak Order, Monotonicity, and $\theta$-Continuity), and $\succsim_\theta$ has a variational representation on $\mathcal{F}_0(\mathcal{X})$: there exists a convex, weak$^*$ lower-semicontinuous function $c_\theta: \Delta^F(\mathcal{X}) \rightarrow [0,\infty]$ with $\inf_{P \in \Delta^F(\mathcal{X})} c_\theta(P)=0$ such that
\begin{align*}
    f \succsim_\theta g \iff U_\theta(f)\geq U_\theta(g) \quad \forall f,g \in \mathcal{F}_0(\mathcal{X})
\end{align*}
where $U_\theta$ is defined as in Equation (\ref{eq: var-rep-utility}). Then, the above equivalence holds for all $f,g \in \mathcal{F}$, and hence $U_\theta$ represents $\succsim_\theta$ on $\mathcal{F}$.
\end{lemma}

\paragraph{Proof of Lemma \ref{lem: simple-to-bounded}.}

\underline{Step 1: $U_\theta$ is continuous on $\mathcal{F}$.} Define $U_\theta: \mathcal{F} \rightarrow \mathbb{R}$ as in Equation (\ref{eq: var-rep-utility}). Fix any $f,g \in \mathcal{F}$ and without loss of generality, suppose $U_\theta(f)\geq U_\theta(g)$. Choose $P \in \Delta^F(\mathcal{X})$ such that
\[
U_\theta(g)=\int g_\theta \ dP+c_\theta(P)
\]
Hence,
\begin{align*}
    U_\theta(f)-U_\theta(g)\leq \int f_\theta \ dP+c_\theta(P)-\int g_\theta \ dP-c_\theta(P) \\
    =\int (f_\theta-g_\theta) \ dP \leq \int |f_\theta-g_\theta| \ dP \leq \|f_\theta-g_\theta\|_\infty
\end{align*}
where the last two inequalities follow from monotonicity of the integral against finitely additive measures. Hence, $U_\theta$ is 1-Lipschitz and hence continuous on $\mathcal{F}$.

\underline{Step 2: $U_\theta$ represents $\succsim_\theta$ on $\mathcal{F}$.} Fix any $f,g \in \mathcal{F}$. 

First, suppose that $f \succsim_\theta g$. By $\theta$-Relevance, $f_\theta \succsim_\theta g_\theta$. For each $k\geq 1$ and $j\geq 1$, define
\[
f_{\theta,k}(x)=\frac{\lfloor kf_\theta(x)\rfloor+1}{k} \quad \text{and} \quad g_{\theta,j}(x)=\frac{\lfloor j g_\theta(x)\rfloor}{j}
\]
Since $f_\theta,g_\theta$ are bounded, each $f_{\theta,k},g_{\theta,j}$ is simple. Also note that $kf_\theta(x)\leq \lfloor kf_\theta(x)\rfloor+1$ and $\lfloor jg_\theta(x)\rfloor\leq jg_\theta(x)$ imply $f_\theta\leq f_{\theta,k}$ and $g_\theta\geq g_{\theta,j}$. Finally, note that $\|f_{\theta,k}-f_\theta\|_\infty\leq k^{-1} \to 0$ and $\|g_{\theta,j}-g_\theta\|_\infty\leq j^{-1} \to 0$. 

For each fixed $j\geq 1$, Monotonicity implies:
\[
f_{\theta,k} \succsim_\theta f_\theta \succsim_\theta g_\theta \succsim_\theta g_{\theta,j} \quad \forall k\geq 1 \implies U_\theta(f_{\theta,k})\geq U_\theta(g_{\theta,j}) \quad \forall k\geq 1
\]
and hence 
\[
U_\theta(f_\theta)=\lim_{k\to\infty} U_\theta(f_{\theta,k})\geq U_\theta(g_{\theta,j})
\]
Finally, taking $j \to \infty$ yields $U_\theta(f_\theta)\geq U_\theta(g_\theta)$. By relabeling $f_\theta$ and $g_\theta$, we also see that if $f_\theta \sim_\theta g_\theta$, then $U_\theta(f_\theta)=U_\theta(g_\theta)$. 

Finally, suppose that $f \succ_\theta g$.  By $\theta$-Relevance, $f_\theta \succ_\theta g_\theta$. By $\theta$-continuity, there exists $\epsilon>0$ small enough such that $f_\theta \succ_\theta g_\theta+\epsilon$. Define $(f_{\theta,k})_k$ and $(g_{\theta,j})_j$ as before. For each fixed $j\geq 1$, Monotonicity implies:
\[
f_{\theta,k}\succsim_\theta f_\theta \succ_\theta g_\theta+\epsilon \succsim_\theta g_{\theta,j}+\epsilon \quad \forall k\geq 1 \implies U_\theta(f_{\theta,k})>U_\theta(g_{\theta,j}+\epsilon)=U_\theta(g_{\theta,j})+\epsilon \quad \forall k\geq 1
\]
where the last equality uses the fact that $U_\theta$ is translation invariant. Hence, taking $k \to \infty$ yields:
\[
U_\theta(f_\theta)\geq U_\theta(g_{\theta,j})+\epsilon
\]
and taking $j \to \infty$ yields:
\[
U_\theta(f_\theta)\geq U_\theta(g_\theta)+\epsilon>U_\theta(g_\theta)
\]
as desired. \qed

\medskip

\paragraph{Proof of Theorem \ref{thm:D-rep}.}
Necessity of the axioms (the backwards directions of (i), (iii), (v)) is straightforward, so we prove sufficiency. Throughout, we study $\succsim_\theta$ restricted to the domain $\mathcal{F}_0(\mathcal{X})$. Note that this corresponds to the choice environment of \cite{maccheroni2006ambiguity} where the state space is $(\mathcal{X},\mathcal{B}(\mathcal{X}))$ and the consequence space is $\mathbb{R}$. 

\underline{Part (i)}: We show the forwards direction. It is straightforward to show that Axioms \ref{ax:7}--\ref{ax:5} imply that $\succsim_\theta$ on $\mathcal{F}_0(\mathcal{X})$ satisfies Axioms A.1-A.6 of \cite{maccheroni2006ambiguity}. By Theorem 3 of \cite{maccheroni2006ambiguity}, there exist a nonconstant affine function $u_\theta: \mathbb{R} \rightarrow \mathbb{R}$ and convex, weak$^*$-lower-semicontinuous function $c_{\theta,u_\theta}: \Delta^F(\mathcal{X}) \rightarrow [0,\infty]$ satisfying $\inf_{P \in \Delta^F(\mathcal{X})} c_{\theta,u_\theta}(P)=0$ such that the function 
\[
f \mapsto \min_{P \in \Delta^F(\mathcal{X})} \left(\int u_\theta(f) \ dP+c_{\theta,u_\theta}(P) \right)
\]
represents $\succsim_\theta$ on $\mathcal{F}_0(\mathcal{X})$. By cardinal uniqueness (Corollary 5 of \cite{maccheroni2006ambiguity}) and Monotonicity, we may choose $u_\theta=\text{Id}$ and define $c_\theta=c_{\theta,\text{Id}}$. Hence, the function
\[
U_\theta(f)=\min_{P \in \Delta^F(\mathcal{X})} \left(\int f \ dP+c_{\theta}(P) \right)
\]
represents $\succsim_\theta$ on $\mathcal{F}_0(\mathcal{X})$. Conclude by applying Lemmas \ref{lem:utility acts} and \ref{lem: simple-to-bounded}. 

\underline{Part (ii)}: Since $\mathbb{R}$ is unbounded, uniqueness and the given expression immediately follow from Proposition 6 of \cite{maccheroni2006ambiguity}.

\underline{Part (iii)}: Additionally suppose that $\succsim_\theta$ satisfies Axiom \ref{ax:grounded at Q}. By Lemma 32 of \cite{maccheroni2006ambiguity}, $c_\theta(Q_\theta)=0$.

\underline{Part (iv)}: Additionally suppose that $\succsim_\theta$ satisfies Axiom \ref{ax:14}. The equalities
\[
\{P \in \Delta^F(\mathcal{X}): c_\theta(P)\leq t\}=\overline{\{P \in \Delta(\mathcal{X}): c_\theta(P)\leq t\}} \quad \forall t\geq 0
\]
immediately follow from Parts (i) and (ii). It remains to show
\[
V_\theta(L)=\sup_{P \in \Delta(\mathcal{X})} \left\{\int L_\theta(x) \ dP(x)-c_\theta(P)\right\} \quad \forall L \in \mathcal{L}
\]
or equivalently by Lemma \ref{lem:utility acts},
\[
\min_{P \in \Delta^F(\mathcal{X})} \left(\int f_\theta \ dP+c_\theta(P) \right)= U_\theta(f)=\inf_{P \in \Delta(\mathcal{X})} \left\{\int f_\theta(x) \ dP(x)+c_\theta(P)\right\} \quad \forall f \in \mathcal{F}
\]
for $U_\theta$ as defined in Equation (\ref{eq: var-rep-utility}). Fix any $f \in \mathcal{F}$. The inequality $\leq$ follows immediately. For the other inequality $\geq$, let $P^* \in \Delta^F(\mathcal{X})$ such that 
\[
U_\theta(f)=\int f_\theta \ dP^*+c_\theta(P^*)
\]
and define $t^*=c_\theta(P^*)$. Since 
\[
P^* \in \{P \in \Delta^F(\mathcal{X}): c_\theta(P)\leq t^*\}=\overline{\{P \in \Delta(\mathcal{X}): c_\theta(P)\leq t^*\}}
\]
by \cite{aliprantis2006infinite} Theorem 2.14, there exists a net
\[
(P_\alpha)_{\alpha\in A} \subseteq \{P \in \Delta(\mathcal{X}): c_\theta(P)\leq t^*\}
\]
such that $P_\alpha \to P^*$ in the weak$^*$ topology. By definition of weak$^*$ lower semicontinuity, 
\[
\liminf_{\alpha \in A} c_\theta(P_\alpha)\geq t^*
\]
By above, 
\begin{align*}
    c_\theta(P_\alpha)\leq t^* \quad \forall \alpha \in A \\
    \implies \sup_{\alpha \geq \alpha_0} c_\theta(P_\alpha) \leq t^* \quad \forall \alpha_0 \in A \\
    \implies \limsup_{\alpha \in A} c_\theta(P_\alpha) \leq t^*
\end{align*}
Hence, $c_\theta(P_\alpha) \to t^*$. Since $P \mapsto \int f_\theta \ dP$ is weak$^*$ continuous (since $f$ is bounded and measurable), we have by Theorem 2.28 of \cite{aliprantis2006infinite} that $\int f_\theta \ dP_\alpha \to \int f_\theta \ dP^*$. Finally, note that
\begin{align*}
    \int f_\theta \ dP_\alpha+c_\theta(P_\alpha)\geq \inf_{P \in \Delta(\mathcal{X})} \left(\int f_\theta \ dP+c_\theta(P) \right) \quad \forall \alpha \in A \\
    \implies U_\theta(f)\geq \inf_{P \in \Delta(\mathcal{X})} \left(\int f_\theta \ dP+c_\theta(P) \right)
\end{align*}
as desired. The remaining assertion immediately follows by substituting $L=0$ and noting that $\inf_{P\in \Delta^F(\mathcal{X})} c_\theta(P)=0$.

\underline{Part (v)}: Additionally suppose that $\succsim_\theta$ satisfies Axiom \ref{ax:15}. By Parts (i) and (ii), this implies: for each $t\geq 0$, the set
\[
\{P\in \Delta(\mathcal{X}): c_\theta(P)\leq t\}
\]
is weakly closed. By definition, $c_\theta: \Delta(\mathcal{X}) \rightarrow [0,\infty]$ is weakly lower-semicontinuous. \qed

\subsection{Axioms for Aggregation Across $\Theta$}

While the preferences $\{\succsim_\theta:\theta\in\Theta\}$ describe the researcher's preferences under known $\theta,$ $\theta$ is unknown in practice, so our ultimate interest is in the overall preference $\succsim_\Theta$. The remaining axioms for $\succsim_\Theta$ relate it to the conditional preferences.

\begin{axiom}[Consistency]\label{ax:8}
If $L\succsim_\theta L'$ for all $\theta\in\Theta$, then $L\succsim_\Theta L'$.
\end{axiom}

Consistency requires that if all conditional preferences agree on a ranking, the overall preference must respect that ranking.

\begin{axiom}[Caution]\label{ax:9}
For all $r\in\mathbb{R}$ and $L\in\mathcal{L}$, if there exists $\theta\in\Theta$ such that $r\succ_\theta L$, then $r\succsim_\Theta L$.
\end{axiom}

Caution requires that if there exists some $\theta\in\Theta$ under which a constant loss $r$ is strictly preferred to a state-dependent loss $L$, then the overall preference must also (weakly) prefer $r$ to $L$.  This is an uncertainty-averse aggregation rule: the researcher is cautious about state-dependent losses whenever any parameter value suggests caution.

\begin{defn}
The family of preferences $\{\succsim_\Theta\}\cup\{\succsim_\theta:\theta\in\Theta\}$  has a \emph{cautious variational representation} if each $\succsim_\theta$ has a variational representation as defined in Equation (\ref{eq:variational-rep 2}) and 
\begin{equation}\label{eq:CV-rep 2}
V_\Theta(L)=\sup_{\theta\in\Theta}V_\theta(L)=\sup_{\theta\in\Theta}\max_{P\in\Delta^F(\mathcal{X})}\left\{\int L_\theta(x)dP(x)-c_\theta(P)\right\}
\end{equation}
represents $\succsim_\Theta$, in the sense that $L\succsim_\Theta L'$ if and only if $V_\Theta(L)\le V_\Theta(L')$. 
\end{defn}

\begin{thm}\label{thm:CD-rep}
\begin{itemize}
    \item[(i)] The family of preferences $\{\succsim_\Theta\}\cup\{\succsim_\theta:\theta\in\Theta\}$ satisfies Axioms \ref{ax:7}--\ref{ax:5} and \ref{ax:8}--\ref{ax:9} if and only if it has a cautious variational representation. 
    \item[(ii)] It additionally satisfies Axioms \ref{ax:grounded at Q}-\ref{ax:15} if and only if each $c_\theta$ satisfies: $c_\theta(Q_\theta)=0$, $\{c_\theta \leq t\}=\overline{\{c_\theta \leq t\} \cap \Delta(\mathcal{X})}$ for all $t\geq 0$, and $c_\theta: \Delta(\mathcal{X})\rightarrow [0,\infty]$ is weakly lower semicontinuous. In this case, we also have: $\inf_{P \in \Delta(\mathcal{X})} c_\theta(P)=0$ and, for $V_\Theta$ defined in Equation (\ref{eq:CV-rep 2}),
    \[
    V_\Theta(L)=\sup_{\theta \in \Theta} \sup_{P \in \Delta(\mathcal{X})} \left(\int L_\theta \ dP-c_\theta(P) \right) \quad \forall L \in \mathcal{L}
    \]
\end{itemize}
\end{thm}

Part (i) of Theorem \ref{thm:CD-rep} adapts arguments from \citet{CerriaViolioetal24} to our setting with loss functions as primitives.  It shows that, under the axioms above, the preference $\succsim_\Theta$ takes a min-max approach to parameter uncertainty, together with a penalized min-max approach to model misspecification.  As we discuss in the text, this class of preferences over loss functions nests some prominent examples which have previously been discussed in the literature on estimation under misspecification.

\paragraph{Proof of Theorem \ref{thm:CD-rep}.} Part (ii) immediately follows from Part (i) and the previous theorem. Furthermore, necessity of the axioms in Part (i) (the backwards direction of Part (i)) is straightforward, so we prove sufficiency.

\underline{Part (i)}: By the exact analog of Lemma \ref{lem:utility acts}, $V_\Theta$ represents $\succsim_\Theta$ on $\mathcal{L}$ if and only if
\begin{equation}\label{eq:U-Theta}
    U_\Theta(f)=\inf_{\theta \in \Theta} U_\theta(f)=\inf_{\theta \in \Theta} \min_{P \in \Delta^F(\mathcal{X})} \left(\int f_\theta(x) \ dP(x)+c_\theta(P) \right)
\end{equation}
represents $\succsim_\Theta^{\mathcal{F}}$ on $\mathcal{F}$. As before, we abuse notation and refer to $\succsim_\Theta^{\mathcal{F}}$ as $\succsim_\Theta$ henceforth. By Theorem \ref{thm:D-rep}, each $\succsim_\theta$ has a variational representation $U_\theta$ on $\mathcal{F}$.

\underline{Step 1: $\succsim_\Theta$ admits CEs on $\mathcal{F}$.} Fix any $f \in \mathcal{F}$. Since $f$ is bounded, there exist $\overline{k}\geq \underline{k} \in \mathbb{R}$ such that $\overline{k}\geq f(\theta,x)\geq \underline{k}$ for all $(\theta,x)$. By Monotonicity, $\overline{k} \succsim_\Theta f \succsim_\Theta \underline{k}$. By $\Theta$-Mixture Continuity and Monotonicity, there exists a unique $\alpha \in [0,1]$ such that $f \sim_\Theta \alpha \overline{k}+(1-\alpha)\underline{k}$. Hence, the certainty equivalent function $CE_\Theta: \mathcal{F}\rightarrow \mathbb{R}$ defined as: $f \sim_\Theta CE_\Theta(f)$ is well-defined. By Monotonicity, $CE_\Theta(\cdot)$ represents $\succsim_\Theta$ on $\mathcal{F}$.

\underline{Step 2: $CE_\Theta\geq U_\Theta$.} Define $U_\Theta$ as in Equation (\ref{eq:U-Theta}). Fix any $f \in \mathcal{F}$ and consider the constant act $U_\Theta(f)$. Note that since $U_\theta(f)\geq U_\Theta(f)=U_\theta(U_\Theta(f))$ for all $\theta \in \Theta$, 
\[
f \succsim_\theta U_\Theta(f) \quad \forall \theta \in \Theta
\]
and hence Consistency implies $f \succsim_\Theta U_\Theta(f)$. Since $CE_\Theta$ is normalized, $CE_\Theta(f)\geq U_\Theta(f)$.

\underline{Step 3: $U_\Theta\geq CE_\Theta$.} Fix any $f \in \mathcal{F}$ and consider the constant act $U_\Theta(f)$. For each $\epsilon>0$, note that there exists $\theta \in \Theta$ such that
\[
U_\Theta(f)+\epsilon> U_\theta(f) \implies U_\Theta(f)+\epsilon \succ_\theta f
\]
By Caution and since $CE_\Theta$ is normalized,
\[
U_\Theta(f)+\epsilon \succsim_\Theta f \implies U_\Theta(f)+\epsilon\geq CE_\Theta(f)
\]
Taking $\epsilon \downarrow 0$ yields $U_\Theta(f)\geq CE_\Theta(f)$. We have therefore shown that $U_\Theta=CE_\Theta$ represents $\succsim_\Theta$ on $\mathcal{F}$, as desired. \qed

\section{Proofs}\label{sec: proofs}

\subsection{Proofs for Results in Section \ref{sec: loss axioms}}

\paragraph{Constraint Preferences.} To characterize constraint preferences, we introduce another axiom. Define
\[
\mathcal{P}_\theta:=\{P \in \Delta^F(\mathcal{X}): \ \succsim_\theta \text{ is more ambiguity averse than } \succsim_{\theta,P}^{SEU}\}
\]

\begin{axiom}\label{ax:18}
For all $\theta \in \Theta$, $\mathcal{P}_\theta=\overline{\mathcal{P}_\theta \cap \Delta(\mathcal{X})}$ and $\mathcal{P}_\theta \cap \Delta(\mathcal{X})$ is weakly closed.
\end{axiom}

We prove the following result, which implies Propositions \ref{prop:GS} and \ref{prop:GS 2}.

\begin{thm}\label{thm:GS-expanded}
\begin{itemize}
    \item[(i)] The preference $\succsim_\theta$ satisfies Axioms \ref{ax:7}-\ref{ax:ten}, Certainty Independence, and Axiom \ref{ax:5} if and only if $\mathcal{P}_\theta$ is nonempty, convex, and weak$^*$-closed and $\succsim_\theta$ has a \emph{maxmin expected utility} representation with ambiguity set $\mathcal{P}_\theta$: the function
    \begin{equation}\label{eq:loss GS-89}
        V_\theta(L)=\max_{P \in \mathcal{P}_\theta} \int L_\theta \ dP
    \end{equation}
    represents $\succsim_\theta$ on $\mathcal{L}$.
    \item[(ii)] Suppose $\succsim_\theta$ satisfies Axioms \ref{ax:7}-\ref{ax:ten}, Certainty Independence, and Axiom \ref{ax:5}. It additionally satisfies Axiom \ref{ax:grounded at Q} if and only if $Q_\theta \in \mathcal{P}_\theta$.
    \item[(iii)] Suppose $\succsim_\theta$ satisfies Axioms \ref{ax:7}-\ref{ax:ten}, Certainty Independence, and Axiom \ref{ax:5}. It additionally satisfies Axiom \ref{ax:18} if and only if it satisfies Axioms \ref{ax:14} and \ref{ax:15}. In this case,
    \[
    V_\theta(L)=\sup_{P \in \mathcal{P}_\theta \cap \Delta(\mathcal{X})} \int L_\theta \ dP
    \]
\end{itemize}
\end{thm}

\textbf{Proof of Theorem \ref{thm:GS-expanded}.} Necessity of the axioms (the backwards directions of each part) is straightforward, so we prove sufficiency.

\underline{Part (i)}: By Lemmas \ref{lem:utility acts} and \ref{lem: simple-to-bounded} and Proposition 19(iii) of \cite{maccheroni2006ambiguity}, $c_\theta$ takes only values $0$ and $\infty$. By Lemma 32 of \cite{maccheroni2006ambiguity},
\[
\mathcal{P}_\theta=\{P \in \Delta^F(\mathcal{X}): c_\theta(P)=0\}
\]
Since $c_\theta: \Delta^F(\mathcal{X})\rightarrow [0,\infty]$ is weak$^*$ lower semicontinuous and convex, $\mathcal{P}_\theta$ is weak$^*$ closed and convex. Since $\Delta^F(\mathcal{X})$ is weak$^*$ compact, $0=\inf_{P \in \Delta^F(\mathcal{X})} c_\theta(P)=\min_{P \in \Delta^F(\mathcal{X})} c_\theta(P)$, so $\mathcal{P}_\theta$ is nonempty. Finally, Equation (\ref{eq:loss GS-89}) follows from Proposition 19(ii) of \cite{maccheroni2006ambiguity}.

\underline{Part (ii)}: This immediately follows from Lemma 32 of \cite{maccheroni2006ambiguity} and Theorem \ref{thm:D-rep}(iii).

\underline{Part (iii)}: Suppose $\succsim_\theta$ satisfies Axioms \ref{ax:7}-\ref{ax:ten}, Certainty Independence, and Axiom \ref{ax:5}. By Lemma 32 of \cite{maccheroni2006ambiguity}, for any $t\geq 0$,
\[
\mathcal{P}_\theta=\{P \in \Delta^F(\mathcal{X}): c_\theta(P)=0\}=\{P \in \Delta^F(\mathcal{X}): c_\theta(P)\leq t\}
\] 
Hence, $\succsim_\theta$ satisfies Axiom \ref{ax:18} if and only if it satisfies Axioms \ref{ax:14} and \ref{ax:15}. The desired representation then follows from Theorem \ref{thm:D-rep}(iv).

\paragraph{Multiplier Preferences.}
Recall that an event $\mathcal{E}\subseteq \mathcal{X}$ is \emph{nonnull under $\succsim_\theta$} if there exist $L,L',M \in \mathcal{L}$ such that $L_\mathcal{E} M \succ_\theta L_\mathcal{E}' M$.

\begin{lemma}\label{lem:simple-nonnull}
Under the basic axioms, if $\mathcal{E}$ is nonnull under $\succsim_\theta$, then there exist $L,L',M \in \mathcal{L}_0$ such that $L_\mathcal{E} M \succ_\theta L_\mathcal{E}' M$.
\end{lemma}

\paragraph{Proof of Lemma \ref{lem:simple-nonnull}.} Let $L,L',M \in \mathcal{L}$ such that $L_\mathcal{E} M \succ_\theta L_\mathcal{E}' M$. By $\theta$-Relevance and since $\mathcal{E}\subseteq \mathcal{X}$, $(L_\theta)_\mathcal{E} M_\theta \succ_\theta (L_\theta')_\mathcal{E} M_\theta$. For each $k \geq 1$, define $L_{\theta,k}$ and $M_{\theta,k}$ exactly analogously to the definition of $f_{\theta,k}$ in the proof of Lemma \ref{lem: simple-to-bounded}. Then, Monotonicity implies:
\[
(L_{\theta,k})_\mathcal{E} (M_{\theta,k})\succsim_\theta (L_\theta)_\mathcal{E} M_\theta \succ_\theta (L_\theta')_\mathcal{E} M_\theta
\]
Next, for each $j\geq 1$, define $L_{\theta,j}'$ exactly analogously to the definition of $g_{\theta,j}$ in the proof of Lemma \ref{lem: simple-to-bounded}. By construction, $\|(L_{\theta,j}')_\mathcal{E} M_{\theta,j}-(L_\theta')_\mathcal{E} M_\theta\|_\infty \to 0$ as $j \to \infty$. Hence, by $\theta$-Continuity, there exists $j\geq 1$ large enough such that 
\[
(L_\theta)_\mathcal{E} M_\theta \succ_\theta (L_{\theta,j}')_\mathcal{E} M_{\theta,j}
\]
Conclude by choosing $k=j$. \qed

\medskip

The following result implies Proposition \ref{prop:multiplier}.

\begin{thm}\label{thm:strz-extended}
Throughout this result, assume that for each $\theta \in \Theta$, $\mathcal{X}$ has at least three disjoint events that are nonnull under $\succsim_\theta$. 
\begin{itemize}
    \item[(i)] Fix any $\theta \in \Theta$. The preference $\succsim_\theta$ satisfies Axioms \ref{ax:7}--\ref{ax:grounded at Q}, Monotone Continuity, and the Sure Thing Principle if and only if $\succsim_\theta$ has a \emph{multiplier representation} with reference probability $Q_\theta$: there exists $\lambda_\theta \in (0,\infty]$ such that the function
    \begin{equation}\label{eq:loss-multiplier}
        V_\theta(L)=\max_{P \in \Delta(\mathcal{X})} \left(\int L_\theta \ dP-\lambda_\theta \cdot \KL(P\|Q_\theta) \right)
    \end{equation}
    represents $\succsim_\theta$ on $\mathcal{L}$.
    \item[(ii)] Suppose each $\succsim_\theta$ satisfies the axioms from Part (i). Assume that there exists a collection of events $\{\mathcal{E}_\theta\}_{\theta \in \Theta}$ and $q \in (0,1)$ such that $Q_\theta(\mathcal{E}_\theta)=q$ for all $\theta \in \Theta$. The family of preferences additionally satisfies Uniform Misspecification Concern if and only if $\lambda_\theta$ is constant across $\theta \in \Theta$.
\end{itemize}
\end{thm}

\paragraph{Proof of Theorem \ref{thm:strz-extended}.} Necessity of the axioms (the backwards directions of each part) is straightforward, so we prove sufficiency.

\underline{Part (i)}: Fix any $\theta \in \Theta$. By Lemmas \ref{lem:utility acts}-\ref{lem:simple-nonnull} and Theorem 1 of \cite{strzalecki2011axiomatic},\footnote{More precisely, by the portion of the argument for the proof of Theorem 1 of \cite{strzalecki2011axiomatic} which delivers the representation for utility acts.} there exist $\lambda_\theta \in (0,\infty]$ and $\Tilde{Q}_\theta \in \Delta(\mathcal{X})$ such that
\[
\max_{P \in \Delta(\mathcal{X})} \left(\int L_\theta \ dP-\lambda_\theta \cdot \KL(P\|\Tilde{Q}_\theta)\right)
\]
represents $\succsim_\theta$ on $\mathcal{L}$. By Axiom \ref{ax:grounded at Q} and uniqueness of $c_\theta$, $\lambda_\theta \cdot \KL(Q_\theta\|\Tilde{Q}_\theta)=0$ and hence $\Tilde{Q}_\theta=Q_\theta$, which delivers Equation (\ref{eq:loss-multiplier}).

\underline{Part (ii)}: Suppose that there exist $\theta,\theta'$ such that $\lambda_\theta>\lambda_{\theta'}$. By Lemma \ref{lem:utility acts}, it suffices to show that the induced preferences over utility acts do not satisfy Uniform Misspecification Concern. Define $f_\theta=1_{\mathcal{E}_\theta}0$ and $g_{\theta'}=1_{\mathcal{E}_{\theta'}}0$. Note that the law of $f_\theta$ under $Q_\theta$ coincides with the law of $g_{\theta'}$ under $Q_{\theta'}$: $q\delta_1+(1-q)\delta_0$. For each $\lambda \in (0,\infty]$, define:
\[
\phi_\lambda(x)=\begin{cases} 
      -\exp(-\lambda^{-1}x) & \lambda<\infty \\
      -x & \lambda=\infty
   \end{cases}
\]
By the dual representation of multiplier preferences: for any $\lambda_\theta,\lambda_{\theta'} \in (0,\infty]$,
\[
U_{\theta}(f_\theta)=\phi_{\lambda_\theta}^{-1}\left(\int \phi_{\lambda_\theta}(f_\theta) \ dQ_\theta \right)=\phi_{\lambda_\theta}^{-1}\Big(q\phi_{\lambda_\theta}(1)+(1-q)\phi_{\lambda_\theta}(0) \Big)
\]
and similarly,
\[
U_{\theta'}(g_{\theta'})=\phi_{\lambda_{\theta'}}^{-1}\Big(q\phi_{\lambda_{\theta'}}(1)+(1-q)\phi_{\lambda_{\theta'}}(0) \Big)
\]
It is straightforward to show that, for any $q \in (0,1)$, the map
\[
\lambda \mapsto \phi_{\lambda}^{-1}\left(q\phi_{\lambda}(1)+(1-q)\phi_\lambda(0) \right)=\begin{cases} 
      -\lambda \log\Big(q\exp(-\lambda^{-1})+(1-q) \Big) & \lambda<\infty \\
      q & \lambda=\infty
\end{cases}
\]
is strictly increasing in $\lambda \in (0,\infty]$, and hence
\[
U_{\theta'}(g_{\theta'})<U_{\theta}(f_\theta)
\]
Note that $U_{\theta}$ and $U_{\theta'}$ are normalized. Hence, choosing $k \in (U_{\theta'}(g_{\theta'}),U_{\theta}(f_\theta))$ yields a violation of Uniform Misspecification Concern, since
\[
k \succ_{\theta'} g_{\theta'} \quad \text{and} \quad f_\theta \succ_\theta k
\]
\qed

\medskip

\paragraph{Constrained Multiplier} 
We show the following result, which implies Theorem \ref{thm:constrained-multiplier}. Let $\chi_C$ denote the convex indicator function for a set $C\subseteq \Delta^F(\mathcal{X})$.

\begin{thm}\label{thm: IP-var}
Fix any $\theta \in \Theta$. Suppose that for each $i=1,2$, $\succsim_{\theta,i}$ has a variational representation on $\mathcal{L}$: there exists a convex, weak$^*$ lower-semicontinuous function $c_{\theta,i}: \Delta^F(\mathcal{X}) \rightarrow [0,\infty]$ with $\inf_{P \in \Delta^F(\mathcal{X}) } c_{\theta,i}(P)=0$ such that
    \begin{equation}\label{var-rep-theta-i}
        V_{\theta,i}(L)=\max_{P \in \Delta^F(\mathcal{X})} \left(\int L_\theta \ dP-c_{\theta,i}(P) \right)
    \end{equation}
    represents $\succsim_{\theta,i}$ on $\mathcal{L}$, and additionally suppose that $c_{\theta,i}(Q_\theta)=0$ for each $i=1,2$. $\succsim_\theta$ satisfies the Basic Axioms ($\theta$-Relevance, Nontrivial Weak Order, Monotonicity, $\theta$-Continuity) and $(\succsim_\theta,\succsim_{\theta,1},\succsim_{\theta,2})$ satisfy Indirect Pareto if and only if $\succsim_\theta$ has a variational representation
    \begin{equation}\label{ip-rep}
        V_\theta(L)=\max_{P \in \Delta^F(\mathcal{X})} \left(\int L_\theta \ dP-(c_{\theta,1}(P)+c_{\theta,2}(P)) \right)
    \end{equation}
    on $\mathcal{L}$. In particular for the case where $c_{\theta,1}(P)=\lambda \cdot \KL(P\|Q_\theta)$ and $c_{\theta,2}=\chi_{\mathcal{P}_\theta}(P)$ for some nonempty, convex, weak$^*$ closed $\mathcal{P}_\theta \subseteq \Delta^F(\mathcal{X})$ containing $Q_\theta$, we may write Equation \ref{ip-rep} as:
    \begin{equation}\label{eq:constrained-multiplier-thm-18}
        V_\theta(L)=\sup_{P \in \Delta(\mathcal{X}) \cap \mathcal{P}_\theta} \left(\int L_\theta \ dP-\lambda \KL(P\|Q_\theta)\right)
    \end{equation}
\end{thm}

The following definitions and lemmas will be useful for the proof of Theorem \ref{thm: IP-var}. Let $V$ be a real vector space. The \emph{strict epigraph} of a function $F: V \rightarrow [-\infty,+\infty]$ is the set $\text{epi}_S(F):=\{(v,r) \in V \times \mathbb{R}: F(v)<r\}$.

\begin{defn}[\cite{zalinescu2002convex} Theorem 2.1.3(ix)]
Given functions $F,G: V \rightarrow (-\infty,+\infty]$, their \emph{inf-convolution} is the function $F \Box G: V \rightarrow [-\infty,+\infty]$, where:
\begin{equation}\label{inf-conv-defn}
    (F \Box G)(v)=\inf_{v_1,v_2 \in V: v_1+v_2=v} \Big(F(v_1)+G(v_2) \Big)
\end{equation}
\end{defn}

\begin{lemma}\label{lem:epi-sum}
For any $F,G: V \rightarrow (-\infty,+\infty]$, $\text{epi}_S(F\Box G)=\text{epi}_S(F)+\text{epi}_S(G)$.\footnote{This relation is stated without proof as Equation (2.6) of \cite{zalinescu2002convex}. For completeness, we provide a proof here, which follows immediately from the definitions.}
\end{lemma}

\begin{proof}
First, we show the forwards inclusion: $\text{epi}_S(F \Box G) \subseteq \text{epi}_S(F)+\text{epi}_S(G)$. Note that
\begin{align*}
    (v,r) \in \text{epi}_S(F \Box G) \\
    \iff [F \Box G](v)=\inf_{v_1+v_2=v}\Big(F(v_1)+G(v_2)\Big)<r \\
    \iff \exists \ y_1+y_2=v \text{ s.t. } F(y_1)+G(y_2)<r 
\end{align*}
Let $\alpha:=\frac{1}{2}(r-F(y_1)-G(y_2))>0$, and define $r_1:=F(y_1)+\alpha$ and $r_2:=G(y_2)+\alpha$. Then we have $(y_1,r_1) \in \text{epi}_S(F)$ and $(y_2,r_2) \in \text{epi}_S(G)$ with $(y_1,r_1)+(y_2,r_2)=(v,r)$.

Second, we show the backwards inclusion: $\text{epi}_S(F)+\text{epi}_S(G) \subseteq \text{epi}_S(F \Box G)$. Let $(y_1,r_1) \in \text{epi}_S(F)$ and $(y_2,r_2) \in \text{epi}_S(G)$. Note that by definition,
\[
[F\Box G](y_1+y_2)=\inf_{v_1+v_2=y_1+y_2} \Big(F(v_1)+G(v_2)\Big)\leq F(y_1)+G(y_2)<r_1+r_2
\]
and hence $(y_1+y_2,r_1+r_2) \in \text{epi}_S(F \Box G)$.
\end{proof}

Let $\mathcal{L}(\mathcal{X})$ be the Banach space of bounded, Borel measurable functions $L: \mathcal{X} \rightarrow \mathbb{R}$, endowed with the $\sup$ norm.\footnote{In particular, $\mathcal{L}(\mathcal{X})$ is a separated, locally convex vector space, as in the setup of \cite{zalinescu2002convex} Section 2.3.} Let $ba(\mathcal{X})$ be the real vector space of bounded, finitely additive, signed Borel measures on $\mathcal{X}$. By \cite{aliprantis2006infinite} Theorem 14.4, $ba(\mathcal{X})$ is the topological (norm) dual of $\mathcal{L}(\mathcal{X})$. Recall that we endowed $ba(\mathcal{X})$ (and $\Delta^F(\mathcal{X})$) with the corresponding weak$^*$-topology. Note that under this topology, $ba(\mathcal{X})$ is a separated, locally convex vector space whose topological dual is $\mathcal{L}(\mathcal{X})$.

\begin{lemma}
For each $i=1,2$, let $c_{i}: \Delta^F(\mathcal{X}) \rightarrow [0,\infty]$ be a convex, weak$^*$ lower-semicontinuous function  with $\inf_{P \in \Delta^F(\mathcal{X}) } c_{i}(P)=0$, let $V_i: \mathcal{L}(\mathcal{X}) \rightarrow \mathbb{R}$ be defined as in Equation (\ref{var-rep-theta-i}):
\[
V_i(L)=\max_{P \in \Delta^F(\mathcal{X})} \left(\int L \ dP-c_i(P) \right)
\]
and let $V: \mathcal{L}(\mathcal{X}) \rightarrow \mathbb{R}$ be defined as in Equation (\ref{ip-rep}):
\[
V(L)=\max_{P \in \Delta^F(\mathcal{X})} \left(\int L \ dP-c_1(P)-c_2(P) \right)
\]
Assume that $c_1(Q)=c_2(Q)=0$ for some $Q \in \Delta^F(\mathcal{X})$. Then, $V=V_1 \Box V_2$.
\end{lemma}

\begin{proof}
For each $i=1,2$, define $\overline{c}_{i}: ba(\mathcal{X}) \rightarrow [0,\infty]$ as:
\begin{equation}
    \overline{c}_{i}(P)=\begin{cases} 
      c_{i}(P) & P \in \Delta^F(\mathcal{X}) \\
      +\infty & \text{else}
   \end{cases}
\end{equation}
By definition of conjugate (Equation (2.30) of \cite{zalinescu2002convex}), $V_i=\overline{c}_i^*$. It is straightforward to verify that $\overline{c}_i$ is proper, convex, and weak$^*$ lower-semicontinuous. By the Fenchel--Moreau theorem (Theorem 2.3.3 of \cite{zalinescu2002convex}), $V_i^*=\overline{c}_i^{**}=\overline{c}_i$. By Theorem 2.3.1 of \cite{zalinescu2002convex}, $(V_1 \Box V_2)^*=V_1^*+V_2^*=\overline{c}_1+\overline{c}_2$. Since each $V_i$ is convex and proper, $V_1 \Box V_2$ is convex by Theorem 2.1.3 of \cite{zalinescu2002convex}. 

For ease of notation, let $W=V_1 \Box V_2$. Next, we show that $W$ is 1-Lipschitz, and hence lower-semicontinuous. First, we recall that since each $V_i$ is 1-Lipschitz, for any $L'',H \in \mathcal{L}$,
\[
V_2(L''+H)-V_2(L'')\leq |V_2(L''+H)-V_2(L'')|\leq \|H\|_\infty
\]
Hence, for any $L,L',L'',H \in \mathcal{L}$ with $L=L'+L''$,
\[
W(L+H)\leq V_1(L')+V_2(L''+H)\leq V_1(L')+V_2(L'')+\|H\|_\infty
\]
Taking the $\inf$ over $L',L'' \in \mathcal{L}$ with $L'+L''=L$ yields:
\[
W(L+H)\leq W(L)+\|H\|_\infty
\]
An analogous argument shows that $W(L)\leq W(L+H)+\|H\|_\infty$. Hence, $W$ is 1-Lipschitz. Finally, we show that $W$ is proper. By definition, $W<+\infty$. By assumption, for any $L,L' \in \mathcal{L}$, we have $V_1(L')\geq \int L' \ dQ$ and $V_2(L-L')\geq \int (L-L') \ dQ$. Hence,
\[
W(L)\geq \int L \ dQ>-\infty
\]
Since $W$ is proper, convex, and lower semicontinuous, another application of the Fenchel--Moreau theorem yields: $W^{**}=W$. Taking conjugates of both sides of $W^*=\overline{c}_1+\overline{c}_2$ yields
\[
V_1 \Box V_2=(\overline{c}_1+\overline{c}_2)^*=V
\]
as desired.
\end{proof}

\begin{lemma}\label{lem:IP-char}
Assume the assumptions of Theorem \ref{thm: IP-var}. $(\succsim_\theta,\succsim_{\theta,1},\succsim_{\theta,2})$ satisfy Indirect Pareto if and only if: for each $L \in \mathcal{L}$ and $r \in \mathbb{R}$,
\[
L \succ_\theta r \iff V_\theta(L_\theta)<r
\]
where $V_\theta$ is defined as in Equation (\ref{ip-rep}).
\end{lemma}

\begin{proof}
Fix any $i=1,2$. Note that: for each $L \in \mathcal{L}$ and $r \in \mathbb{R}$,
\begin{align*}
    L \succ_{\theta,i} r \iff V_{\theta,i}(L_\theta)<r \iff (L_\theta,r) \in \text{epi}_S(V_{\theta,i})
\end{align*}
Hence, $(\succsim_\theta,\succsim_{\theta,1},\succsim_{\theta,2})$ satisfy Indirect Pareto if and only if: for each $L \in \mathcal{L}$ and $r \in \mathbb{R}$,
\begin{align*}
    L \succ_\theta r \iff \\
    \exists L_1,L_2 \in \mathcal{L}, \ r_1,r_2 \in \mathbb{R} \text{ s.t. } L=L_1+L_2, \ r=r_1+r_2, \ (L_{i\theta},r_i) \in \text{epi}_S(V_{\theta,i}) \\
    \iff (L_\theta,r) \in \text{epi}_S(V_{\theta,1})+\text{epi}_S(V_{\theta,2})=\text{epi}_S(V_{\theta,1}\Box V_{\theta,2})=\text{epi}_S(V_\theta)
    \iff V_\theta(L_\theta)<r
\end{align*}
where the last two equalities follow from the previous lemmas.
\end{proof}

\paragraph{Proof of Theorem \ref{thm: IP-var}.}
Necessity of Indirect Pareto immediately follows from Lemma \ref{lem:IP-char}, and necessity of the remaining axioms are straightforward. To prove sufficiency, observe that under the Basic Axioms, 
\[
L \succsim_\theta L' \iff L_\theta \succsim_\theta L_\theta' \iff \inf\{r: L_\theta \succ_\theta r\}\leq \inf\{r: L_\theta' \succ_\theta r\}
\]
where the first equivalence follows from $\theta$-Relevance. For the forwards direction, observe that if $L_\theta \succsim_\theta L_\theta'$ then $\{r: L_\theta' \succ_\theta r\} \subseteq \{r: L_\theta \succ_\theta r\}$. For the backwards direction, assume that $L_\theta \succ_\theta L_\theta'$. By Monotonicity and $\theta$-Continuity, there exists $r,r' \in \mathbb{R}$ such that $L_\theta \succ_\theta r\succ_\theta r' \succ_\theta L_\theta'$. Hence,
\[
\inf\{r: L_\theta \succ_\theta r\}\leq r<r' \leq \inf\{r: L_\theta' \succ_\theta r\}
\]
By Lemma \ref{lem:IP-char}, $\inf\{r: L_\theta \succ_\theta r\}=V_\theta(L_\theta)$. Finally, to prove the representation Equation (\ref{eq:constrained-multiplier-thm-18}) in the case where $c_{\theta,1}(P)=\lambda \cdot \KL(P\|Q_\theta)$ and $c_{\theta,2}=\chi_{\mathcal{P}_\theta}(P)$ for some nonempty, convex, weak$^*$ closed $\mathcal{P}_\theta \subseteq \Delta^F(\mathcal{X})$ containing $Q_\theta$, define $c_\theta=c_{\theta,1}+c_{\theta,2}$. For each $t\geq 0$, define 
\[
\{c_\theta\leq t\}=\{P \in \Delta^F(\mathcal{X}): c_\theta(P)\leq t\}
\]
By definition of KL divergence,
\[
\{c_\theta\leq t\}=\{P \in \Delta(\mathcal{X}): \lambda \KL(P\|Q_\theta)\leq t\}\cap \mathcal{P}_\theta \subseteq \Delta(\mathcal{X})
\]
Hence,
\begin{align*}
    \{c_\theta\leq t\}=\{c_\theta\leq t\}\cap \Delta(\mathcal{X}) \\
    \implies \{c_\theta\leq t\}=\overline{\{c_\theta\leq t\}}=\overline{\{c_\theta\leq t\}\cap \Delta(\mathcal{X})}
\end{align*}
Equation (\ref{eq:constrained-multiplier-thm-18}) then follows from Theorem \ref{thm:D-rep}(iv). \qed

\paragraph{Proof of Proposition \ref{prop:constrained-dual init}}

The primal problem is
\[
V_\theta(L)=\sup_{P\in\mathcal{P}_\theta}\left\{\int L(\theta,x)\,dP(x)-\lambda\cdot\KL(P\|Q_\theta)\right\}
\]
where $\mathcal{P}_\theta=\{P:E_P[\varphi(\theta,X)]=0\}$.  Factoring out $\lambda$,
\[
V_\theta(L)=-\lambda\inf_{P:E_P[\varphi]=0}\left\{\KL(P\|Q_\theta)+\int\left(-\frac{1}{\lambda}L(\theta,x)\right)dP(x)\right\}.
\]
Define the tilted distribution $P_0$ by
\[
\frac{dP_0}{dQ_\theta}(x)=\frac{\exp\left(\frac{1}{\lambda}L(\theta,x)\right)}{E_{Q_\theta}\left[\exp\left(\frac{1}{\lambda}L(\theta,X)\right)\right]}.
\]
A direct calculation shows that
\[
\KL(P\|Q_\theta)-\frac{1}{\lambda}\int L(\theta,x)\,dP(x)=\KL(P\|P_0)-\log E_{Q_\theta}\left[\exp\left(\frac{1}{\lambda}L(\theta,X)\right)\right],
\]
so
\[
V_\theta(L)=\lambda\log E_{Q_\theta}\left[\exp\left(\frac{1}{\lambda}L(\theta,X)\right)\right]-\lambda\inf_{P:E_P[\varphi]=0}\KL(P\|P_0).
\]
For the applications below we invoke this result for $L$ such that the right-hand side is finite, since in the case where it is infinite $V_\theta(L)$ is infinite as well. The remaining step is a standard KL minimization argument for linear moment constraints; see, for example, Kitamura (2009), which implies that
\[
\inf_{P:E_P[\varphi]=0}\KL(P\|P_0)=\sup_{\beta\in\mathbb{R}^b}\left\{-\log E_{P_0}\left[\exp\left(\beta'\varphi(\theta,X)\right)\right]\right\}.
\]
Since
\[
E_{P_0}\left[\exp\left(\beta'\varphi(\theta,X)\right)\right]=\frac{E_{Q_\theta}\left[\exp\left(\frac{1}{\lambda}L(\theta,X)+\beta'\varphi(\theta,X)\right)\right]}{E_{Q_\theta}\left[\exp\left(\frac{1}{\lambda}L(\theta,X)\right)\right]},
\]
substituting and simplifying yields
\[
V_\theta(L)=\inf_{\beta\in\mathbb{R}^b}\lambda\cdot\log\left(E_{Q_\theta}\left[\exp\left(\frac{1}{\lambda}L(\theta,X)-\beta'\varphi(\theta,X)\right)\right]\right),
\]
where the sign change on $\beta$ absorbs the negation. \qed

\subsection{Proofs for Results in Section \ref{sec: LAM}}

We first state and prove some auxiliary results which will be useful in the proof of Theorem \ref{thm: Constrained LAM}.

\begin{prop}\label{prop: Extended Representation Thm}
Assume that the model $\{Q_{n,\theta}:\theta\in\Theta\}$ is locally asymptotically normal at $\theta_0\in \mathrm{int}(\Theta)$ with scaling coefficient $\sqrt{n}$ and nonsingular Fisher information $I_0$.  Let
\[
Y_{n,h}=\frac{1}{\sqrt{n}}\sum_{i=1}^n\psi(\theta_{n,h},X_i).
\]
Under Assumption \ref{assu: moments}, if $T_n$ is a sequence of $\mathbb{R}^d$-valued statistics that is tight under $Q_{n,0}$, then for every subsequence there exists a further subsequence $\{s\}$ and a possibly randomized measurable function $t:\mathbb{R}^p\times\mathbb{R}^k\times[0,1]\to\mathbb{R}^d$ such that, with $U\sim \mathrm{Unif}[0,1]$ independent of $(X,Y)$ in the limit experiment \eqref{eq: Limit Experiment} and $T=t(X,Y,U)$,
\[
(T_s,Y_{s,h})\stackrel{d}{\to}(T,\,Y+\Psi h)
\quad\text{under }Q_{s,h}
\]
for every $h\in\mathbb{R}^p$.
\end{prop}

\paragraph{Proof of Proposition \ref{prop: Extended Representation Thm}}
Let $\{n_j\}$ be any subsequence.  Since $T_n$ is tight under $Q_{n,0}$ and $(S_n,Y_{n,0})$ converges jointly by Assumption \ref{assu: moments}, Prokhorov's theorem yields a further subsequence $\{s\}\subseteq\{n_j\}$ along which
\[
(T_s,S_s,Y_{s,0})\stackrel{d}{\to}(T,S,Y)
\quad\text{under }Q_{s,0}.
\]
By LAN and Le Cam's third lemma, for every $h\in\mathbb{R}^p$ the same subsequence converges under $Q_{s,h}$ to a limit law that we again denote by $(T,S,Y)$, where the marginal law of $(S,Y)$ is
\[
\left(
\begin{array}{c}
S\\
Y
\end{array}
\right)
\sim
N\left(
\left(
\begin{array}{c}
I_0 h\\
-\Psi h
\end{array}
\right),
\left(
\begin{array}{cc}
I_0 & -\Psi'\\
-\Psi & \Omega
\end{array}
\right)
\right).
\]
Let $X=I_0^{-1}S$.  Then $(X,Y)$ has distribution $Q_h$ in \eqref{eq: Limit Experiment}.  By the standard representation for randomized procedures in the limit experiment (see e.g. Lemma 7.11 in \citealt{van_der_vaart_asymptotic_1998}), there exists a function $t:\mathbb{R}^p\times\mathbb{R}^k\times[0,1]\to\mathbb{R}^d$ and an independent $U\sim\mathrm{Unif}[0,1]$ such that
\[
T=t(X,Y,U)
\]
under the limit law.  Finally, by contiguity
\[
Y_{s,h}=Y_{s,0}+\Psi h+o_p(1)\stackrel{d}{\to}Y+\Psi h
\]
under $Q_{s,h}$, and the claim follows. \qed

\begin{lemma}\label{lem: beta-objective}
Let $W_{M,h}=w_M(Y+\Psi h)$ be as in \eqref{eq: W_{M,h}}, and let $A$ be a non-negative random variable with $A\ge 1$ almost surely.  Then the function
\[
G_h(\beta)=E_{Q_h}\left[A\exp\left(\beta'W_{M,h}\right)\right],
\qquad
\beta\in\mathbb{R}^{b},
\]
is strictly convex on its effective domain and has compact sublevel sets.
\end{lemma}

\paragraph{Proof of Lemma \ref{lem: beta-objective}}
Let $v\in\mathbb{R}^{b}$ be nonzero.  Since $\Omega$ has full rank, $Y+\Psi h$ has full support on $\mathbb{R}^k$ under $Q_h$.  The random variable $v'W_{M,h}$ is therefore a nontrivial polynomial in $Y+\Psi h$, and is almost surely nonzero. Hence
\[
v'\nabla^2 G_h(\beta)v
=
E_{Q_h}\left[
A\exp\left(\beta'W_{M,h}\right)(v'W_{M,h})^2
\right]
>0
\]
whenever $G_h(\beta)<\infty$, so $G_h$ is strictly convex on its effective domain.

To prove compactness of sublevel sets, let $\{\beta_n\}$ satisfy $\|\beta_n\|\to\infty$.  After passing to a subsequence, write $\beta_n=\|\beta_n\|u_n$ with $u_n\to u$ and $\|u\|=1$.  Then $u'W_{M,h}$ is again a nonzero polynomial, and because each component of $W_{M,h}$ is centered under $Q_h$, $u'W_{M,h}$ has mean zero and hence takes strictly positive values on a nonempty open set of $Y$ values.  Therefore there exist a bounded open set $B\subset\mathbb{R}^k$ and $c>0$ such that $u'w(y)\ge 2c$ for all $y\in B$.  By continuity of $w$, for all sufficiently large $n$ we have $u_n'w(y)\ge c$ on $B$.  Since $A\ge 1$ and $Q_h\{Y+\Psi h\in B\}>0$,
\[
G_h(\beta_n)
\ge
E_{Q_h}\left[
\exp\left(\beta_n'W_{M,h}\right)\mathbf{1}\{Y+\Psi h\in B\}
\right]
\ge
Q_h\{Y+\Psi h\in B\}\exp(c\|\beta_n\|)\to\infty.
\]
Thus every sublevel set of $G_h$ is bounded.  Since $G_h$ is lower semicontinuous, its sublevel sets are thus compact. \qed

\begin{lemma}\label{lem: fixed-h-liminf}
Suppose that for some $h\in\mathbb{R}^p$ and some subsequence $\{s\}$,
\[
(T_{s,h},W_{M,s,h})\stackrel[d]{}\to(T_h,W_{M,h})
\quad\text{under }Q_{s,h},
\]
where $T_{s,h}\in\mathbb{R}^d$ and $W_{M,s,h}\in\mathbb{R}^{b}$.  Then
\[
\liminf_{s\to\infty}
\inf_{\beta\in\mathbb{R}^{b}}
E_{Q_{s,h}}
\left[
\ell^*(T_{s,h})\exp(\beta'W_{M,s,h})
\right]
\ge
\inf_{\beta\in\mathbb{R}^{b}}
E_{Q_h}
\left[
\ell^*(T_h)\exp(\beta'W_{M,h})
\right].
\]
\end{lemma}

\paragraph{Proof of Lemma \ref{lem: fixed-h-liminf}}
For each $s$ and $\beta\in\mathbb{R}^{b}$, define
\[
F_s(\beta)=E_{Q_{s,h}}
\left[
\ell^*(T_{s,h})\exp(\beta'W_{M,s,h})
\right],
\]
and define the limit objective
\[
H(\beta)=E_{Q_h}
\left[
\ell^*(T_h)\exp(\beta'W_{M,h})
\right].
\]
Let
\[
m=\inf_{\beta\in\mathbb{R}^{b}}H(\beta)\in[0,\infty].
\]
We want to show that
\[
\liminf_{s\to\infty}\inf_{\beta\in\mathbb{R}^{b}}F_s(\beta)\ge m.
\]

Let $g_r:\mathbb{R}^d\to\mathbb{R}_+$ be bounded Lipschitz functions such that
$g_r\uparrow \ell^*$ pointwise, where such a sequence exists by Baire's Theorem (see e.g. Theorem 6.4.1 in \citealt{Cobaszetal2019Lipscitz}), and let
$\chi_r:\mathbb{R}^{b}\to[0,1]$ be Lipschitz functions such that
$\chi_r\uparrow 1$ pointwise,
$\chi_r(w)=1$ for $\|w\|\le r$, and
$\chi_r(w)=0$ for $\|w\|\ge r+1$.
For $\beta\in\mathbb{R}^{b}$ define
\[
f_r(t,w;\beta)=g_r(t)\chi_r(w)\exp(\beta'w).
\]
Then for each fixed $\beta$,
\[
0\le f_r(t,w;\beta)\uparrow \ell^*(t)\exp(\beta'w)
\qquad\text{pointwise in }(t,w).
\]
Hence, by the monotone convergence theorem,
\[
H_r(\beta)=E_{Q_h}[f_r(T_h,W_{M,h};\beta)]\uparrow H(\beta)
\qquad\text{for each }\beta.
\]

We now argue by contradiction.  Suppose first that $m<\infty$ and result fails.  Then there exist $\varepsilon>0$, a further subsequence, again denoted $\{s\}$, and vectors
$\beta_s\in\mathbb{R}^{b}$ such that
\[
F_s(\beta_s)\le m-\varepsilon
\qquad\forall s.
\]
Suppose next that $m=\infty$ and that the conclusion fails.  Then there exist a further subsequence, again denoted $\{s\}$, a finite constant $M$, and vectors
$\beta_s\in\mathbb{R}^{b}$ such that
\[
F_s(\beta_s)\le M
\qquad\forall s.
\]
Thus, in either case, there exists a finite constant $C$ such that along a subsequence
\[
F_s(\beta_s)\le C
\qquad\forall s.
\]
We consider two cases.

\textit{Case 1:} $\{\beta_s\}$ is bounded. Passing to a further subsequence if necessary, let $\beta_s\to\beta$.
Fix $r$.  Since $\{\beta_s\}\cup\{\beta\}$ is contained in a compact set, and
$\chi_r(w)=0$ for $\|w\|\ge r+1$, the function
$f_r(t,w;\tilde\beta)$ is bounded and continuous on
$\mathbb{R}^d\times\mathbb{R}^{b}\times K$ for any compact set
$K\supset\{\beta_s:s\ge 1\}\cup\{\beta\}$.
Because
\[
(T_{s,h},W_{M,s,h},\beta_s)\stackrel[d]{}\to (T_h,W_{M,h},\beta),
\]
it follows that
\[
E_{Q_{s,h}}\left[f_r(T_{s,h},W_{M,s,h};\beta_s)\right]
\to
H_r(\beta).
\]
Since
\[
f_r(t,w;\beta_s)\le \ell^*(t)\exp(\beta_s'w)
\qquad\text{pointwise,}
\]
we have
\[
E_{Q_{s,h}}\left[f_r(T_{s,h},W_{M,s,h};\beta_s)\right]\le F_s(\beta_s)\le C
\qquad\forall s.
\]
Therefore
\[
H_r(\beta)\le C
\qquad\forall r.
\]
Letting $r\to\infty$ and using monotone convergence gives
\[
H(\beta)\le C.
\]
If $m<\infty$, then $C<m$, which contradicts the definition of $m$.
If $m=\infty$, then $C<\infty$, which again contradicts the definition of $m$.

\textit{Case 2:} $\|\beta_s\|\to\infty$. Write
\[
\beta_s=c_su_s,
\qquad
c_s=\|\beta_s\|,
\qquad
\|u_s\|=1.
\]
Passing to a further subsequence if necessary, let $u_s\to u$ with $\|u\|=1$.
Since $\Omega$ has full rank, the Gaussian vector $Y+\Psi h$ has full support on
$\mathbb{R}^k$ under $Q_h$.  Hence $u'W_{M,h}$ is a nonzero centered (i.e. mean zero) polynomial in
$Y+\Psi h$, and therefore it is not almost surely zero.  Since it is centered, it follows that
\[
Q_h\{u'W_{M,h}>0\}>0.
\]
Choose $\eta,\delta>0$ such that
\[
Q_h\{u'W_{M,h}>\eta\}>2\delta.
\]
Because
\[
(W_{M,s,h},u_s)\stackrel[d]{}\to (W_{M,h},u),
\]
the continuous mapping theorem implies that
\[
u_s'W_{M,s,h}\stackrel[d]{}\to u'W_{M,h}.
\]
Therefore, by the Portmanteau theorem, for all sufficiently large $s$,
\[
Q_{s,h}\{u_s'W_{M,s,h}>\eta\}\ge \delta.
\]
Since $\ell\ge 0$, we have $\ell^*\ge 1$, so for all sufficiently large $s$,
\[
F_s(\beta_s)
\ge
E_{Q_{s,h}}
\left[
\exp(c_su_s'W_{M,s,h})1\{u_s'W_{M,s,h}>\eta\}
\right]
\ge
\delta e^{c_s\eta}.
\]
Because $c_s\to\infty$, the right-hand side diverges to $\infty$, which contradicts
$F_s(\beta_s)\le C$.

Both cases lead to a contradiction.  Hence
\[
\liminf_{s\to\infty}\inf_{\beta\in\mathbb{R}^{b}}F_s(\beta)\ge m
=
\inf_{\beta\in\mathbb{R}^{b}}H(\beta),
\]
as claimed. \qed

\paragraph{Proof of Theorem \ref{thm: Constrained LAM}}

For $h\in\mathbb{R}^p$, define $T_{n,h}=\sqrt{n}\left(\delta_n(X^n)-\kappa(\theta_{n,h})\right).$
Let $\{n_j\}$ be an arbitrary subsequence.
Suppose first that $\{T_{n_j,0}\}$ is not tight under $Q_{n_j,0}$.
Then there exist $\varepsilon>0$, a further subsequence $\{s\}\subseteq\{n_j\}$,
and an increasing sequence of constants $M_s\to\infty$ such that
\[
Q_{s,0}\{\|T_{s,0}\|>M_s\}\ge\varepsilon
\qquad\forall s.
\]
Since $\ell(u)\to\infty$ as $\|u\|\to\infty$ and $Q_{s,0}\in\mathcal{P}_{s,0}^M$,
$E_{Q_{s,0}}[\ell(T_{s,0})]
\ge
\varepsilon\inf_{\|u\|>M_s}\ell(u)\to\infty.$
Hence along $\{s\}$ the left-hand side of the theorem is equal to $\infty$, and the
desired lower bound is trivial.

We may therefore restrict attention to the case where $\{T_{n_j,0}\}$ is tight.
By Proposition \ref{prop: Extended Representation Thm}, there exists a further
subsequence $\{s\}\subseteq\{n_j\}$ and a possibly randomized statistic
$T=t(X,Y,U)$ in the limit experiment such that for every $h\in\mathbb{R}^p$,
\[
(T_s,Y_{s,h})\stackrel{d}{\to}(T,Y+\Psi h)
\qquad\text{under }Q_{s,h}.
\]
Since
$T_{s,h}=T_s-\sqrt{s}\left(\kappa(\theta_{s,h})-\kappa(\theta_0)\right),$
differentiability of $\kappa$ at $\theta_0$ implies
$T_{s,h}\stackrel{d}{\to}T-Kh
\quad\text{under }Q_{s,h}.$
Hence, by the assumed convergence of moments up to order $M$ and the continuous mapping theorem,
\[
(T_{s,h},W_{M,s,h})\stackrel{d}{\to}(T-Kh,\;W_{M,h})
\quad\text{under }Q_{s,h}
\]
for every $h\in\mathbb{R}^p$.

Let $\mathbb{Q}^p=\{q_1,q_2,\ldots\}$ and define $I_b=\{q_1,\ldots,q_b\}$.  By Lemma \ref{lem: fixed-h-liminf}, for every rational $h\in\mathbb{Q}^p$,
\[
\liminf_{s\to\infty}
\inf_{\beta}
E_{Q_{s,h}}
\left[
\ell^*(T_{s,h})\exp(\beta'W_{M,s,h})
\right]
\ge
\inf_{\beta}
E_{Q_h}
\left[
\ell^*(T-Kh)\exp(\beta'W_{M,h})
\right].
\]
It follows that
\[
\sup_I\liminf_{n\to\infty}\sup_{h\in I}\inf_{\beta}
\lambda\log
E_{Q_{n,h}}
\left[
\ell^*(T_{n,h})\exp(\beta'W_{M,n,h})
\right]
\ge\]
\[
\lim_{b\to\infty}\liminf_{n\to\infty}\sup_{h\in I_b}\inf_{\beta}
\lambda\log
E_{Q_{n,h}}
\left[
\ell^*(T_{n,h})\exp(\beta'W_{M,n,h})
\right]
\ge\]
\[
\sup_{h\in\mathbb{Q}^p}\inf_{\beta}
\lambda\log
E_{Q_h}
\left[
\ell^*(T-Kh)\exp(\beta'W_{M,h})
\right].
\]

We next show that the right-hand side is unchanged if $\mathbb{Q}^p$ is replaced by $\mathbb{R}^p$.  Write
\[
G(h,\beta)=
E_{Q_h}
\left[
\ell^*(T-Kh)\exp(\beta'W_{M,h})
\right].
\]
Using the likelihood ratio of $Q_h$ to $Q_0$, we may write
\[
G(h,\beta)
=
E_{Q_0}
\left[
\ell^*(T-Kh)\exp(\beta'W_{M,h})
\exp\left(h'I_0X-\frac{1}{2}h'I_0h\right)
\right].
\]
For each realization $(T,X,Y)$ the integrand is non-negative and continuous in $(h,\beta)$, so Fatou's lemma implies that $G$ is jointly lower semicontinuous.  Moreover, since $\ell^*\ge 1$,
\[
G(h,\beta)\ge E_{Q_h}\left[\exp(\beta'W_{M,h})\right]
=
E_{Q_0}\left[\exp(\beta'W_{M,0})\right].
\]
By Lemma \ref{lem: beta-objective}, the right-hand side has compact sublevel sets.  Hence Berge's theorem implies that
\[
h\mapsto \inf_{\beta}G(h,\beta)
\]
is lower semicontinuous.  Since $\mathbb{Q}^p$ is dense in $\mathbb{R}^p$, we obtain
\[
\sup_{h\in\mathbb{Q}^p}\inf_{\beta}G(h,\beta)
=
\sup_{h\in\mathbb{R}^p}\inf_{\beta}G(h,\beta).
\]

Finally, the statistic $T=t(X,Y,U)$ may depend on the auxiliary randomization variable $U$.  For each fixed $h$ and $\beta$, the weight $\exp(\beta'W_{M,h})$ depends only on $(X,Y)$, while $\ell^*$ is convex because $\ell$ is convex.  Therefore Jensen's inequality implies
\[
E_{Q_h}\left[
\ell^*(T-Kh)\exp(\beta'W_{M,h})
\right]
\ge
E_{Q_h}\left[
\ell^*(E[T\mid X,Y]-Kh)\exp(\beta'W_{M,h})
\right].
\]
Thus randomization cannot improve the criterion, and the lower bound is further bounded below by the infimum over non-randomized decision rules $\delta:\mathbb{R}^p\times\mathbb{R}^k\to\mathbb{R}^d$.  This shows that for any sequence of sample sizes $\{n_j\}$, we can extract a further subsequence along which the asserted lower bound holds.  It follows that the same lower bound holds for the liminf along the original
sequence of sample sizes $\{n\}$, which proves the theorem. \qed

\paragraph{Proof of Corollary \ref{cor: infinite case}}
By Theorem \ref{thm: Constrained LAM},
\[
\inf_M\sup_I\liminf_{n\to\infty}\sup_{h\in I}\sup_{P\in\mathcal{P}^M_{n,h}}
\left\{
\mathbb{E}_{P}\left[l_n(\delta(X^{n}),\theta_{n,h})\right]-\lambda \KL(P\|Q_{n,h})
\right\}
\]
\[
\ge
\inf_M\inf_{\delta} \sup_{h\in\mathbb{R}^p} \inf_{\beta}
\lambda\cdot\log\left(E_{Q_{h}}\left[\ell^{*}\left(\delta\left(X,Y\right)-Kh\right)\exp\left(\beta'W_{M,h}\right)\right]\right).
\]
By duality,
\[
\inf_{\beta} \lambda \log\left(E_{Q_{h}}\left[\ell^{*}\left(\delta(X,Y)-Kh\right)\exp\left(\beta'W_{M,h}\right)\right]\right)=\]
\[
\sup_{P\in\mathcal{P}^M_h}\left\{E_{P}\left[\ell\left(\delta(X,Y)-Kh\right)\right]-\lambda \KL(P\|Q_h)\right\},
\]
where
\[
\mathcal{P}^M_h=
\left\{
P\in\Delta(\mathbb{R}^{p+k})
:
E_P\left[\prod_{s=1}^k\left(Y_s+(\Psi h)_s\right)^{m_s}\right]
=
E\left[\prod_{s=1}^k \xi_s^{m_s}\right]
\text{ for all }m\in\mathbb{N}_0^k,\ 1\le |m|\le M
\right\}.
\]
Since $\mathcal{P}^M_h$ is decreasing in $M$, the lower bound is bounded below by
\[
\inf_\delta \sup_{h\in\mathbb{R}^p}\sup_{P\in\mathcal{P}^\infty_h}
\left\{E_{P}\left[\ell\left(\delta(X,Y)-Kh\right)\right]-\lambda \KL(P\|Q_h)\right\},
\]
where $\mathcal{P}^\infty_h=\bigcap_{M=1}^\infty \mathcal{P}^M_h$.
Because the normal distribution is determined by its moments, $P\in\mathcal{P}^\infty_h$ if and only if $Y+\Psi h\sim N(0,\Omega)$ under $P$, or equivalently $Y\sim N(-\Psi h,\Omega)$ under $P$.

Using the chain rule for KL divergence,
\[
\KL(P_{X,Y}|Q_{X,Y})=\KL(P_Y|Q_Y)+E_{P_Y}\left[\KL(P_{X|Y}|Q_{X|Y})\right],
\]
we can therefore rewrite the $M=\infty$ problem as
\[
\inf_\delta\sup_{h\in\mathbb{R}^p}\sup_{P_{X|Y}}
E_{Q_{Y,h}}\left[
E_{P_{X|Y}}[\ell(\delta(X,Y)-Kh)\mid Y]
-
\lambda \KL(P_{X|Y}\|Q_{X|Y,h})\right].
\]
Applying Proposition \ref{thm:multiplier-dual} conditional on $Y$ gives
\[
\inf_\delta\sup_{h\in\mathbb{R}^p}
\lambda\cdot E_{Q_{Y,h}}\left[
\log\left(
E_{Q_{X|Y,h}}\left[\ell^*(\delta(X,Y)-Kh)\mid Y\right]
\right)
\right],
\]
which is the desired bound. \qed

\subsection{Proofs for Results in Section \ref{sec: optimal rules}}

\paragraph{Proof of Theorem \ref{thm: equivariant}}

As a first step, it is helpful to note that every distribution in $\mathcal{P}_h^M$
can be obtained from a distribution in $\mathcal{P}_0^M$ by a group transformation.
Specifically, let $g\in\mathbb{R}^p$ act on $Z=(X,Y)$ by
$g\circ Z=(X+g,\;Y-\Psi g).$
If $P$ denotes the distribution of $Z$, write $g\circ P$ for the distribution of
$g\circ Z$.  The induced action on the parameter-distribution pair $(h,P)$ is then
$g\circ(h,P)=(h+g,\;g\circ P).$
\begin{lemma}\label{lem: parameterization}
    $P\in \mathcal{P}_{0}^M \iff g \circ P \in \mathcal{P}_{g}^M$
\end{lemma}
\paragraph{Proof of Lemma \ref{lem: parameterization}} For $W_{M,h}=w_M(Y+\Psi h),$ observe that
    \begin{align*}
        P \in \mathcal{P}_{0}^M \iff \mathbb{E}_P[w_M(Y)] = 0 \iff \mathbb{E}_{g\circ P}[w_{M}(Y+\Psi g)] = 0 \iff g\circ P \in \mathcal{P}_{g}^M,
    \end{align*}
from which the result is immediate. \qed

This lemma implies that we can parameterize $\mathcal{P}_{g}^M$ by $\mathcal{P}_{0}^M.$  Specifically, let $\phi \in \Phi = \mathcal{P}_{0}^M$,
and observe that by Lemma \ref{lem: parameterization}
\begin{align*}
    \sup_{h} \sup_{P\in \mathcal{P}_{h}^M} E_{P}\left[\ell\left(\delta\left(X,Y\right)-Kh\right)\right] - \lambda \KL(P\|Q_h) = \\
    \sup_{h, \phi}E_{P(h,\phi)}\left[\ell\left(\delta\left(X,Y\right)-Kh\right)\right] - \lambda \KL(P(h, \phi)\|Q_h),
\end{align*}
where $P(h, \phi) = h\circ \phi$.  Our statistical model is then $\{P(h,\phi): h \in \mathbb{R}^p, \phi \in \Phi \}.$  We can define the group operation on this parameter space by
\begin{align*}
        g \circ (h,\phi) = (g+h, \phi),
    \end{align*}
and observe that
    \begin{align*}
        P(g\circ(h,\phi)) = P(g+h, \phi) = g \circ P(h, \phi),
    \end{align*}
    so this definition in the reparametrized model is compatible with our earlier definition.

For the purpose of the decision problem that we're considering, it's without loss of generality to focus on $\phi$ which has a probability density function with respect to Lebesgue measure. This is because $Q_0$ is a multivariate normal distribution and if $\phi \not \ll Q_0$, then $\KL(P(0,\phi)\|Q_0) = +\infty$.

Since the decision problem is invariant, by the generalized Hunt-Stein theorem (Theorem 48.16 of \citealt{strasser1985}) and level-compactness of $\ell$, for any $\delta$ there exists an equivariant $\delta^E$ such that 
    \begin{align*}
        \sup_{g \in \mathbb{R}^p}E_{P(g\circ (h,\phi))}\left[\ell\left(\delta\left(X,Y\right)-K(g+h)\right)\right] \geq E_{P(0,\phi)}\left[\ell\left(\delta^E\left(X,Y\right)\right)\right], \forall \phi
    \end{align*}
    Since $g\in G = \mathbb{R}^p$ operates transitively on $H = \mathbb{R}^p$, this is equivalent to
    \begin{align*}
         \sup_{h \in \mathbb{R}^p}E_{P(h,\phi)}\left[\ell\left(\delta\left(X,Y\right)-Kh\right)\right] \geq E_{P(0,\phi)}\left[\ell\left(\delta^E\left(X,Y\right)\right)\right], \forall \phi
    \end{align*}
    Notice that for given $\phi$, 
    \begin{align*}
        \KL(P(h, \phi)\|Q_h) = \KL(P(g+h, \phi)\|Q_{g+h}).
    \end{align*}
    Therefore,
    \begin{align*}
        \sup_{h \in \mathbb{R}^p}E_{P(h,\phi)}\left[\ell\left(\delta\left(X,Y\right)-Kh\right)\right] - \lambda \KL(P(h, \phi)\|Q_h) \geq  \\
        E_{P(0,\phi)}\left[\ell\left(\delta^E\left(X,Y\right)\right)\right] - \lambda \KL(P(0, \phi)\|Q_0), \forall \phi
    \end{align*}
    Take supremum over $\phi$ on both sides of the inequality,
    \begin{align*}
        \sup_{h \in \mathbb{R}^p, \phi \in \Phi}E_{P(h,\phi)}\left[\ell\left(\delta\left(X,Y\right)-Kh\right)\right] - \lambda \KL(P(h, \phi)\|Q_h) \geq  \\
        \sup_{\phi\in \Phi}E_{P(0,\phi)}\left[\ell\left(\delta^E\left(X,Y\right)\right)\right] - \lambda \KL(P(0, \phi)\|Q_0),
    \end{align*}
    as we aimed to show. \qed

\paragraph{Proof of Proposition \ref{prop: best equivariant finite M}}

\paragraph{Part (a).}  For each $\beta$, the minimization of
$E_{Q_0}\left[\ell^*(\delta^E(X,Y))\exp(\beta'W_{M,0})\right]$
over $\delta^E\in\mathcal{D}^E$ is equivalent to minimizing
$E_{Q_0}\left[\ell^*(\delta^E(X,Y))
\frac{\exp(\beta'W_{M,0})}{E_{Q_0}[\exp(\beta'W_{M,0})]}\right],$
which is the problem of finding the best equivariant rule under the tilted distribution with density proportional to
$q_0(X,Y)\exp(\beta'W_{M,0})$.  This is again a location problem for the group $G=\mathbb{R}^p$, since under the group action
\[
(X,Y,h)\mapsto (X+g,Y-\Psi g,h+g),
\]
the term $W_{M,h}$ is preserved:
\[
W_{M,h+g}(X+g,Y-\Psi g)=W_{M,h}(X,Y).
\]
Hence Theorem 6.5 of \citet{Eaton1989} implies that the best equivariant rule is the Bayes rule under the right Haar, that is, flat, prior.  The posterior density under this prior is
\[
\pi_\beta(h\mid X,Y)=
\frac{q_0(X-h,Y+\Psi h)\exp(\beta'W_{M,h})}
{\int q_0(X-h',Y+\Psi h')\exp(\beta'W_{M,h'})\,dh'},
\]
and the optimal rule is
\[
\delta^*_\beta(X,Y)\in
\argmin_{a\in\mathbb{R}^d}
\int \ell^*(a-Kh)\pi_\beta(h\mid X,Y)\,dh.
\]

\paragraph{Part (b).}  
Let $Z^I=Y+\Psi X$, and observe that this is a maximal invariant for our problem.  Since the difference between any two equivariant rules is invariant, and $KX$ is equivariant, any equivariant rule can be written as
\[
\delta^E(X,Y)=KX+d(Z^I)
\]
for some function $d:\mathbb{R}^k\to\mathbb{R}^d$.
Under $Q_0$, the pair $(X,Z^I)$ is jointly Gaussian with
\[
\Cov_{Q_0}(X,Z^I)=\Cov_{Q_0}(X,Y)+\Cov_{Q_0}(X,\Psi X)
=-I_0^{-1}\Psi'+I_0^{-1}\Psi'=0,
\]
so $X$ and $Z^I$ are independent.  Writing $W_{M,0}=w_M(Y)=w_M(Z^I-\Psi X)$ for the polynomial $w_M$ defining the moment vector, we obtain
\[
E_{Q_0}\left[\ell^*(\delta^E(X,Y))\exp(\beta'W_{M,0})\right]
=
E_{Q_{Z^I,0}}\left[\Gamma_\beta(d(Z^I),Z^I)\right],
\]
where
\[
\Gamma_\beta(a,z)=
E_{Q_{X,0}}
\left[
\ell^*(KX+a)\exp\left(\beta'w_M(z-\Psi X)\right)
\right].
\]

For each fixed $z$, the function $(a,\beta)\mapsto \Gamma_\beta(a,z)$ is jointly lower semicontinuous by Fatou's lemma.  Moreover, for every compact set $B\subset\mathbb{R}^b$,
\[
\inf_{\beta\in B}\Gamma_\beta(a,z)\to\infty
\qquad\text{as }\|a\|\to\infty.
\]
To see this, note that $P_{Q_{0}}\{\|X\|\le 1\}>0$, that $\ell^*(Kx+a)\to\infty$ uniformly over $\|x\|\le 1$ as $\|a\|\to\infty$, and that
\[
\inf_{\beta\in B,\ \|x\|\le 1}\exp\left(\beta'w_M(z-\Psi x)\right)>0
\]
because $B$ and $\{z-\Psi x:\|x\|\le 1\}$ are compact.  Consequently, we may define
\[
g_\beta(z)=\inf_{a\in\mathbb{R}^d}\Gamma_\beta(a,z)=\min_{a\in\mathbb{R}^d}\Gamma_\beta(a,z),
\] which is lower semicontinuous in $(\beta,z)$ by Berge's theorem of the maximum, and
\[
m(\beta)=\min_{\delta^E\in\mathcal{D}^E}
E_{Q_0}\left[\ell^*(\delta^E(X,Y))\exp(\beta'W_{M,0})\right]
=
E_{Q_{Z^I,0}}\left[g_\beta(Z^I)\right].
\]

Since $g_\beta(Z^I)\ge 0$, lower semicontinuity of $ g_\beta(z)$ and Fatou's lemma imply that $m$ is lower semicontinuous.  In addition, since $\ell^*\ge 1$,
\[
m(\beta)\ge E_{Q_0}\left[\exp(\beta'W_{M,0})\right].
\]
By Lemma \ref{lem: beta-objective} with $A=1$, the right-hand side has compact sublevel sets.  Hence $m$ also has compact sublevel sets.  Therefore $m$ attains its infimum at some $\beta^*\in\mathbb{R}^b$.

Part (c).  By part (a), $\delta^*_{\beta^*}$ minimizes the joint objective \eqref{eq: joint opt} over $(\delta^E,\beta)$.  By Theorem \ref{thm: equivariant}, the minimax value over all decision rules equals the minimax value over equivariant rules, and by the duality of Proposition \ref{prop:constrained-dual init} the latter equals \eqref{eq: joint opt}.  Hence $\delta^*_{\beta^*}$ is minimax optimal. \qed

\paragraph{Proof of Proposition \ref{prop: convexity}}

\paragraph{Part (a).}  The function $\ell$ is convex by Assumption \ref{assu: loss}, so $\Gamma\mapsto\ell(\Gamma\phi)$ is convex in $\Gamma$.  The product $\ell^*(\Gamma\phi)\exp(\beta'W_{M,0})=\exp(\ell(\Gamma\phi)/\lambda+\beta'W_{M,0})$ is the exponential of a function that is jointly convex in $(\Gamma,\beta)$, and the exponential of a convex function is convex.  Taking expectations preserves convexity.

\paragraph{Part (b).}  Fix $\Gamma_1,\Gamma_2\in\mathbb{R}^{d\times J}$ and $t\in[0,1]$, and let
$\Gamma_t=t\Gamma_1+(1-t)\Gamma_2$.  Since $\Gamma\mapsto \Gamma\phi(X,Y)$ is linear and
$\ell$ is convex, $\ell(\Gamma_t\phi)\le t\,\ell(\Gamma_1\phi)+(1-t)\,\ell(\Gamma_2\phi)$
pointwise.  Exponentiating and using $\ell^*(u)=\exp(\ell(u)/\lambda)$ gives
\[
\ell^*(\Gamma_t\phi)\le \ell^*(\Gamma_1\phi)^t\,\ell^*(\Gamma_2\phi)^{1-t}
\]
pointwise.  Hence, by Hölder's inequality,
\[
E_{Q_{X|Y,0}}[\ell^*(\Gamma_t\phi)\mid Y]
\le
E_{Q_{X|Y,0}}[\ell^*(\Gamma_1\phi)\mid Y]^t
E_{Q_{X|Y,0}}[\ell^*(\Gamma_2\phi)\mid Y]^{1-t}.
\]
Taking logs yields
\[
\log E_{Q_{X|Y,0}}[\ell^*(\Gamma_t\phi)\mid Y]
\le
t\log E_{Q_{X|Y,0}}[\ell^*(\Gamma_1\phi)\mid Y]
+
(1-t)\log E_{Q_{X|Y,0}}[\ell^*(\Gamma_2\phi)\mid Y].
\]
Thus $\Gamma\mapsto \log E_{Q_{X|Y,0}}[\ell^*(\Gamma\phi)\mid Y]$
is convex for each $Y$, and taking expectation over $Q_{Y,0}$ preserves convexity. \qed

\paragraph{Proof of Proposition \ref{prop: optim estimator square loss}}

We interpret the objective throughout in the extended-real sense.  If every equivariant rule has infinite risk, the proposition is immediate.  Hence we may restrict attention to the case where at least one equivariant rule has finite risk.

Let $Z^I=Y+\Psi X$, and observe that this is a maximal invariant for our problem.  Since the difference between any two equivariant rules is invariant, and $KX$ is equivariant, any equivariant rule can be written as
\[
\delta^E(X,Y)=KX+d(Z^I)
\]
for some function $d:\mathbb{R}^k\to\mathbb{R}^d$.

For $M<\infty$, define
\[
\mathcal{R}_M(d,\beta)=E_{Q_0}\left[\exp\left(\frac{1}{\lambda}\|KX+d(Z^I)\|^2+\beta'W_{M,0}\right)\right].
\]

\paragraph{Part (a): $M=0$ and $M=1$}
For $M=0$, there is no $\beta$.  Since $X\perp Z^I$ under $Q_0$, we have
\[
\mathcal{R}_0(d)=E_{Q_{Z^I,0}}\left[\varphi(d(Z^I))\right],\text{ where }
\varphi(a)=E_{Q_{X,0}}\left[\exp\left(\frac{1}{\lambda}\|KX+a\|^2\right)\right].
\]
The map $a\mapsto\varphi(a)$ is strictly convex, and by symmetry of the distribution of $X$ under $h=0$,
\[
\nabla\varphi(0)=\frac{2}{\lambda}E_{Q_{X,0}}\left[KX\exp\left(\frac{1}{\lambda}\|KX\|^2\right)\right]=0.
\]
Hence $\varphi$ is uniquely minimized at $a=0$, so $\mathcal{R}_0(d)$ is minimized by $d\equiv 0$, i.e. $\delta^*(X,Y)=KX$.

For $M=1$, write  $W_{1,0}=Y$.  Using $Y=Z^I-\Psi X$ and $X\perp Z^I$:
\[
\mathcal{R}_1(d,\beta)=E_{Q_{Z^I,0}}\left[e^{\beta'Z^I}H_\beta(d(Z^I))\right],\qquad
H_\beta(a)=E_{Q_{X,0}}\left[\exp\left(\frac{1}{\lambda}\|KX+a\|^2-\beta'\Psi X\right)\right].
\]
For each $\beta$, $H_\beta(\cdot)$ has a unique minimizer $a_\beta$ so long as the minimized value is finite, and
\[
\inf_d\mathcal{R}_1(d,\beta)=E_{Q_0}\left[e^{\beta'Z^I}\right]H_\beta(a_\beta).
\]
Thus
\[
\inf_{d,\beta}\mathcal{R}_1(d,\beta)=\inf_{a,\beta}E_{Q_0}\left[e^{\beta'Z^I}\right]H_\beta(a).
\]
The objective is jointly convex in $(a,\beta)$.  At $(a,\beta)=(0,0)$,
\[
\nabla_a\mathcal{R}_1(0,0)=\frac{2}{\lambda}E_{Q_0}\left[KX\exp\left(\frac{1}{\lambda}\|KX\|^2\right)\right]=0,
\]
and
\[
\nabla_\beta\mathcal{R}_1(0,0)=E_{Q_0}\left[Y\exp\left(\frac{1}{\lambda}\|KX\|^2\right)\right]=0,
\]
again by symmetry of the distribution of $X$.  Convexity therefore implies $(a,\beta)=(0,0)$ is globally optimal, so $\delta^*(X,Y)=KX$ for $M=1$ as well.

\paragraph{Part (b): $M\ge2$}
Define the primal constrained multiplier risk
\[
\mathcal{R}_M(\delta^E)=\sup_{P\in\mathcal{P}_{0}^M}\left\{E_P\left[\|\delta^E(X,Y)\|^2\right]-\lambda \KL(P\|Q_0)\right\}.
\]
For the $M=\infty$ problem, write
\[
R_\infty(d)=\lambda E_{Q_{Y,0}}\left[\log G_d(Y)\right],\qquad
G_d(y)=E_{Q_{X|Y,0}}\left[\exp\left(\frac{1}{\lambda}\|KX+d(Z^I)\|^2\right)\Bigm|Y=y\right],
\]
where $\delta^E(X,Y)=KX+d(Z^I)$.

\begin{lemma}\label{lem: subgradient M infinity square}
Suppose $\ell(u)=\|u\|^2$ and $R_\infty(d)<\infty$.  Define
\[
w_d(X,Y)=\frac{\exp\left(\frac{1}{\lambda}\|KX+d(Z^I)\|^2\right)}{G_d(Y)}
\]
and
\[
g_d(Z^I)=2E_{Q_0}\left[\left(KX+d(Z^I)\right)w_d(X,Y)\mid Z^I\right].
\]
Then for every measurable perturbation $\gamma(Z^I)$ such that
$R_\infty(d+\gamma)<\infty$ and
$E_{Q_0}\!\left[\left|g_d(Z^I)'\gamma(Z^I)\right|\right]<\infty$,
\[
R_\infty(d+\gamma)\ge R_\infty(d)+E_{Q_0}\left[g_d(Z^I)'\gamma(Z^I)\right].
\]
In particular, if $g_d(Z^I)=0$ almost surely, then $d$ is globally optimal for the $M=\infty$ problem.
\end{lemma}

\paragraph{Proof of Lemma \ref{lem: subgradient M infinity square}}
For each $y$, define the convex functional
\[
\mathcal F_y(u)=\lambda\log E_{Q_{X|Y,0}}\left[\exp\left(\frac{1}{\lambda}u(X,y)\right)\mid Y=y\right]
\]
on the set of $u(\cdot,y)$ for which the expectation is finite.  Let
$u_d(X,Y)=\|KX+d(Z^I)\|^2$
and define the tilted conditional density
\[
w_d(X,Y)=\frac{\exp\left(\frac{1}{\lambda}u_d(X,Y)\right)}{E_{Q_{X|Y,0}}\left[\exp\left(\frac{1}{\lambda}u_d(X,Y)\right)\mid Y\right]}.
\]
Then for any perturbation $v(X,Y)$ such that both sides are finite,
\[
\mathcal F_Y(u_d+v)-\mathcal F_Y(u_d)
=
\lambda\log E_{Q_{X|Y,0}}\left[\exp\left(\frac{1}{\lambda}v(X,Y)\right)w_d(X,Y)\mid Y\right].
\]
If we apply Jensen's inequality to the tilted measure with Radon-Nikodym derivative
$w_d(\cdot,Y)$ with respect to $Q_{X|Y,0}(\cdot\mid Y)$, we thus have
\[
\mathcal F_Y(u_d+v)-\mathcal F_Y(u_d)
\ge
E_{Q_{X|Y,0}}[v(X,Y)w_d(X,Y)\mid Y].
\]

Now take
\[
v(X,Y)=u_{d+\gamma}(X,Y)-u_d(X,Y)
=
2\left(KX+d(Z^I)\right)'\gamma(Z^I)+\|\gamma(Z^I)\|^2.
\]
Taking expectations over $Q_{Y,0}$ yields
$R_\infty(d+\gamma)-R_\infty(d)
\ge
E_{Q_0}\left[v(X,Y)w_d(X,Y)\right].$ 

Define
\[
g_d(Z^I)=2E_{Q_0}\left[\left(KX+d(Z^I)\right)w_d(X,Y)\mid Z^I\right].
\]
Then
\[
E_{Q_0}[v(X,Y)w_d(X,Y)]
=
E_{Q_0}\left[g_d(Z^I)'\gamma(Z^I)\right]
+
E_{Q_0}\left[w_d(X,Y)\|\gamma(Z^I)\|^2\right]
\ge
E_{Q_0}\left[g_d(Z^I)'\gamma(Z^I)\right].
\]
Therefore
\[
R_\infty(d+\gamma)\ge R_\infty(d)+E_{Q_0}\left[g_d(Z^I)'\gamma(Z^I)\right].
\]
Thus $g_d$ is a subgradient of $R_\infty$ at $d$.  If $g_d(Z^I)=0$ almost surely, then $0$ is a subgradient at $d$, so convexity of $R_\infty$ implies that $d$ is globally optimal. \qed

Now consider the linear class
\[
\delta_C(X,Y)=KX+CZ^I,
\qquad
C\in\mathbb{R}^{d\times k},
\]
and define $r(C)=R_\infty(\delta_C)$.
Since $\lambda\log E[e^{U/\lambda}\mid Y]\ge E[U\mid Y]$ by Jensen's inequality,
\[
r(C)\ge E_{Q_0}\left[\|KX+CZ^I\|^2\right]
=E_{Q_0}\left[\|KX\|^2\right]+\tr\left(C\,\Var_{Q_0}(Z^I)\,C'\right),
\]
where we used $E_{Q_0}[XZ^{I\prime}]=0$.  Let $\Sigma_I=\Var_{Q_0}(Z^I)$ and let $\Pi$ be the orthogonal projection onto the range of $\Sigma_I$.  Since $Z^I=\Pi Z^I$ almost surely, $r(C)=r(C\Pi)$, and on this reduced space $\Sigma_I$ is positive definite.  Hence $\tr(C\Sigma_I C')\to\infty$ whenever $\|C\Pi\|\to\infty$, so sublevel sets of $r$ are compact, and there exists a minimizing value $C^*$ whenever $\inf_{C}r(C)$ is finite.

Define
\[
\frac{dP^\infty_C}{dQ_0}(X,Y)=
\frac{\exp\left(\frac{1}{\lambda}\|\delta_C(X,Y)\|^2\right)}{E_{Q_{X|Y,0}}\left[\exp\left(\frac{1}{\lambda}\|\delta_C(X,Y)\|^2\right)\mid Y\right]}.
\]
For any matrix $D\in\mathbb{R}^{d\times k}$, consider $\gamma_D(Z^I)=DZ^I$.  By Lemma \ref{lem: subgradient M infinity square}, optimality of $C^*$ relative to the class of linear rules implies
\[
0=E_{Q_0}\left[g_{d_{C^*}}(Z^I)'DZ^I\right]
=2E_{P^\infty_{C^*}}\left[\delta_{C^*}(X,Y)'DZ^I\right]
\quad\text{for all }D,
\]
and thus
\[
E_{P^\infty_{C^*}}\left[\delta_{C^*}(X,Y)Z^{I\prime}\right]=0.
\]
Since $\delta_C$ is linear and we consider quadratic loss,  $(X,Y)$ is Gaussian under $P^\infty_{C^*}$, so $(\delta_{C^*}(X,Y),Z^I)$ is jointly Gaussian.  Moreover, $E_{P^\infty_{C^*}}[\delta_{C^*}(X,Y)]=0$, since the marginal distribution of $Y$ under $P^\infty_{C^*}$ is $Q_{Y,0}$ and the tilted conditional mean of $X$ given $Y$ is linear in $Y$ with no constant term.  Since uncorrelated jointly Gaussian random vectors are independent,
\[
E_{P^\infty_{C^*}}\left[\delta_{C^*}(X,Y)\mid Z^I\right]=0.
\]
Equivalently, $g_{d_{C^*}}(Z^I)=0$ almost surely, so Lemma \ref{lem: subgradient M infinity square} implies that $\delta_{C^*}$ is globally optimal for the $M=\infty$ problem over all equivariant rules.

Next fix any linear rule $\delta_C$ and define
\[
\phi_C(y)=\lambda\log E_{Q_{X|Y,0}}\left[\exp\left(\frac{1}{\lambda}\|\delta_C(X,Y)\|^2\right)\mid Y=y\right].
\]
Since $X|Y=y$ is Gaussian and $\delta_C$ is affine in $X$, $\phi_C(y)$ corresponds to the moment generating function of a non-central $\chi^2$ distribution and (when finite) is quadratic in $y$.

For any $P\in\mathcal{P}^2_0$, decompose
$P=P_{X|Y}P_Y$ and
$Q_0=Q_{X|Y,0}Q_{Y,0}$.
Then
\begin{align*}
&E_P\left[\|\delta_C(X,Y)\|^2\right]-\lambda \KL(P\|Q_0) \\
=\;&E_{P_Y}\left[E_{P_{X|Y}}\left[\|\delta_C(X,Y)\|^2\mid Y\right]
-\lambda \KL\left(P_{X|Y}\|Q_{X|Y,0}(\cdot\mid Y)\right)\right]-\lambda \KL(P_Y\|Q_{Y,0}).
\end{align*}
Applying Proposition \ref{thm:multiplier-dual} conditional on $Y$ gives
\[
E_{P_{X|Y}}\left[\|\delta_C(X,Y)\|^2\mid Y\right]
-\lambda \KL\left(P_{X|Y}\|Q_{X|Y,0}(\cdot\mid Y)\right)
\le \phi_C(Y),
\]
from which it follows that
\[
E_P\left[\|\delta_C(X,Y)\|^2\right]-\lambda \KL(P\|Q_0)
\le
E_{P_Y}[\phi_C(Y)]-\lambda \KL(P_Y\|Q_{Y,0}).
\]
Because $P\in\mathcal{P}^2_0$ imposes $E_P[Y]=0$ and $E_P[YY']=\Omega$, and $\phi_C$ is quadratic,
\[
E_{P_Y}[\phi_C(Y)]=E_{Q_{Y,0}}[\phi_C(Y)].
\]
Hence
\[
\mathcal{R}_2(\delta_C)\le E_{Q_{Y,0}}[\phi_C(Y)]=\mathcal{R}_\infty(\delta_C).
\]
The reverse inequality $\mathcal{R}_2(\delta_C)\ge\mathcal{R}_\infty(\delta_C)$ is immediate from
$\mathcal{P}^\infty_0\subseteq\mathcal{P}^2_0$, where
$\mathcal{P}^\infty_0=\cap_{m\ge1}\mathcal{P}^m_0$.
Therefore $\mathcal{R}_2(\delta_C)=\mathcal{R}_\infty(\delta_C)$.
Moreover, $P^\infty_C$ attains $\mathcal{R}_\infty(\delta_C)$, is Gaussian, and has $Y\sim Q_{Y,0}$, so $P^\infty_C\in\mathcal{P}^M_0$ for every $M\ge2$.

Using $\mathcal{P}^\infty_0\subseteq \mathcal{P}^M_0\subseteq \mathcal{P}^2_0$ for $M\ge2$, we obtain for every linear rule $\delta_C$,
\[
\mathcal{R}_M(\delta_C)=\mathcal{R}_\infty(\delta_C).
\]
Applying this at $C=C^*$,
\[
\mathcal{R}_M(\delta_{C^*})=\mathcal{R}_\infty(\delta_{C^*})=\inf_C\mathcal{R}_\infty(\delta_C).
\]
Finally, for any equivariant $\delta^E$ and $M\ge2$,
\[
\mathcal{R}_M(\delta^E)\ge\mathcal{R}_\infty(\delta^E)\ge \mathcal{R}_\infty(\delta_{C^*})
=\mathcal{R}_M(\delta_{C^*}),
\]
so $\delta_{C^*}$ is minimax optimal for each $M\ge2$, and the optimal value is the same for all $M\ge2$. \qed
\paragraph{Proof of Proposition \ref{prop: Convergence - Estimation Error}}
For $h\in\mathbb{R}^p$, define
$T_{n,h}=\sqrt{n}\left(\delta_n^c-\kappa(\theta_{n,h})\right).$
By the definition of $\delta_n^c$,
\[
T_{n,h}
=
\sqrt{n}\left(\kappa(\hat\theta_n^{MLE})-\kappa(\theta_{n,h})\right)
+\delta^c\left(0,\frac{1}{\sqrt{n}}\sum_{i=1}^n\psi(\hat\theta_n^{MLE},X_i);\hat K_n,\hat\Sigma_n\right).
\]
By differentiability of $\kappa$ at $\theta_0$ and Assumption \ref{assu: smoothness},
\[
\sqrt{n}\left(\kappa(\hat\theta_n^{MLE})-\kappa(\theta_{n,h})\right)\stackrel[d]{Q_{n,h}}{\to}K(X-h).
\]
Again by Assumption \ref{assu: smoothness}, under $Q_{n,h}$,
\[
\left(\frac{1}{\sqrt{n}}\sum_{i=1}^n\psi(\hat\theta_n^{MLE},X_i),\hat K_n,\hat\Sigma_n\right)\stackrel[d]{Q_{n,h}}{\to}(Y+\Psi X,K,\Sigma),
\]
so continuity of $\delta^c$ and the continuous mapping theorem give
\[
\delta^c\left(0,\frac{1}{\sqrt{n}}\sum_{i=1}^n\psi(\hat\theta_n^{MLE},X_i);\hat K_n,\hat\Sigma_n\right)\stackrel[d]{Q_{n,h}}{\to}\delta^c(0,Y+\Psi X;K,\Sigma).
\]
Using equivariance with $g=-X$,
\[
\delta^c(0,Y+\Psi X;K,\Sigma)=\delta^c(X,Y;K,\Sigma)-KX.
\]
Therefore
\[
T_{n,h}\stackrel[d]{Q_{n,h}}{\to}K(X-h)+\delta^c(0,Y+\Psi X;K,\Sigma)=\delta^c(X,Y;K,\Sigma)-Kh,
\]
as claimed. \qed

\paragraph{Proof of Corollary \ref{corr: Attainability}}

For $h\in\mathbb{R}^p$, let
\[
T_{n,h}=\sqrt{n}\left(\delta_n^c-\kappa(\theta_{n,h})\right),\quad
T_h=\delta^c(X,Y;K,\Sigma)-Kh,
\]
and define $f(t,w;\beta)=\ell^*(t)\exp(\beta'w)$.
By Proposition \ref{prop: Convergence - Estimation Error} and the definition of $W_{M,n,h}$, under $Q_{n,h}$,
\[
(T_{n,h},W_{M,n,h})\stackrel[d]{Q_{n,h}}{\to} (T_h,W_{M,h}).
\]
By Lemma \ref{lem: fixed-h-liminf}, for each $h$,
\begin{equation}\label{eq: fixed h risk}
\liminf_{n\to\infty}\inf_{\beta}E_{Q_{n,h}}\left[f(T_{n,h},W_{M,n,h};\beta)\right]
\ge
\inf_{\beta}E_{Q_h}\left[f(T_h,W_{M,h};\beta)\right].
\end{equation}

Next, by equivariance of $\delta^c$ and the Gaussian shift structure of the limit experiment,
\[
E_{Q_h}\left[f(T_h,W_{M,h};\beta)\right]
=
E_{Q_0}\left[\ell^*\left(\delta^c(X,Y;K,\Sigma)\right)\exp\left(\beta'W_{M,0}\right)\right]
\]
for every $h$ and $\beta$. Hence $\beta^{*,c}$ from Assumption \ref{assu: UI} minimizes the right-hand side for every $h$.
Therefore, for each $h$,
\[
\limsup_{n\to\infty}\inf_{\beta}E_{Q_{n,h}}\left[f(T_{n,h},W_{M,n,h};\beta)\right]
\le
\lim_{n\to\infty}E_{Q_{n,h}}\left[f(T_{n,h},W_{M,n,h};\beta^{*,c})\right]
\]
\begin{equation}\label{eq: beta^c risk}
=
E_{Q_h}\left[f(T_h,W_{M,h};\beta^{*,c})\right]
=
\inf_{\beta}E_{Q_h}\left[f(T_h,W_{M,h};\beta)\right],
\end{equation}
where the equality in the middle follows from Assumption \ref{assu: UI}. Combining \eqref{eq: fixed h risk} and \eqref{eq: beta^c risk}, for each $h$,
\[
\lim_{n\to\infty}\inf_{\beta}E_{Q_{n,h}}\left[f(T_{n,h},W_{M,n,h};\beta)\right]
=
\inf_{\beta}E_{Q_h}\left[f(T_h,W_{M,h};\beta)\right].
\]
Since $I$ is finite, taking $\sup_{h\in I}$ preserves convergence, and applying $\lambda\log(\cdot)$ gives
\[
\lim_{n\to\infty}\sup_{h\in I}\inf_\beta\lambda\cdot\log\left(E_{Q_{n,h}}\left[\ell^*(T_{n,h})\exp\left(\beta'W_{M,n,h}\right)\right]\right)
=
\sup_{h\in I}\inf_\beta\lambda\cdot\log\left(E_{Q_h}\left[\ell^*(T_h)\exp\left(\beta'W_{M,h}\right)\right]\right).
\]
This is the first claim.

For the asymptotic optimality claim, choose $\delta^c$ optimal in the limit experiment. The convergence argument above gives that the asymptotic risk is
\[
\sup_{h\in\mathbb{R}^p}\inf_\beta\lambda\cdot\log\left(E_{Q_h}\left[\ell^*\left(\delta^c(X,Y)-Kh\right)\exp\left(\beta'W_{M,h}\right)\right]\right).
\]
By Theorem \ref{thm: Constrained LAM}, however, this is the best attainable risk, proving the claim. \qed

\end{document}